\documentclass[12pt]{article}
\usepackage{fullpage}
\usepackage{times}

\usepackage[bookmarks=false]{hyperref}

\usepackage{doi}
\let\citeyear\cite

\usepackage{amsthm}

\theoremstyle{plain}
\newtheorem{theorem}{Theorem}
\newtheorem{lemma}[theorem]{Lemma}

\newtheorem{corollary}[theorem]{Corollary}

\theoremstyle{definition}
\newtheorem{definition}{Definition}
\theoremstyle{remark}
\newtheorem{remark}{Remark}
\newtheorem{example}{Example}

\usepackage{amsmath}
\usepackage{array}
\usepackage{cite}
\usepackage{cases}
\usepackage{paralist}
\usepackage{xcolor}
\usepackage{graphicx}
\DeclareGraphicsExtensions{.pdf,.eps}

\newcommand*{\fixstatement}[1]{}
\newcommand*{\printproofs}{}

\newcounter{sublemmacounter}%
\renewcommand{\thesublemmacounter}{\arabic{sublemmacounter}}%
\newenvironment{sublemma}[1][Part]
{\par\bigskip\noindent%
\refstepcounter{sublemmacounter}%
\textsc{#1 \thesublemmacounter:} }
{\par\smallskip}

\newenvironment{subproof}[1][Subproof]
{\par\smallskip\noindent\textsc{#1:} }
{\hfill$\boxdot$\par\smallskip}

\usepackage{hyperref}

\usepackage{prettyref}
\newrefformat{sec}{Sect.\,\ref{#1}}
\newrefformat{def}{Def.\,\ref{#1}}
\newrefformat{thm}{Theorem\,\ref{#1}}
\newrefformat{prop}{Proposition\,\ref{#1}}
\newrefformat{lem}{Lemma\,\ref{#1}}
\newrefformat{cor}{Corollary\,\ref{#1}}
\newrefformat{ex}{Example\,\ref{#1}}
\newrefformat{tab}{Table\,\ref{#1}}
\newrefformat{fig}{Fig.\,\ref{#1}}
\newrefformat{rem}{Remark\,\ref{#1}}
\newrefformat{case}{Case\,\ref{#1}}
\newrefformat{foot}{Footnote\,\ref{#1}}
\newrefformat{slem}{Part\,\ref{#1}}
\newcommand{\rref}[2][]{\prettyref{#2}}
\newrefformat{bthm}{Theorem\,\ref{#1}}
\newrefformat{blem}{Lemma\,\ref{#1}}
\newrefformat{app}{Appendix\,\ref{#1}}

\newrefformat{nolem}{\textbf{\textit{\textcolor{red}{conjecture\,\ref{#1}}}}}

\usepackage{xcolor}

\usepackage{amsbsy}

\usepackage[bbsets,Dfprime]{math}
\usepackage[modernsign,substopindex,bbmquant,colormquant,mquantifiertype,mconnectiveformal,bracketinterpret,modifopindex,seqarrow,seqoptional,sidenotecalculus,abbrseqcontext,shortterms,nosigmaterms,novarterms]{logic}
\usepackage[pretest,nocommandblocks]{progreg}
\usepackage[bracketinterpret]{dL}

\renewcommand{\linterpretationsname}{\mathcal{S}}
\newcommand*{\I}[1][\xi]{\dLint[state=#1]}
\newcommand*{\sol}[1][]{x}%

\newcommand{\runcost}[1]{}

\providecommand*{\pdiffgame}[3]{{#1}{\ifthenelse{\equal{#2}{}}{}{{}\&^{\hspace{-3pt}d}#2}\ifthenelse{\equal{#3}{}}{}{\&}#3}}

\newcommand{\bebecomes}{\mathrel{::=}}
\newcommand{\alternative}{~|~}
\renewcommand{\with}{:}

\newcommand*{\passfunction}[1]{#1}%

\newcommand*{\genDG}[3]{f(#1,#2,#3)}
\newcommand{\ivr}{\psi}%

\renewcommand{\der}[2][]{\nabla(#2)}%

\newcommand{\dder}[1][x']{\def\dderargI{{#1}}\dderRelay}
\newcommand{\dderRelay}[2][]{\subst[\D{#2}]{\dderargI}{#1}}

\newcommand{\sder}[1]{D#1}%

\newcommand{\sdder}[1][x']{\def\sdderargI{{#1}}\sdderRelay}
\newcommand{\sdderRelay}[2][]{#1\stimes D #2}

\let\norm\abs
\renewcommand*{\norm}[2][]{\ifthenelse{\equal{#1}{}}{\abs{#2}}{\|#2\|_{#1}}}

\newcommand{\ortho}[1]{#1^\perp}%

\renewcommand{\mquantifier}[4]
           {#1%
           #3{\,}%
           #4}

\definecolor{semblue}{rgb}{0,0,0.7}
\definecolor{vgreen}{rgb}{.1,.5,0}
\definecolor{vred}{rgb}{.7,0,0}
\definecolor{vblue}{rgb}{.1,.15,.62}

\renewcommand{\mquantifiercolor}[1]{\textcolor{vblue}{#1}}

\newcommand*{\arithmetize}[1]{{#1}^\Re}%

\renewcommand{\measurables}[2]{\mathcal{M}_{#2}}%

\newcommand*{\controlD}[2][Y]{\measurables{}{#1}}%
\newcommand*{\controlA}[2][Z]{\measurables{}{#1}}%
\newcommand{\stratD}[1][Z]{\def\startDargI{{#1}}\stratDRelay}
\newcommand{\stratDRelay}[2][Y]{\mathcal{S}_{\startDargI\to#1}}

\newcommand{\stratA}[1][Y]{\def\startAargI{{#1}}\stratARelay}
\newcommand{\stratARelay}[2][Z]{\mathcal{S}_{\startAargI\to#1}}

\newcommand*{\response}[5][]{\sol[#1](#2; #3, \passfunction{#4},\passfunction{#5})}

\usepackage{ifpdf}
\ifpdf
\fi

\usepackage{fancyhdr}
\title{Differential Hybrid Games}
\author{Andr\'e Platzer\thanks{
  Computer Science Department, Carnegie Mellon University, Pittsburgh, USA
  {aplatzer@cs.cmu.edu}
}}
\date{}

\begin{document}
\maketitle
\allowdisplaybreaks
\begin{abstract}
This article introduces \emph{differential hybrid games}, which combine differential games with hybrid games.
In both kinds of games, two players interact with continuous dynamics.
The difference is that hybrid games also provide all the features of hybrid systems and discrete games, but only deterministic differential equations.
Differential games, instead, provide differential equations with continuous-time game input by both players, but not the luxury of hybrid games, such as mode switches and discrete-time or alternating adversarial interaction.
This article augments \emph{differential game logic} with modalities for the combined dynamics of differential hybrid games.
It shows how hybrid games subsume differential games and
introduces \emph{differential game invariants} and \emph{differential game variants} for proving properties of differential games inductively.
\end{abstract}

\section{Introduction}
\newsavebox{\Rval}%
\sbox{\Rval}{$\scriptstyle\mathbb{R}$}%
\irlabel{qear|\usebox{\Rval}}%
\irlabel{orl|$\lor$r}%

\emph{Differential games} \cite{Isaacs:DiffGames,Friedman,Hajek,KrasovskiSubbotin,Petrosjan93,BardiRP99,DBLP:journals/tams/ElliottK74,DBLP:journals/indianam/EvansSouganidis84,DBLP:journals/diffeq/Souganidis85} support adversarial interaction and game play during the continuous dynamics of a differential equation in continuous time.
They allow the two players to control inputs to the differential equation during its continuous evolution by measurable functions of continuous time.
This is to be contrasted with hybrid systems \cite{DBLP:conf/lics/Henzinger96} and hybrid games, where differential equations are deterministic and the only decision is how long to evolve.
Differential games are useful, e.g., for studying pursuit-evasion in aircraft if both players can react continuously to each other.
They are a good match for tight-loop, analog, or rapid adversarial interaction.

\emph{Hybrid games} \cite{DBLP:conf/hybrid/NerodeRY96,DBLP:conf/concur/HenzingerHM99,DBLP:journals/IEEE/TomlinLS00,DBLP:journals/jopttapp/DharmattiDR06,DBLP:journals/corr/abs-0911-4833,DBLP:journals/tcs/VladimerouPVD11,DBLP:journals/tocl/Platzer15} are games of two players on a hybrid system's discrete and continuous dynamics where the players have control over discrete-time choices during the evolution of the system, but the continuous dynamics stays deterministic and its duration is the only choice in the game.
Hybrid games can model discrete aspects like decision delays, discontinuous state change, or games with different controls and different dynamics in different modes of the system.
They are a good match for sporadic or discrete-time adversarial interaction with discrete sensors or reaction delays for structurally complex systems.

The primary purpose of this article is to show that both game principles are not in conflict but can be integrated seamlessly to complement each other.
This article introduces \emph{differential hybrid games} that combine the aspects of differential games with those of hybrid games resulting in a model where discrete, continuous, and adversarial dynamics mix freely.
This makes it possible to model games that combine continuous-time interactions (e.g.\ auto-evasion curves for aircraft) with discrete-time interactions (e.g.\ whether to ask an intruder pilot to synchronize on collision avoidance or whether to follow a nonstandard manual flight maneuver).
Differential hybrid games also possess the advantages of hybrid systems so that structurally more complex cases with different parts and different subsystems with different dynamics can be modeled.

The key insight behind \emph{hybrid systems} is that it helps to understand each aspect of a system separately on its natural level \cite{DBLP:conf/lics/Platzer12a}. Discrete dynamics are a good fit for some aspects. Continuous dynamics are more natural for others.
Differential hybrid games enable the same flexibility for games rather than for systems, so that each adversarial aspect in a cyber-physical system can be understood on its most natural level.
Which level that is depends on modeling/analysis tradeoffs.
Differential hybrid games provide a unifying framework in which both game aspects coexist and combine freely to enable such tradeoffs.
Their differential games are needed to describe quick continuous-time control interaction, while only their hybrid game aspects provide discrete-time adversarial choices, discontinuous state change, and subsystem structuring mechanisms.
This article studies, e.g., airships with continuous adversarial change of local turbulence and sporadic discrete adversarial change of wind fields with static obstacles.
The system changes radically whenever the wind field or the relevant static obstacles change, which only happens sporadically, because airships do not move into entirely new wind conditions or near the next mountain so often.
Local turbulence changes quickly, however, and needs appropriate reactions all the time.

This article introduces a generalization of \emph{differential game logic} \dGL \cite{DBLP:journals/tocl/Platzer15} to \emph{differential} hybrid games, extending differential game logic for ordinary \emph{hybrid games} \cite{DBLP:journals/tocl/Platzer15} by adding differential games.
Since this extension yields a compositional logic and a compositional proof technique, the primary attention in this article is on how differential games combine seamlessly with hybrid games and how properties of differential games can be proved soundly. Proof techniques for the resulting differential hybrid games then follow from logical compositionality principles.

In addition to presenting the first logic and modeling language for differential hybrid games, this article presents inductive proof rules for differential games to obtain the first sound and compositional proof calculus for differential hybrid games.
\emph{Differential game invariants} and their companions (\emph{differential game variants} and \emph{differential game refinements}) give a logical approach for differential games, complementing geometric viability theory and other approaches based on numerical integration of partial differential equations (PDEs).
The advantage is that differential game (in)variants provide simple and sound witnesses for the existence of winning strategies for differential games, even in unbounded time, and without having to build a formally verified numerical solver for PDEs with formally verified error bounds to obtain sound formal verification results, which would be quite a formidable challenge.

Soundness is a substantial matter in differential games due to their surprising subtleties.
It took several decades to correctly relate Isaacs' PDEs to the differential games they were intended for \cite{Isaacs:DiffGames,BardiRP99}. %
After a long period of gradual progress, differential games are now handled by numerically solving the PDEs they induce or by corresponding geometric equivalents from viability theory formulations for such PDEs \cite{Cardaliaguet2007}.
Soundness issues with a number of such approaches for differential games were reported \cite{DBLP:journals/tac/MitchellBT05}.\footnote{%
The results presented here are of independent interest, because they provide a fix for an incorrect cyclic quantifier dependency in the correctness proof in said paper \cite{DBLP:journals/tac/MitchellBT05} that was confirmed by the authors.
}

This raises the challenge how to prove properties of differential games with the correctness demands of a proof system.
This article advocates for dedicated proof rules for differential games alongside proof rules for hybrid games.
Logic is good at then combining those sound proof rules soundly with each other in a modular way.

The article concludes with a theoretical insight. Hybrid systems have been shown to be equivalently reducible proof-theoretically to differential equations \cite{DBLP:journals/jar/Platzer08} and even to equivalently reduce to discrete systems \cite{DBLP:conf/lics/Platzer12b}.
This trend reverses for hybrid games, which do \emph{not} reduce to differential games, but subsume them.

\paragraph{Contributions}

The primary contributions are a compositional programming language for differential hybrid games that combine discrete, continuous, and adversarial dynamics freely, along with a proof calculus and expressibility results.
The most important novel feature of the proof calculus are sound induction principles for differential games.
The most interesting technical contribution is their soundness proof.
Superdifferentials for a conceptual simplification are another interesting aspect of this article.

While the results are elegant and all background for the proofs is given, these proofs draw from many areas, including logic, proof theory, Carath\'eodory solutions, viscosity solutions of partial differential equations, real algebraic geometry, and real analysis.
Byproducts of the soundness proof yield results of independent interest.
All new proofs, which are the ones for results without citations, are included inline.

\section{Preliminaries} \label{sec:preliminaries}

This section briefly recalls basic notions that will be used throughout this article.
The article mostly considers Euclidean vector spaces with the Euclidean scalar product of vectors $v,w$  denoted by \(v\stimes w\).
The Euclidean norm of a vector $v$ is denoted by \(\norm{v}\mdefeq\sqrt{v\stimes v}\).

A set $Z$ is \emph{compact} (called quasi-compact by Bourbaki) iff every open cover has a finite subcover.
In metric spaces such as $\reals^k$, a set $Z$ is compact iff it is \emph{sequentially compact}, i.e.\ every sequence in $Z$ has a convergent subsequence with limit in $Z$.
In Euclidean spaces, a set is compact iff it is closed and bounded (Heine-Borel theorem).

\begin{remark}[Preimage] \label{rem:preimage}
  The preimage \m{\ipreimage{f}(A) \,{=}\, \{x \in X \with f(x){\in} A\}} of a set $A\subseteq Y$ under function \(f:X\to Y\) satisfies the usual properties:
  \begin{enumerate}
  \item \(A\subseteq B\) implies \(\ipreimage{f}(A) \subseteq \ipreimage{f}(B)\)
  \label{case:preimage-subseteq}
  \item \(\ipreimage{f}(\scomplement{A}) = \scomplement{(\ipreimage{f}(A))}\) for the complement $\scomplement{A}$ of $A$
  \item \(\ipreimage{f}(\capfold_{i\in I} A_i) = \capfold_{i\in I} \ipreimage{f}(A_i)\) for any index family $I$
  \item \(\ipreimage{(f\compose g)}(A) = \ipreimage{g}(\ipreimage{f}(A))\))
  where \(f\compose g\) is the composition mapping $z$ to \(f(g(z))\)
  \label{case:preimage-compose}
  \end{enumerate}
\end{remark}
A function \(f:X\to Y\) between measurable spaces is \emph{measurable} iff the preimage of every measurable subset in $Y$ is measurable in $X$.
By \rref{case:preimage-compose}, the composition $f \compose g$ is $\mathcal{M}$-measurable if $f$ is \emph{Borel} measurable and $g$ is $\mathcal{M}$-measurable.
It is important for this composition that $f$ is Borel measurable, otherwise the measure space changes.

A function \(f:X\to\reals^k\) on a normed vector space $X$ is \emph{$\lambda$-H\"older continuous} iff there is an $L\in\reals$ such that \(\norm{f(x)-f(y)}\leq L\norm{x-y}^\lambda\) for all $x,y$.
Hence, $f$ is 0-H\"older continuous iff it is bounded.
Functions that are 1-H\"older continuous are called \emph{Lipschitz-continuous}.
A function \(f:X\times Y\to\reals^k\) is \emph{uniformly Lipschitz in $x$} iff there is an $L\in\reals$ such that \(\abs{f(x,y)-f(a,y)}\leq L\abs{x-a}\) for all $x,a\in X, y\in Y$.
A function \(f:X\to\reals^k\) is \emph{uniformly continuous} iff, for all $\varepsilon>0$, there is a $\delta>0$ such that \(\norm{f(x)-f(y)}<\varepsilon\) for all \(\norm{x-y}<\delta\). 
A function \(f:I\to\reals^k\) on an interval $I\subseteq\reals$ is \emph{absolutely continuous} iff, for all $\varepsilon>0$, there is a $\delta>0$ such that \(\sum_i\norm{f(x_i)-f(y_i)}<\varepsilon\) for all finite sequences of pairwise disjoint sub-intervals \(\interval{(x_i,y_i)}\subseteq I\) with \(\sum_i\abs{x_i-y_i}<\delta\).
If $f$ is continuously differentiable on a compact set,
then $f$ is differentiable with bounded derivatives,
in which case $f$ is Lipschitz-continuous,
which implies $f$ is absolutely continuous,
implying that $f$ is uniformly continuous,
in which case $f$ is continuous,
making $f$ Borel measurable \cite{Walter:Ana2}.
Continuous functions on compact sets are bounded and uniformly continuous.
A \emph{semialgebraic function} is a function between \emph{semialgebraic sets} (i.e.\ definable by finite unions and intersections of polynomial equations and inequalities) whose graph is semialgebraic.

A sequence of functions \(f_n:X\to\reals^k\) \emph{converges uniformly} to \(f:X\to\reals^k\) for \(n\to\infty\) iff $f_n$ converges to $f$ in supremum norm \(\norm[\infty]{f_n-f}\to0\) for \(n\to\infty\), which is equivalent to: for all $\varepsilon>0$ there is an $n_0$ such that \(\norm{f_n(x)-f(x)}<\varepsilon\) for all $n\geq n_0$ and all $x\in X$.
\section{Differential Game Logic} \label{sec:dGL}

This section introduces the \emph{differential game logic \dGL of differential hybrid games}, which adds differential games to differential game logic of (non-differential) hybrid games from previous work \cite{DBLP:journals/tocl/Platzer15}.
The difference between hybrid games and differential hybrid games is that only the latter allow differential games with player input, while the former allow only differential equations instead.
The respective differential game logics are built in the same way around the respective game models.

Differential hybrid games are games of two players, Angel and Demon.
Differential game logic uses modalities, where \m{\dbox{\alpha}{\phi}} refers to the existence of winning strategies for Demon for the objective specified by formula $\phi$ in differential hybrid game $\alpha$ and \m{\ddiamond{\alpha}{\phi}} refers to the existence of winning strategies for Angel for objective $\phi$ in differential hybrid game $\alpha$.
So \m{\dbox{\alpha}{\phi}} and \m{\ddiamond{\alpha}{\lnot\phi}} refer to complementary winning conditions ($\phi$ for Demon $\lnot\phi$ for Angel) in the same differential hybrid game $\alpha$.
Indeed, the two formulas \m{\dbox{\alpha}{\phi}} and \m{\ddiamond{\alpha}{\lnot\phi}} cannot both be true in the same state (\rref{thm:dGL-determined}).

\subsection{Syntax}

The terms $\theta$ of \dGL are polynomial terms (more general ones are possible but not necessarily decidable).
In applications, it is convenient to use \(\min,\max\) terms as well, which are definable as semialgebraic functions \cite{tarski_decisionalgebra51}.
Differential game logic formulas and differential hybrid games are defined by simultaneous induction.
Similar simultaneous inductions are used throughout the definitions and proofs for \dGL.

\begin{definition}[Differential hybrid games] \label{def:dGL-DHG}
The \emph{differential hybrid games of differential game logic {\dGL}} are defined by the following grammar\footnote{%
The $\pdual{}$ in differential game \(\pdiffgame{\D{x}=\genDG{x}{y}{z}}{y\in Y}{z\in Z}\) is a mnemonic reminder that Demon controls \emph{d}ual input $y$ and Angel controls $z$.
The order of notation further reminds that, at any point in time, Demon chooses an action $y\in Y$ before Angel chooses a $z\in Z$ at that time (in a sense made precise in \rref{sec:dGL-semantics}).
} ($\alpha,\beta$ are differential hybrid games, $x$ is a variable, $\theta$ a term, $\ivr$ a \dGL formula, \(y\in Y\) and \(z\in Z\) are formulas in free variable $y$ or $z$, respectively, and $f(x,y,z)$ is a term in the free variables $x,y,z$):
\[
  \alpha,\beta ~\bebecomes~
  \pdiffgame{\D{x}=\genDG{x}{y}{z}}{y\in Y}{z\in Z}\\
  \alternative
  \pupdate{\pumod{x}{\theta}}
  \alternative
  \ptest{\ivr}
  \alternative
  \alpha\cup\beta
  \alternative
  \alpha;\beta
  \alternative
  \prepeat{\alpha}
  \alternative
  \pdual{\alpha}
\]
\end{definition}
\begin{definition}[\dGL formulas] \label{def:dGL-formula}
The \emph{formulas of differential game logic {\dGL}} are defined by the following grammar ($\phi,\psi$ are \dGL formulas, 
$\theta_i$ are (polynomial) terms, $x$ is a variable, and $\alpha$ is a differential hybrid game):
  \[
  \phi,\psi ~\bebecomes~
  \theta_1\geq\theta_2 \alternative
  \lnot \phi \alternative
  \phi \land \psi \alternative
  \lexists{x}{\phi} \alternative 
  \ddiamond{\alpha}{\phi}
  \alternative \dbox{\alpha}{\phi}
  \]
\end{definition}
Other operators $>,=,\leq,<,\lor,\limply,\lbisubjunct,\forall{x}{}$ can be defined as usual, e.g., \m{\lforall{x}{\phi} \mequiv \lnot\lexists{x}{\lnot\phi}}.
The modal formula \m{\ddiamond{\alpha}{\phi}} expresses that Angel has a winning strategy\footnote{%
The names are arbitrary but the mnemonic is that, just like the diamond operator $\ddiamond{\cdot}{}$, Angel has wings.}
to achieve $\phi$ in differential hybrid game $\alpha$, i.e.\ Angel has a strategy to reach any of the states satisfying \dGL formula $\phi$ when playing differential hybrid game $\alpha$, no matter what Demon does.
The modal formula \m{\dbox{\alpha}{\phi}} expresses that Demon has a winning strategy to achieve $\phi$ in differential hybrid game $\alpha$, i.e.\ a strategy to reach any of the states satisfying $\phi$, no matter what Angel does.

As usual, \m{\pupdate{\pumod{x}{\theta}}} is an \emph{assignment} and $\ptest{\ivr}$ the \emph{test game} or \text{challenge} that Angel only passes if formula $\ivr$ holds true in the current state.
Otherwise she loses immediately, because she failed a test.
Further, \(\pchoice{\alpha}{\beta}\) is a \emph{game of choice} where Angel gets to choose to play $\alpha$ or to play $\beta$ whenever \(\pchoice{\alpha}{\beta}\) is played.
The \emph{sequential game} \(\alpha;\beta\) plays $\alpha$ followed by $\beta$ (unless a player lost a challenge during $\alpha$).
The \emph{repeated game} $\prepeat{\alpha}$ plays $\alpha$ repeatedly and permits Angel to decide after each play of $\alpha$ whether she wants to play another iteration (unless a player lost a challenge).
The \emph{dual game} $\pdual{\alpha}$ is the same as game $\alpha$ except that all choices that Angel has in $\alpha$ are resolved by Demon in $\pdual{\alpha}$ and all choices that Demon has in $\alpha$ are resolved by Angel in $\pdual{\alpha}$, similar to the effect of flipping a chessboard around by $180^\circ$ so that both players change sides.
During a differential hybrid game, players can lose prematurely by violating the rules of the game, expressed in the tests. The winning condition is specified in the postcondition $\phi$.

The important addition compared to prior work \cite{DBLP:journals/tocl/Platzer15} is the inclusion of differential games in the syntax for hybrid games.
Predicate symbols have been removed, because they are of no immediate concern for the core focus here: adding differential games to hybrid games.
All occurrences of $y,z$ in \(\pdiffgame{\D{x}=\genDG{x}{y}{z}}{y\in Y}{z\in Z}\) are bound.
Finally, $x,y,z$ can be vectorial if \(\genDG{x}{y}{z}\) is a vectorial term of the same dimension as $x$.
During a differential game \(\pdiffgame{\D{x}=\genDG{x}{y}{z}}{y\in Y}{z\in Z}\), the state follows the differential equation \(\D{x}=\genDG{x}{y}{z}\), yet Demon provides a measurable input for $y$ satisfying $y\in Y$ always, and Angel, knowing Demon's current input, provides a measurable input for $z$ satisfying $z\in Z$, while Angel controls the duration (\rref{sec:dGL-semantics}).

Observe the use of suggestive notation that is adopted in the interest of traceability with mathematical practice throughout this article: \(\lforall{y{\in}Y}{\phi}\) stands for \(\lforall{y}{(y\in Y \limply \phi)}\) and \(\lexists{y{\in}Y}{\phi}\) for \(\lexists{y}{(y\in Y \land \phi)}\), in which $y\in Y$ is understood as convenient notation for a logical formula of one free variable (vector) $y$.

A \emph{nondeterministic assignment} \(\prandom{c}\) assigns any real value to $x$ by Angel's choice, so \(\dbox{\prandom{c}}{\phi} \mequiv \lforall{c}{\phi}\) and \(\ddiamond{\prandom{c}}{\phi} \mequiv \lexists{c}{\phi}\).
Nondeterministic assignments are definable, e.g., as the differential game \(\pdiffgame{\D{c}=z}{}{z\in B}\), where $z\in B$ is $-1{\leq}z{\leq}1$,
or with differential equations \(\prandom{c} \mequiv \pevolve{\D{c}=1};\pevolve{\D{c}=-1}\) since durations are unobservable without clocks.

\begin{example}[Zeppelin] \label{ex:Zeppelin}
Coping with the intricacies of wind is an omnipresent challenge for aircraft, and particularly pronounced for airships where the wind may sometimes be stronger than their own propulsion.
It is not their propulsion that keeps airships in the air, but they are big and lighter than air to generate buoyancy.
At a fixed height, consider a Zeppelin-class airship with position $x\in\reals^2$ that can fly in all directions by turning its propeller into that direction.
The propeller itself can generate up to a velocity $p$.
The Zeppelin is floating in the sky within a homogeneous wind velocity field $v\in\reals^2$, but is also subject to local turbulence changing quickly in unpredictable directions of magnitude ${\leq}r$.
The Zeppelin is trying to prevent a collision with an obstacle of radius ${\leq}c$ at position $o\in\reals^2$.
In order to fly safely, the Zeppelin needs to find a way of controlling its propeller so that it remains collision-free despite the turbulence.
Sporadically at discrete time points, the homogeneous wind field $v$ may change or another obstacle $o$ possibly with a different radius $c$ may appear in the Zeppelin's horizon.
To simplify the matter, the obstacles are assumed to be reasonably separated so that the Zeppelin ever only has to worry about one obstacle at a time (unlike \rref{fig:Zeppelin}).

\begin{figure}[tb]
  \centering
  \includegraphics[height=5.8cm]{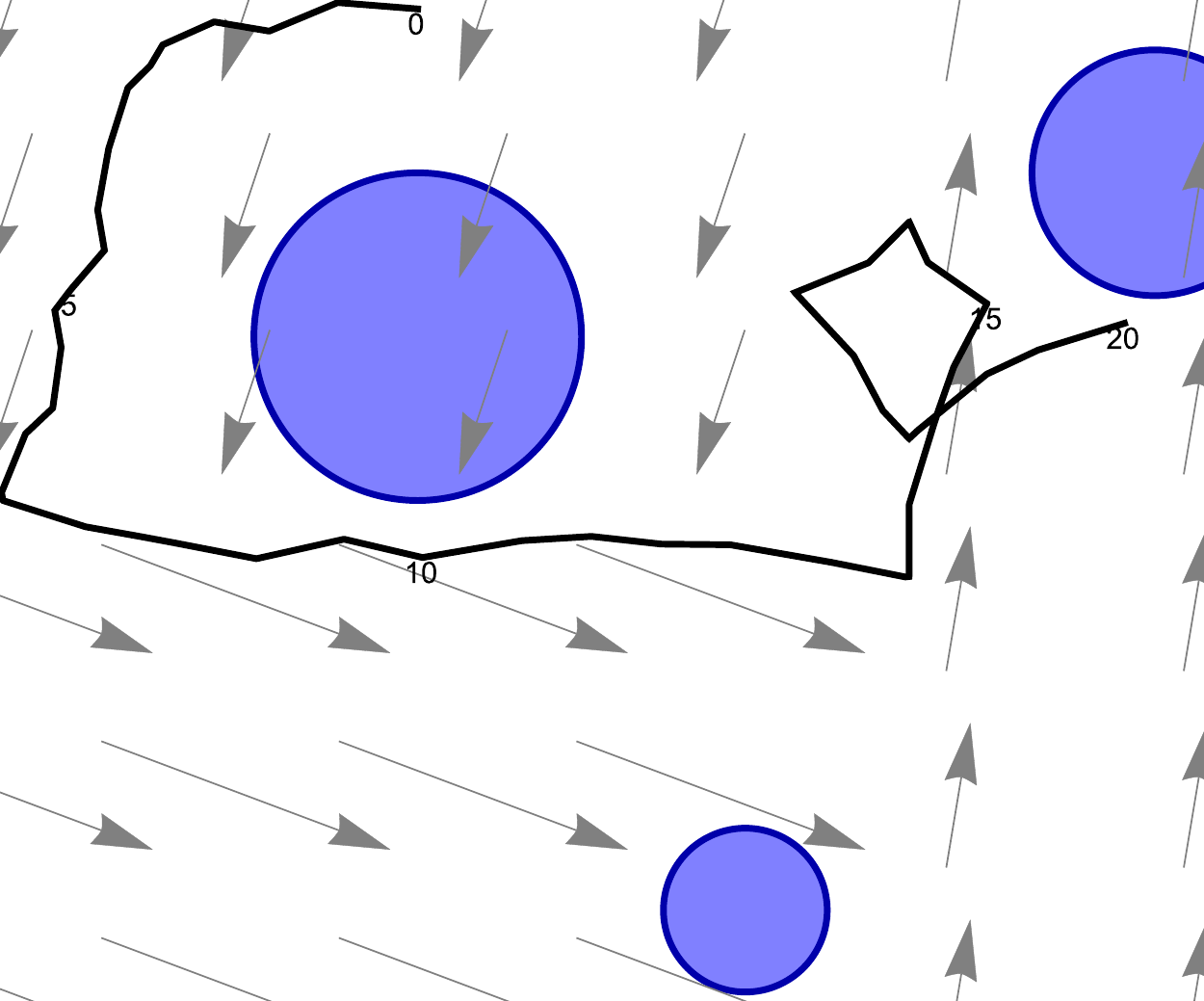}
  \caption{Challenging Zeppelin obstacle parcours with wind fields and one possible response trajectory}
  \label{fig:Zeppelin}
\end{figure}%

If the turbulence is stronger than the propeller (\(r>p\)), the Zeppelin is unmaneuverable and hopelessly at the mercy of the waves in the air.
If the propeller is stronger than the wind field and the turbulence combined (\(p>\norm{v}+r\)), the Zeppelin is essentially able to overcome all wind by sheer force.
In between these two extremes, however, when \(r<p<\norm{v}+r\), the Zeppelin has the particularly interesting challenge of having to maneuver in a clever way to ensure the combined wind field and possible local turbulences cannot lead to a collision that its propeller can no longer prevent.

The following \dGL formula expresses that there is a winning strategy to fly the Zeppelin safely around the obstacle (\(\norm{x-o}^2\geq c^2\)) if it was initially safe and each obstacle is recognized at sufficient distance according to a condition $C$ that is yet to be identified:
\begin{equation}
\begin{aligned}
c>0 &\land \norm{x-o}^2\geq c^2 \limply\\
&\big[(
\prandom{v};\prandom{o};\prandom{c}; \ptest{C};\\
&\phantom{\big[\big(}
\pdiffgame{\D{x}=v+py+rz}{y\in B}{z\in B}\\
&\prepeat{)}\big]\, \norm{x-o}^2\geq c^2
\end{aligned}
\label{eq:Zeppelin}
\end{equation}
To give the Zeppelin a chance, assume some choices of $p$ and $r$ for which \(p>r\geq0\) so that the propeller is not weaker than the turbulence.
For example \(p = 3/4, r = 1/2\) as in \rref{fig:Zeppelin}, which are weaker than all its wind fields, though.
Using vectorial notation, let \(y\in B\) be \(y_1^2+y_2^2\leq1\) and similarly let \(z\in B\) be \(z_1^2+z_2^2\leq1\) to describe the unit disc, among which direction vectors are chosen by the two players during the differential game (line 3).
Unit vectors correspond to full speed ahead, while vectors of smaller norm will lead to less power.
According to the semantics of differential games (\rref{sec:dGL-semantics}), the propeller $y\in B$ will have to act before it knows about the local turbulence $z\in B$, because it is hard to predict the chaotic changes of turbulence.

In addition to this rapid interaction, where the propeller tries to overcome the local turbulence to prevent collisions, the differential hybrid game in \rref{eq:Zeppelin} includes a repetition (operator $\prepeat{}$ in line 4), also under the opponent's control, that allows Angel to repeat lines 2--4 any number of times.
During each repetition, the differential hybrid game in \rref{eq:Zeppelin} allows arbitrary wind field changes (by a nondeterministic assignment \(\prandom{v}\)), and allows the next relevant obstacle $o$ to appear arbitrarily with a new radius $c$.
After those arbitrary changes, Angel needs to pass the subsequent test $\ptest{C}$, though, so that she can only update $v,o,c$ according to the formula $C$, which will be identified in \rref{sec:diffgame-axioms} to prevent impossibly late notice of obstacles.
Demon's opponent in \dGL formula \rref{eq:Zeppelin} can, thus, sporadically switch to a new obstacle and possibly also a new wind field as long as there is enough space in between to satisfy $C$.
Demon's job in the differential game is to use the propeller to avoid the current obstacle despite the additional turbulence under Angel's control.
If the obstacles are too close together and the wind fields change too radically, the Zeppelin navigation problem is exceedingly tricky (\rref{fig:Zeppelin}).
\end{example}

\subsection{Differential Games}

The semantics of differential game constructs in differential hybrid games is based on nonanticipative strategies for differential games \cite{DBLP:journals/tams/ElliottK74,DBLP:journals/indianam/EvansSouganidis84,BardiRP99}. %

\begin{definition}[Differential game] \label{def:diffgame}
  Let \(Y\subset\reals^k, Z\subset\reals^l\) be compact sets of controls for the respective players.
  Let function \(f:[\eta,T]\times\reals^n\times Y\times Z\to\reals^n\) be bounded, uniformly continuous, and in $x$ uniformly Lipschitz.
  For time horizon \(T\geq0\), initial time \(\eta<T\) and initial state \(\xi\in\reals^n\),
  a \emph{differential game} has the form
  \begin{equation}
  \begin{cases}
  \D{x}(s) = f(s,x(s),y(s),z(s)) & \eta\leq s\leq T\\
  x(\eta) = \xi
  \end{cases}
  \label{eq:diffgame}
  \end{equation}
  where the \emph{controls} \(y:[\eta,T]\to Y\) and \(z:[\eta,T]\to Z\) are measurable functions for the respective players for $Y$ and for $Z$.
  The set of (measurable) controls are denoted $\controlD{\eta}$ and $\controlA{\eta}$, respectively.
   The \emph{terminal payoff}, i.e.\ payoff at time horizon $T$, is defined by a bounded and Lipschitz function \(g:\reals^n\to\reals\).
\end{definition}

A \emph{Carath\'eodory solution} of a differential equation is an absolutely continuous function satisfying the differential equation \emph{a.e.} (\emph{almost everywhere}, which means except on a subset of a set of measure 0).
By a classical result, the behavior of a differential game is uniquely determined for each pair of controls:

\begin{lemma}[Response] \label{lem:response}
For each controls \(y\in\controlD{\eta}, z\in\controlA{\eta}\) and initial data $\eta,\xi$, the differential equation \rref{eq:diffgame} has a unique Carath\'edory solution \(x:[\eta,T]\to\reals^n\), called \emph{response} and denoted by \(\response[f]{s}{\xi}{y}{z} \mdefeq x(s)\) as a function of time $s$ with parameters $\xi,y,z$.
Finally, \(\response[f]{s}{\xi}{y}{z}=\response[f]{s}{\xi}{\hat{y}}{\hat{z}}\) if \(y=\hat{y}\) a.e.\ and \(z=\hat{z}\) a.e.
\end{lemma}
\begin{proofatend}
\(f(s,x,y(s),z(s))\) is continuous in $x$, because $f$ is (even uniformly) continuous, and measurable in $s$, because it is a composition of a continuous so Borel-measurable function $f$ with measurable functions (\rref{sec:preliminaries}).
Let \(l(s)\) denote the maximum of the bound of $f$ and its Lipschitz constant $L$.
Then $l$ is measurable and integrable on $[\eta,T]$ (it is constant) and satisfies
\(\abs{f(s,0,y(s),z(s)}\leq l(s)\) and the (generalized) Lipschitz condition for all $s,x,a$:
\[
\abs{f(s,x,y(s),z(s)) - f(s,a,y(s),z(s))} \leq l(s)\abs{x-a}
\]
Thus, Carath\'eodory's existence and uniqueness theorem \cite[\S10.XX]{Walter:DGL} shows the existence of a unique solution that can be continued to the boundary of the domain, hence is global since $f$ is bounded.
Finally, the fundamental theorem of calculus for Lebesgue-integrals \cite[Thm.\,9.23]{Walter:Ana2} implies that changing $y$ and $z$ on a set of measure zero does not change the response \(\response[f]{s}{\xi}{y}{z}\).
\end{proofatend}
Of course, the players do not know their opponent's control. Yet, for each possible control pair \(y\in\controlD{\eta}\) and \(z\in\controlA{\eta}\), the response of the differential game is unique by \rref{lem:response}, even if it is still hard to predict computationally.

The (terminal\footnote{%
   Differential games with a running payoff $h$ can be converted to terminal payoff \m{\tilde{g}(x,r)=g(x)+r} when adding a differential equation \m{\D{r}(s)=h(s,x(s),y(s),z(s))} that accumulates the running payoff $h(s,x,y,z)$.}) \emph{payoff} for controls \(y\in\controlD{\eta}\) and \(z\in\controlA{\eta}\) of \rref{eq:diffgame} at the terminal time $T$ is the value of $g$ at the final state at time $T$, i.e.:
\begin{equation}
g(x(T)) = g(\response[f]{T}{\xi}{y}{z})
\label{eq:payoff}
\end{equation}

If both players were to commit to a control before the differential game starts, then they play \(y\in\controlD{\eta}\) and \(z\in\controlA{\eta}\) as \emph{open-loop strategies} and, by \rref{lem:response}, the only question would be what information they have before choosing their respective controls, e.g., in which order they choose \(y\in\controlD{\eta}\) and \(z\in\controlA{\eta}\).

As soon as the players get to observe the state of the system and react to it, though, the situation gets more interesting but also more difficult, because the state at a time $s$ depends on the controls that both players chose until time $s$, so their reactions depend on previous actions by both players.
That still leaves the question what information the players have when they choose their actions.
A nonanticipative strategy (for $Z$) is a function that maps the opponent's control functions to the player's control functions. The strategy gets the opponent's full control signal $y\in\controlD{\eta}$ as input, but, for fairness reasons, its resulting control value in $Z$ at any time $s$ is only allowed to depend on the values that $y$ had until time $s$ (no dependency on the future).
A nonanticipative strategy for $Z$ does, however, give the player for $Z$ a slight edge of having access also to the opponent's action at the present time $s$.
A nonanticipative strategy produces equivalent controls at time $s$ for two controls that agree up to time $s$.
Equality almost everywhere implies that the game response is unchanged (\rref{lem:response}), so that the appropriate notion of equivalent controls is equality a.e.

Let $\stratA{\eta}$ be the set of (\emph{causal} or) \emph{nonanticipative strategies} for $Z$, i.e.\ the set of functions \(\beta:\controlD{\eta}\to\controlA{\eta}\) such that for all times \(\eta\leq s\leq T\) and all controls $y,\hat{y}\in\controlD{\eta}$:
\begin{align*}
~~&\text{if}~y=\hat{y}~\text{a.e. on}~[\eta,s] 
\\
&\text{then}~\beta(y)=\beta(\hat{y})~\text{a.e. on}~[\eta,s]
~~(\text{i.e.}~\beta(y)(\tau)=\beta(\hat{y})(\tau) ~\text{for a.e.}~\eta\leq\tau\leq s)
\intertext{%
That is, nonanticipative strategies for $Z$ give the player for $Z$ the current state, history (which is irrelevant because the games are Markovian), and the opponent's current action.
The reaction $\beta(y)$ of a nonanticipative strategy to $y$ cannot, however, depend on the opponent's future input beyond the current time.
Unlike ill-defined approaches with state-feedback strategies etc., time-dependent controls ensure that the response exists and is unique \cite[\S2.2]{Hajek}.
Dually, the set of \emph{nonanticipative strategies} for $Y$ is $\stratD{\eta}$, i.e.\ the set of \(\alpha:\controlA{\eta}\to\controlD{\eta}\) such that for all \(\eta\leq s\leq T\) and all $z,\hat{z}\in\controlA{\eta}$:
}
&\text{if}~z=\hat{z} ~\text{a.e. on}~[\eta,s]
~%
\text{then}~\alpha(z)=\alpha(\hat{z}) ~\text{a.e. on}~[\eta,s]
\end{align*}

\subsection{Semantics} \label{sec:dGL-semantics}

The semantics for differential game logic with differential hybrid games embeds the semantics of differential games within differential hybrid games while simultaneously extending the meaning of hybrid games seamlessly to differential hybrid games by adding differential game winning regions.
The modular design of \dGL makes this integration of differential games with hybrid games simple by exploiting their compositional semantics.
Since hybrid games have been described before \cite{DBLP:journals/tocl/Platzer15}, the primary focus will be on elaborating the new case of differential games.

A \emph{state} $\iget[state]{\I}$ is a mapping from variables to $\reals$.
Let $\linterpretations{\Sigma}{V}$ be the set of states, which, for $n$ variables, is isomorphic to Euclidean space $\reals^n$.
For a subset \m{X\subseteq\linterpretations{\Sigma}{V}} the complement \m{\linterpretations{\Sigma}{V} \setminus X} is denoted $\scomplement{X}$.
Let
\m{\iget[state]{\imodif[state]{\I}{x}{\kappa}}} denote the state that agrees with state~$\iget[state]{\I}$ except for the interpretation of variable~\m{x}, which is changed to~\m{\kappa \in \reals}.
The value of term $\theta$ in state $\iget[state]{\I}$ is denoted by \m{\ivaluation{\I}{\theta}} and defined as in first-order real arithmetic.
As \dGL formulas are defined by simultaneous induction with differential hybrid games,
the denotational semantics of \dGL formulas will be defined (\rref{def:dGL-semantics}) by simultaneous induction along with the denotational semantics, $\strategyfor[\alpha]{\argholder}$ and $\dstrategyfor[\alpha]{\argholder}$, of differential hybrid games (\rref{def:DHG-semantics}).

Unlike the \dGL quantifiers \m{\lexists{x}{}},\m{\lforall{x}{}}in the logical language of \dGL, the short notation for quantifiers in the mathematical metalanguage is \m{\mexists{\iget[state]{\I}}{}}(for some ${\iget[state]{\I}}$) and \m{\mforall{\iget[state]{\I}}{}}(for all ${\iget[state]{\I}}$).

\begin{definition}[\dGL semantics] \label{def:dGL-semantics}
The \emph{semantics of a \dGL formula} $\phi$ 
is the subset \m{\imodel{\I}{\phi}\subseteq\linterpretations{\Sigma}{V}} of states in which $\phi$ is true.
It is defined inductively as follows:
\begin{enumerate}
\setlength{\itemsep}{-2pt}
\item \(\imodel{\I}{\theta_1\geq\theta_2} = \{\iportray{\I} \in \linterpretations{\Sigma}{V} \with \ivaluation{\I}{\theta_1}\geq\ivaluation{\I}{\theta_2}\}\)
\item \(\imodel{\I}{\lnot\phi} = \scomplement{(\imodel{\I}{\phi})} = \linterpretations{\Sigma}{V}\setminus\imodel{\I}{\phi}\)
\item \(\imodel{\I}{\phi\land\psi} = \imodel{\I}{\phi} \cap \imodel{\I}{\psi}\)
\item
{\def\Im{\imodif[state]{\I}{x}{\kappa}}%
\(\imodel{\I}{\lexists{x}{\phi}}
= \{\iportray{\I} \in \linterpretations{\Sigma}{V} \with \mexists{\kappa\in\reals} {\iget[state]{\Im} \in \imodel{\I}{\phi}}\}\)
= \(\{\iportray{\I} \in \linterpretations{\Sigma}{V} \with \iget[state]{\Im} \in \imodel{\I}{\phi} ~\text{for some}~\kappa\in\reals\}\)
}
\item \(\imodel{\I}{\ddiamond{\alpha}{\phi}} = \strategyfor[\alpha]{\imodel{\I}{\phi}}\)
\item \(\imodel{\I}{\dbox{\alpha}{\phi}} = \dstrategyfor[\alpha]{\imodel{\I}{\phi}}\)
\end{enumerate}
Formula $\phi$ is \emph{valid}, written \m{\entails\phi}, iff \m{\imodel{\I}{\phi}=\linterpretations{\Sigma}{V}}, i.e.\ $\phi$ is true in all states $\iget[state]{\I}$.
\end{definition}

\begin{definition}[Semantics of differential hybrid games] \label{def:DHG-semantics}
The \emph{semantics of a differential hybrid game} $\alpha$ is a function $\strategyfor[\alpha]{\argholder}$, that, for each set of Angel's winning states \(X\subseteq\linterpretations{\Sigma}{V}\), gives the \emph{winning region} of Angel, i.e.\ the set of states \m{\strategyfor[\alpha]{X}} from which Angel has a winning strategy to achieve $X$ (whatever strategy Demon chooses).
It is defined inductively as follows:
\begin{enumerate}
\setlength{\itemsep}{-2pt}
\item \(\strategyfor[\pdiffgame{\D{x}=\genDG{x}{y}{z}}{y\in Y}{z\in Z}]{X}\)
=\\ \(\{\iget[state]{\I} \in \linterpretations{\Sigma}{V} \with\)
\(\mexists{T{\geq}0}{\mexists{\beta\in\stratA{\eta}}{\mforall{y\in\controlD{\eta}}{\mexists{0{\leq}\zeta{\leq}T}}}}\)
\(\response[f]{\zeta}{\iget[state]{\I}}{y}{\beta(y)} \in X\}\)

\item \(\strategyfor[\pupdate{\pumod{x\,}{\theta}}]{X} = \{\iportray{\I} \in \linterpretations{\Sigma}{V} \with \modif{\iget[state]{\I}}{x}{\ivaluation{\I}{\theta}} \in X\}\)

\item \(\strategyfor[\ptest{\ivr}]{X} = \imodel{\I}{\ivr}\cap X\)
\item \(\strategyfor[\pchoice{\alpha}{\beta}]{X} = \strategyfor[\alpha]{X}\cup\strategyfor[\beta]{X}\)
\item \(\strategyfor[\alpha;\beta]{X} = \strategyfor[\alpha]{\strategyfor[\beta]{X}}\)
\item \(\strategyfor[\prepeat{\alpha}]{X} = \capfold\{Z\subseteq\linterpretations{\Sigma}{V} \with X\cup\strategyfor[\alpha]{Z}\subseteq Z\}\)

\item \(\strategyfor[\pdual{\alpha}]{X} = \scomplement{(\strategyfor[\alpha]{\scomplement{X}})}\)
\end{enumerate}
The \emph{winning region} of Demon is a function \m{\dstrategyfor[\alpha]{\argholder}}, which, for each of Demon's winning states $X\subseteq\linterpretations{\Sigma}{V}$ gives the set of states \m{\dstrategyfor[\alpha]{\argholder}} from which Demon has a winning strategy to achieve $X$ (whatever strategy Angel chooses).
It is defined inductively as:
\begin{enumerate}
\setlength{\itemsep}{-2pt}
\item \(\dstrategyfor[\pdiffgame{\D{x}=\genDG{x}{y}{z}}{y\in Y}{z\in Z}]{X}\)
=\\
\(\{\iget[state]{\I} \in \linterpretations{\Sigma}{V} \with\)
\(\mforall{T{\geq}0}{\mforall{\beta\in\stratA{\eta}}{\mexists{y\in\controlD{\eta}}{\mforall{0{\leq}\zeta{\leq}T}{}}}}\)
\(\response[f]{\zeta}{\iget[state]{\I}}{y}{\beta(y)} \in X\}\)

\item \(\dstrategyfor[\pupdate{\pumod{x\,}{\theta}}]{X} = \{\iportray{\I} \in \linterpretations{\Sigma}{V} \with \modif{\iget[state]{\I}}{x}{\ivaluation{\I}{\theta}} \in X\}\)
\item \(\dstrategyfor[\ptest{\ivr}]{X} = \scomplement{(\imodel{\I}{\ivr})}\cup X\)
\item \(\dstrategyfor[\pchoice{\alpha}{\beta}]{X} = \dstrategyfor[\alpha]{X}\cap\dstrategyfor[\beta]{X}\)
\item \(\dstrategyfor[\alpha;\beta]{X} = \dstrategyfor[\alpha]{\dstrategyfor[\beta]{X}}\)
\item \(\dstrategyfor[\prepeat{\alpha}]{X} = \cupfold\{Z\subseteq\linterpretations{\Sigma}{V} \with Z\subseteq X\cap\dstrategyfor[\alpha]{Z}\}\)
\item \(\dstrategyfor[\pdual{\alpha}]{X} = \scomplement{(\dstrategyfor[\alpha]{\scomplement{X}})}\)
\end{enumerate}
\end{definition}

The compositional semantics of differential hybrid games agrees with that of hybrid games \cite{DBLP:journals/tocl/Platzer15} except for the addition of differential games.
Time horizon $T$, nonanticipative strategy $\passfunction{\beta}$ for $Z$, and stopping times $\zeta$ of differential games are Angel's choice while control of $y$ is Demon's choice.
Angel first chooses a finite time horizon $T$ and nonanticipative strategy $\beta$, but the corresponding control $\passfunction{\beta(y)}$ from her nonanticipative strategy $\beta$ gives her a chance to observe Demon's current action from Demon's control $y$.
Angel ultimately gets to inspect the resulting state and decide at what time $\zeta$ she wants to stop playing the differential game.
This is the continuous counterpart of $\prepeat{\alpha}$, where Angel gets to inspect the state and decide whether she wants to repeat the loop again or not, which follows from the fixpoint semantics of $\prepeat{\alpha}$; see \cite{DBLP:journals/tocl/Platzer15}.
The fact that Angel has to choose some arbitrarily large but finite time horizon $T$ first corresponds to her not being allowed to play the differential game indefinitely, just like she is not allowed to repeat playing $\prepeat{\alpha}$ forever, which again results from its least fixpoint semantics \cite{DBLP:journals/tocl/Platzer15}.
Demon has a winning strategy in the differential game \(\pdiffgame{\D{x}=\genDG{x}{y}{z}}{y\in Y} {z\in Z}\) to achieve $X$ if for all of Angel's time horizons $T$ and all of Angel's nonanticipative strategies $\beta$ for $Z$ there is a control $y\in\controlD{\eta}$ for Demon such that, for all of Angel's stopping times $\zeta$, the game ends in one of Demon's winning states (i.e.\ in $X$).
Demon knows $\beta\in\stratA{\eta}$ when choosing $y\in\controlD{\eta}$, so he can predict the states over time by solving \rref{eq:diffgame} from \rref{def:diffgame} via \rref{lem:response}.
Angel can predict the states over time by \rref{lem:response} as well, since her strategy $\beta\in\stratA{\eta}$ receives Demon's control $y\in\controlD{\eta}$ as an input.
But Angel's nonanticipative $\beta$ allows \(\beta(y)(s)\) to depend on $y(s)$, which gives her the information advantage for the current action.

The (dual) quantifier order for $\ddiamond{\cdot}{}$ is the same, so that Angel finds some $\beta\in\stratA{\eta}$ that works for any $y\in\controlD{\eta}$ since she cannot predict what Demon will play.
Hence, the informational advantage of the opponent's current action as well as the advantage of controlling time in a differential game consistently goes to Angel, whether asking for Angel's winning strategy in $\ddiamond{\cdot}{}$ or for Demon's winning strategy in $\dbox{\cdot}{}$.
The same game is played in \(\dbox{\pdiffgame{\D{x}=\genDG{x}{y}{z}}{y\in Y}{z\in Z}}{}\) and in \(\ddiamond{\pdiffgame{\D{x}=\genDG{x}{y}{z}}{y\in Y}{z\in Z}}{}\) with the same order of information as indicated by the notation of differential games, just from the perspective of winning strategies for different players.

The last quantifier $\zeta$ might appear to be unimportant, because, if Angel wins, then from any state there is a maximum time horizon $T$ within which she wins so that it seems like it would be enough for her to choose that maximum time horizon $T$ and check for the terminal state at time $T$.
However, Demon might then still let Angel ``win'' earlier by playing suboptimally if that gives him the possibility of moving outside Angel's winning condition again before the winning condition is checked at time $T$.
It is, thus, important for Angel to be able to stop the differential game at any time based on the state she observes.
She will want to stop when the game reached her target.
Consider, e.g., the \emph{race car game}, where Demon is in control of a car toward a goal \m{x^2<1} and Angel is in control of time but has no other control input $z$:
\[
x=-9\land t=0\limply
\dbox{\pdiffgame{\D{x}=y\syssep\D{t}=1}{y\in{\scriptstyle[}1,2{\scriptstyle]}}{}}{(4{<}t{<}8 \limply x^2{<}1)}
\]
If Angel were to declare that she will choose $\zeta=T$ upfront by advance notice, then Demon could compute an optimal velocity \m{y=\frac{8}{T}=\frac{8}{\zeta}\in[1,2]} and will win at time $\zeta=T$, since Angel has to choose \(4<\zeta<8\) to stand a chance to win.
Since Angel, however, only declares a time bound $T$ and chooses the actual stopping time $\zeta$ only after Demon revealed $y\in\controlD{\eta}$, she can choose \(T=7\) and end the game at \(\zeta=4.1\) to win if Demon has not moved $x$ to \(x^2<1\) at time $4.1$ yet.
If Demon has moved $x$ to \(x^2<1\), so \(x>-1\), then Angel still wins by waiting the full time \(\zeta=T\), which is at least 2 more seconds, during which Demon must have moved by at least 2 along $\D{x}=y$ and left \(x^2<1\) to lose.
This example hinges on a postcondition that is neither open nor closed.

Since differential hybrid games have the same information structure for \m{\ddiamond{\alpha}{}} and \m{\dbox{\alpha}{}}, just referring to another player's winning strategy, the determinacy theorem \cite[Thm.\,3.1]{DBLP:journals/tocl/Platzer15} extends to differential hybrid games.
In each state, exactly one player has a winning strategy, i.e.\ either Angel has a winning strategy to achieve $\lnot\phi$ or Demon has a winning strategy to achieve $\phi$.
The proof is much easier than (partial) determinacy results for other scenarios and other information patterns of differential games \cite{DBLP:journals/siamco/Cardaliaguet96}, and gains simplicity compared to Borel determinacy, just like the determinacy theorem for hybrid games that it is based on \cite[Thm.\,3.1]{DBLP:journals/tocl/Platzer15}.

\begin{theorem}[Determinacy] \label{thm:dGL-determined}%
  Differential hybrid games are determined, i.e.\ 
\ \m{\entails \lnot\ddiamond{\alpha}{\lnot\phi} \lbisubjunct \dbox{\alpha}{\phi}}.
\end{theorem}%
\begin{proofatend}
The proof shows by induction on the structure of $\alpha$ that
\m{\scomplement{\strategyfor[\alpha]{\scomplement{X}}} = \dstrategyfor[\alpha]{X}} for all \m{X\subseteq\linterpretations{\Sigma}{V}}, which implies the validity of \m{\lnot\ddiamond{\alpha}{\lnot\phi} \lbisubjunct \dbox{\alpha}{\phi}} using \m{X\mdefeq\imodel{\I}{\phi}} by \rref{def:dGL-semantics}.
The only difference to a corresponding determinacy proof for hybrid games \cite{DBLP:journals/tocl/Platzer15} is the additional \rref{case:diffgame-determined} for differential games, which follows directly from \rref{def:DHG-semantics}:
\begin{enumerate}
\item \label{case:diffgame-determined}
\(\scomplement{\strategyfor[\pdiffgame{\D{x}= \genDG{x}{y}{z}}{y\in Y}{z\in Z}]{\scomplement{X}}}\)
\\= \(\scomplement{\{\iget[state]{\I} \in \linterpretations{\Sigma}{V} \with\)
\(\mexists{T{\geq}0}{\mexists{\beta\in\stratA{\eta}}{\mforall{y\in\controlD{\eta}}{\mexists{0{\leq}\zeta{\leq}T}}}}\)
\(\response[f]{\zeta}{\iget[state]{\I}}{y}{\beta(y)}) \in \scomplement{X}\}}\)
\\= \(\{\iget[state]{\I} \in \linterpretations{\Sigma}{V} \with\)
\(\mnot\, \mexists{T{\geq}0}{\mexists{\beta\in\stratA{\eta}}{\mforall{y\in\controlD{\eta}}{\mexists{0{\leq}\zeta{\leq}T}}}}\)
\(\response[f]{\zeta}{\iget[state]{\I}}{y}{\beta(y)}) \not\in X\}\)
\\= \(\{\iget[state]{\I} \in \linterpretations{\Sigma}{V} \with\)
\(\mforall{T{\geq}0}{\mforall{\beta\in\stratA{\eta}}{\mexists{y\in\controlD{\eta}}{\mforall{0{\leq}\zeta{\leq}T}}}}\)
\(\response[f]{\zeta}{\iget[state]{\I}}{y}{\beta(y)}) \in X\}\)
\\= \(\dstrategyfor[\pdiffgame{\D{x}=\genDG{x}{y}{z}}{y\in Y}{z\in Z}]{X}\)

\item \(\scomplement{\strategyfor[\pupdate{\pumod{x}{\theta}}]{\scomplement{X}}}
= \scomplement{\{\iportray{\I} \in \linterpretations{\Sigma}{V} \with \modif{\iget[state]{\I}}{x}{\ivaluation{\I}{\theta}} \not\in X\}}
= \strategyfor[\pupdate{\pumod{x}{\theta}}]{X}
= \dstrategyfor[\pupdate{\pumod{x}{\theta}}]{X}\)

\item \(\scomplement{\strategyfor[\ptest{\ivr}]{\scomplement{X}}}
= \scomplement{(\imodel{\I}{\ivr}\cap\scomplement{X})}
= \scomplement{(\imodel{\I}{\ivr})}\cup \scomplement{(\scomplement{X})}
= \dstrategyfor[\ptest{\ivr}]{X}\)

\item \(\scomplement{\strategyfor[\pchoice{\alpha}{\beta}]{\scomplement{X}}}
= \scomplement{(\strategyfor[\alpha]{\scomplement{X}} \cup \strategyfor[\beta]{\scomplement{X}})}
= \scomplement{\strategyfor[\alpha]{\scomplement{X}}} \cap \scomplement{\strategyfor[\beta]{\scomplement{X}}}
= \dstrategyfor[\alpha]{X}\cap\dstrategyfor[\beta]{X}
= \dstrategyfor[\pchoice{\alpha}{\beta}]{X}\)

\item \(\scomplement{\strategyfor[\alpha;\beta]{\scomplement{X}}}
= \scomplement{\strategyfor[\alpha]{\strategyfor[\beta]{\scomplement{X}}}}
= \scomplement{\strategyfor[\alpha]{\scomplement{\dstrategyfor[\beta]{X}}}}
= \dstrategyfor[\alpha]{\dstrategyfor[\beta]{X}}
= \dstrategyfor[\alpha;\beta]{X}\)

\item \(\scomplement{\strategyfor[\prepeat{\alpha}]{\scomplement{X}}}
= \scomplement{\left(\capfold\{Z\subseteq\linterpretations{\Sigma}{V} \with \scomplement{X}\cup\strategyfor[\alpha]{Z}\subseteq Z\}
\right)}\)
= \(\scomplement{\left(\capfold\{Z\subseteq\linterpretations{\Sigma}{V} \with \scomplement{(X\cap\scomplement{\strategyfor[\alpha]{Z}})}\subseteq Z\}\right)}\)\\
= \(\scomplement{\left(\capfold\{Z\subseteq\linterpretations{\Sigma}{V} \with \scomplement{(X\cap\dstrategyfor[\alpha]{\scomplement{Z}})}\subseteq Z\}\right)}\)
= \(\cupfold\{Z\subseteq\linterpretations{\Sigma}{V} \with Z\subseteq X\cap\dstrategyfor[\alpha]{Z}\}
= \dstrategyfor[\prepeat{\alpha}]{X}\).
\footnote{The penultimate equation follows from the $\mu$-calculus equivalence
\(\gfp{Z}{\mapply{\Upsilon}{Z}} \mequiv \lnot\lfp{Z}{\lnot\mapply{\Upsilon}{\lnot Z}}\) and the fact that least pre-fixpoints are fixpoints and that greatest post-fixpoints are fixpoints for monotone functions.}

\item \(\scomplement{\strategyfor[\pdual{\alpha}]{\scomplement{X}}}
= \scomplement{(\scomplement{\strategyfor[\alpha]{\scomplement{(\scomplement{X})}}})}
= \scomplement{\dstrategyfor[\alpha]{\scomplement{X}}}
= \dstrategyfor[\pdual{\alpha}]{X}\)
\qedhere
\end{enumerate}
\end{proofatend}%

A formula is called \emph{atomically open} if its negation normal form is built from $\land$,$\lor$,$>$,$<$.
Atomically open formulas define topologically open sets.
The converse is not true, because there can be spurious extra subformulas:
\(0<x\land x<5 \lor x=2\) is topologically open but has an irrelevant topologically closed subformula \(x=2\).
A formula is \emph{atomically closed} if its negation normal form is built from $\land,\lor,\geq,\leq,=$.
Atomically closed formulas define topologically closed sets.
Both converses can always be made true by transforming formulas to avoid superfluous subformulas \cite[2.7.2]{BochnakCR98}. %
The primary focus in this article is on postconditions that are open or closed.

\begin{lemma}[$\reals$ arithmetization] \label{lem:arithmetize}
There is an effective mapping $\arithmetize{(\cdot)}$ from first-order formulas to (continuous) terms of mixed polynomials, $\min$, and $\max$.
If $F$ is atomically open, then
\(\entails F \lbisubjunct (\arithmetize{F}>0)\).
If $F$ is atomically closed,
\(\entails F \lbisubjunct (\arithmetize{F}\geq0)\).
\end{lemma}
\begin{proofatend}
By quantifier elimination \cite{tarski_decisionalgebra51}, $F$ can be assumed to be quantifier-free and in negation normal form.
The term $\arithmetize{F}$ of mixed polynomials, $\min$, $\max$ for $F$ is defined inductively, which obeys the desired properties:
\begin{align*}
  \arithmetize{(a\geq b)} &\mequiv \arithmetize{(a>b)} \mequiv a-b\\
  \arithmetize{(a<b)} &\mequiv \arithmetize{(b>a)}\\
  \arithmetize{(a\leq b)} &\mequiv \arithmetize{(b\geq a)}\\
  \arithmetize{(a=b)} &\mequiv \arithmetize{(a\geq b\land b\geq a)}\\
  \arithmetize{(F\land G)} &\mequiv \min(\arithmetize{F},\arithmetize{G})\\
  \arithmetize{(F\lor G)} &\mequiv \max(\arithmetize{F},\arithmetize{G})
  \qedhere
\end{align*}
\end{proofatend}

Even if not all are necessary, the assumptions in \rref{def:diffgame} will be required to hold when playing differential games from \rref{def:dGL-DHG}.
They can be checked using the relations in \rref{sec:preliminaries}, which are decidable for the relevant terms in first-order real arithmetic \cite{tarski_decisionalgebra51}.
The easiest criterion, however, is the following:
\begin{lemma}[Well-definedness] \label{lem:well-defined}
  If $f$ is bounded for compact \(\imodel{\I}{y\in Y}\),\(\imodel{\I}{z\in Z}\) and $F$ is open or closed, then all differential games for
  \(\dbox{\pdiffgame{\D{x}=f(x,y,z)}{y\in Y}{z\in Z}}{F}\)
  and\\
  \(\ddiamond{\pdiffgame{\D{x}=f(x,y,z)}{y\in Y}{z\in Z}}{F}\)
  are well-defined.
\end{lemma}
\begin{proofatend}
Let $b$ a bound on the norm of $f$.
For any initial state $\iget[state]{\I}$ and any time horizon $T\geq0$ (\rref{def:DHG-semantics}) any response \(\response[f]{\zeta}{\iget[state]{\I}}{y}{\beta(y)}\) of \rref{eq:diffgame} remains on the compact ball of radius \(bT\) around $\iget[state]{\I}$.
Without changing the differential game, $f$ can, thus, be replaced by an $\hat{f}$ that agrees with $f$ on this compact ball and accordingly for the payoff.
On that compact set, the \dGL term $f$ and the arithmetization $\arithmetize{F}$ define Lipschitz-continuous functions (even when using $\min,\max$ terms) as follows.
Polynomials are smooth and, thus, Lipschitz on compact sets.
The absolute value function is Lipschitz.
The composition \(\min(x,y)=(x+y)/2-\abs{x-y}/2\) of Lipschitz functions is Lipschitz\footnote{%
\(\abs{f(g(x))-f(g(y))}\leq L\abs{g(x)-g(y)} \leq LK\abs{x-y}\) makes the composition \(f\compose g\) $LK$-Lipschitz when $f$ is $L$-Lipschitz and $g$ is $K$-Lipschitz, respectively.}
and so is \(\max(x,y)=-\min(-x,-y)\).
By Tietze \cite[2.19]{Walter:Ana2}, there are Lipschitz-continuous extensions $\hat{f}$ of $f$ and $\hat{g}$ of \(\arithmetize{F}\) that agree on the compact ball and remain bounded.
The differential game \(\pdiffgame{\D{x}=\hat{f}(x,y,z)}{y\in Y}{z\in Z}\) with payoff \(\hat{g}\) is, thus, equivalent by Lemmas~\ref{lem:response} and~\ref{lem:arithmetize} and it meets the requirements of \rref{def:diffgame}.
\end{proofatend}
For any horizon $T$ and initial state $\iget[state]{\I}$ as used in \rref{def:DHG-semantics}, the right-hand side $f$ of a differential game can be replaced in similar ways by a bounded function without changing the game \cite{DBLP:journals/siamco/GruneS11},
since $f$ is continuous by \rref{def:dGL-DHG}.
Unlike semantic differential games (\rref{def:diffgame}), the differential games in the logic \dGL (\rref{def:dGL-DHG}) have no implicit time-dependency but need an explicit extra clock variable $t$ with differential equation \m{\D{t}=1} to express time-dependencies.
Retaining an explicit time-dependency for semantic differential games (\rref{def:diffgame}) is helpful for the soundness proofs, however.

\section{Differential Game Proofs} \label{sec:diffgame-axioms}

This section introduces sound induction principles for differential games with differential game invariants and differential game variants as well as ways of comparing differential games by differential game refinements.

Differential equations are already hard to solve and it is challenging or impossible to use their solutions for proofs \cite{DBLP:conf/lics/Platzer12b}.
It is even more difficult, however, to solve differential games, because their Carath\'eodory solutions depend on the control choices adopted by the two players, which can be arbitrary measurable functions and are mutually dependent.
A direct representation of this would, thus, require not just alternating quantification over arbitrary measurable functions, but also the ability to solve all resulting Carath\'eodory-type ordinary differential equations and to prove properties about all their respective behaviors -- a truly daunting enterprise.

\begin{figure*}[tbh]
\begin{calculus}
\cinferenceRule[diffgameind|DGI]{differential game invariants}
{\linferenceRule[sequent]
{ \lexists{y\in Y}{\lforall{z \in Z}{ \dder[\D{x}][f(x,y,z)]{F} }}}
{
  F \limply \dbox{\pdiffgame{\D{x}=\genDG{x}{y}{z}}{y\in Y}{z\in Z}}{F}
}
}{}
\cinferenceRule[diffgamefin|DGV]{differential game variant}
{\linferenceRule[sequent]
{ \lexists{\varepsilon{>}0}{\lforall{x}{\lexists{z \in Z}{\lforall{y\in Y}{( g\leq0 \limply \dder[\D{x}][f(x,y,z)]{g} \geq\varepsilon)}}}}}
{
  \ddiamond{\pdiffgame{\D{x}=\genDG{x}{y}{z}}{y\in Y}{z\in Z}}{g\geq0}
}
}{}
  \cinferenceRule[diffgamerefine|DGR]{differential game refinement}
  {\linferenceRule[sequent]
    {\lforall{u\in U}{\lexists{y\in Y}{\lforall{z\in Z}{\lexists{v\in V}{\lforall{x}{(f(x,y,z)=g(x,u,v))}}}}}}
    {\dbox{\pdiffgame{\D{x}=g(x,u,v)}{u\in U}{v\in V}}{F} \limply \dbox{\pdiffgame{\D{x}=f(x,y,z)}{y\in Y}{z\in Z}}{F}}
  }
  {}
\end{calculus}
  \caption{Differential game proof rules}
  \label{fig:diffgameind}
\end{figure*}
\emph{Differential game invariants}, instead, define a simple induction principle for differential games.
The proof rule of \emph{differential game invariants} and its counterpart for \emph{differential game variants} are shown in \rref{fig:diffgameind}.
Differential game invariants (rule \irref{diffgameind}) have a simple intuition.
\irref{diffgameind} checks if, in each $x$ ($\forall{x}$ is implicit in the premise of \irref{diffgameind} and follows from the definition of validity in \rref{def:dGL-semantics}), there is a local choice of control action $y$ for Demon such that, for all choices of Angel's control action $z$, the derivative \(\dder[\D{x}][f(x,y,z)]{F}\) of $F$ holds when substituting the right-hand side \(f(x,y,z)\) of the differential game for the left-hand side $\D{x}$.
The precise meaning of \(\dder[\D{x}][f(x,y,z)]{F}\) will be developed subsequently.
If the derivative \m{\dder[\D{x}][f(x,y,z)]{F}} represents the change of the truth-value of $F$ along \(\D{x}=f(x,y,z)\), then \irref{diffgameind} makes intuitive sense since its premise means that there is always a local way for Demon to make $F$ ``more true'' with $y$, whatever Angel is trying with $z$.
So, Demon has a winning strategy no matter how long Angel decides to evolve.
Recall that Angel gets to inspect Demon's current $y$ action in her nonanticipative strategy before choosing $z$, which explains the order of quantifiers in \irref{diffgameind}, where Demon first chooses $y\in Y$ that works for all of Angel's $z\in Z$ since Angel chooses last.

\emph{Differential game variants} (rule \irref{diffgamefin}) also have a simple intuition.
Angel can reach the postcondition if, from any state where she has not won yet, there is a progress of at least some $\varepsilon>0$ towards the goal that, uniformly at all $x$, she can realize for some $z$ control choice of hers, no matter what $y$ action Demon chose.
The quantifier order $z,y$ in \irref{diffgamefin} is conservative compared to $y,z$ to simplify the proofs.
Other postconditions are possible based on \rref{lem:arithmetize}, but \irref{diffgamefin} becomes notationally more involved then.

\emph{Differential game refinements} (rule \irref{diffgamerefine}) relate differential games whose equations can be aligned when matching the $u\in U$ control that Demon sought in its antecedent with some of Demon's $y\in Y$ control in the succedent if any control $z\in Z$ that Angel has in the succedent can be conversely matched by a control $v\in V$ that Angel already had in the antecedent.
Via the induced identification of controls, Demon's winning strategy for the differential game in the antecedent carries over to a winning strategy for the differential game in the succedent if Demon has more control power in the succedent while Angel has less.
A dual of \irref{diffgamerefine} for $\ddiamond{\cdot}{}$ derives by \rref{thm:dGL-determined}.
With a cut, \irref{diffgamerefine} can transform differential games from the succedent to the antecedent by refinement.

As with invariants, it may sometimes be difficult to find good differential game invariants or differential game variants for the proof of a property.
Once found, however, they are computationally attractive, since easy to check by decidable arithmetic.

Differential game invariants and differential game variants use syntactic total derivations to compute \emph{differential game derivatives} syntactically.

\begin{definition}[Derivation] \label{def:derivation}
  The operator~\m{\der{\argholder}} that is defined as follows on terms is called \emph{syntactic (total) derivation} from (for simplicity polynomial) terms to differential terms, i.e.\ terms in which differential symbols \m{\D{x}} for variables $x$ are allowed\index{derivation!syntactic}:
  \index{differential!symbol}%
  \begin{subequations}
  \begin{align*}
    \der{r} & = 0
      \hspace{2.1cm}\text{for numbers}~r\in\reals
    \\
    \der{x} & =  \D{x}
      \hspace{2cm}\text{for variables}~x
      \\
    \der{a+b} & = \der{a} + \der{b}
    \\
    \der{a \itimes b} & = \der{a}\itimes b + a\itimes\der{b}
  \end{align*}
  \end{subequations}%
  It extends to (quantifier-free) first-order real-arithmetic formulas $F$ as follows:
  \begin{subequations}
  \begin{align*}
    \der{F\land G} &\,\mequiv\, \der{F} \land \der{G}\\
    \der{F\lor G} &\,\mequiv\, \der{F} \land \der{G}
    \\
    \der{a\geq b} &\,\mequiv\, \der{a>b} \,\mequiv\, \der{a} \geq \der{b}\\
    \der{a\leq b} &\,\mequiv\, \der{a<b} \,\mequiv\, \der{a} \leq \der{b}\\
    \der{a=b} &\,\mequiv\, \der{a\neq b} \,\mequiv\, \der{a} = \der{b}
  \end{align*}
  \end{subequations}%
  Define \m{\dder[\D{x}][\theta]{F}} to be the result of substituting term $\theta$ for $\D{x}$ in \m{\der{F}} and substituting 0 for all other differential symbols \m{\D{c}} that have no differential equation / differential game.
\end{definition}

The relation of the syntactic derivation $\der{e}$ to analytic differentiation is as follows, which identifies the semantics of the syntactic term \m{\subst[\der{e}]{\D{x}}{\theta}} with a Lie-derivative.
\begin{lemma}[Derivation] \label{lem:Lie-relation}
  Let $\theta$ be a (vectorial) term of the same dimension as $x$ and let $e$ be any term, then 
  \(\ivaluation{\I[\xi]}{\subst[\der{e}]{\D{x}}{\theta}} = \ivaluation{\I[\xi]}{\theta} \stimes D_x \ivaluation{\I[\xi]}{e}\),
  where \(D_x \ivaluation{\I[\xi]}{e}\) is the gradient at state $\xi$ with respect to variables $x$ of the value of term $e$.
\end{lemma}
\begin{proofatend}
By a notational variation of a previous result \cite[Lem.\,3.3]{DBLP:journals/lmcs/Platzer12}.
\end{proofatend}

The rules in \rref{fig:diffgameind} assume the well-definedness condition from \rref{lem:well-defined}.
A complete axiomatization for the other hybrid game operators of \dGL \cite{DBLP:journals/tocl/Platzer15} is in \rref{app:dGL-HG-axiomatization}.
They play no further role for this article, though, except to manifest how seamlessly differential games proving integrates with hybrid games proving in \dGL.

While a strong point of \dGL is that it enables such a seamless integration of differential games and hybrid games in modeling and analysis, the subsequent examples focus primarily on differential games in order to highlight its novel aspects.
Consider the strength game with \(-1{\leq} y{\leq}1\) abbreviated by \(y\in I\), which proves easily with \irref{diffgameind}:
{\def\arraystretch{1.1}%
\begin{sequentdeduction}[array]
  \linfer[diffgameind]
  {\linfer
    {\linfer[qear]
      {\lclose}
      {\lsequent{} {\lexists{y\in I}{\lforall{z\in I}{0\leq 3x^2(-1+2y+z)}}}}
    }
    {\lsequent{} {\lexists{y\in I}{\lforall{z\in I}{\subst[(0\leq 3x^2\D{x})]{\D{x}}{-1+2y+z}}}}}
  }
  {\lsequent{1\leq x^3} {\dbox{\pdiffgame{\D{x}=-1+2y+z}{y\in I}{z\in I}}{\,1\leq x^3}}}
\end{sequentdeduction}
}%
Using vectorial notation, let \(y\in B\) be \(y_1^2+y_2^2\leq1\).
Let terms \m{L\leq M} denote the maximum speeds of vectors $l$ and $m$.
The simple pursuit \cite{Isaacs:DiffGames}, that vector $m$ can escape the vector $l$, proves easily:
\begin{sequentdeduction}[array]
  \linfer[diffgameind]
  {\linfer
    {\linfer[qear]
      {\lclose}
      {\lsequent{} {\lexists{y\in B}{\lforall{z\in B}{(2(l-m)\stimes(Lz-My)\geq0)}}}}
    }
    {\lsequent{} {\lexists{y\in B}{\lforall{z\in B}{\subst[(2(l-m)\stimes(\D{l}-\D{m})\geq0)]{\D{m}}{My}\subst[\,]{\D{l}}{Lz}}}}}
  }
  {\lsequent{\norm{l{-}m}^2{>}0} {\dbox{\pdiffgame{\D{m}=My\syssep\D{l}=Lz}{y{\in} B}{z{\in} B}}{\,\norm{l{-}m}^2{>}0}}}
\end{sequentdeduction}
Almost the same proof shows that a positive distance \({\norm{l{-}m}^2{\geq}1}\) can be maintained.
A non-convex  region \(y\in Y\) defined as \(y^2=1\) or games with input by just one player work as well (similar for higher dimensions):
{\def\arraystretch{1.1}%
\begin{sequentdeduction}[array]
  \linfer[diffgameind]
  {\linfer
    {\linfer[qear]
      {\lclose}
      {\lsequent{} {\lexists{y^2=1}{(3x^2 x^3y\geq4x x^3y)}}}
    }
    {\lsequent{} {\lexists{y^2=1}{\subst[(3x^2\D{x}\geq4x\D{x})]{\D{x}}{x^3y}}}}
  }
  {\lsequent{x^3>2x^2-2} {\dbox{\pdiffgame{\D{x}=x^3y}{y^2=1}{}}{\,x^3>2x^2-2}}}
\end{sequentdeduction}
}%
To fit to the simple well-definedness condition (\rref{lem:well-defined}), the differential equation\\ \(\D{x}=\max(\min(x^3y,k),-k)\) could be used instead, which proves for all bounds $k\geq0$.
Alternatively, global bounding \(\D{x}=x^3y/(1+\sqrt{(x^3y)^2})\), which does not change the game outcome \cite{DBLP:journals/siamco/GruneS11}, proves, too.
These simple proofs entail for all nonanticipative strategies the existence of measurable control functions to win the game.

The last example proof and the following \irref{diffgamerefine} refinement proof
\begin{sequentdeduction}[array]
\linfer[diffgamerefine]
  {\linfer[qear]
    {\lclose}
    {\lsequent{} {\lforall{u^2=1}{\lexists{0{\leq}y{\leq}1}{(2x^3y-x^3=x^3u)}}}}
  }
  {\lsequent{\dbox{\pdiffgame{\D{x}=x^3y}{y^2=1}{}}{x^3>2x^2-2}} {\dbox{\pdiffgame{\D{x}=2x^3y-x^3}{0{\leq}y{\leq}1}{}}{\,x^3>2x^2-2}}}
\end{sequentdeduction}
combine with a cut to a proof of:
\[
  {{x^3>2x^2-2} \limply {\dbox{\pdiffgame{\D{x}=2x^3y-x^3}{0\leq y\leq1}{}}{\,x^3>2x^2-2}}}
\]
Another example for the use of proof rule \irref{diffgamerefine} is in the proof of \rref{lem:freezing}.
With the above use of the well-definedness condition, the spiral game proves using rule \irref{diffgamefin}:
\begin{sequentdeduction}[array]
\linfer[diffgamefin]
  {\linfer
    {\linfer[qear]
      {\lclose}
      {\lsequent{} {\lexists{\varepsilon{>}0}{\lforall{x}{\lforall{u}{\lexists{{-}1{\leq}z{\leq}1}{\lforall{{-}2{\leq}y{\leq}2}
      {\big( x^2+u^2\geq1 \limply -2x(zx-yu)-2u(zu+yx) \geq \varepsilon \big)}
      }}}}}}
    }
    {\lsequent{} {\lexists{\varepsilon{>}0}{\lforall{x}{\lforall{u}{\lexists{{-}1{\leq}z{\leq}1}{\lforall{{-}2{\leq}y{\leq}2}
      {\big( 1-x^2-u^2\leq0 \limply \subst[(-2x\D{x}-2u\D{u})]{\D{x}}{zx-yu}\subst[\,]{\D{u}}{zu+yx} \geq \varepsilon \big)}
    }}}}}}
  }
  {\lsequent{} {\ddiamond{\pdiffgame{\D{x}=zx-yu\syssep\D{u}=zu+yx}{{-}2{\leq}y{\leq}2}{{-}1{\leq}z{\leq}1}}{\,1-x^2-u^2\geq0}}}
\end{sequentdeduction}

\begin{example}[Zeppelin]
\def\arraystretch{1.2}%
First continue the differential game of \rref{ex:Zeppelin} in isolation, focusing on an obstacle \(o=(0,0)\) at the origin with radius $c=1$ for simplicity.
If the Zeppelin propeller outpowers the wind and turbulence (\(p-r\geq\norm{v},r\geq0\)), the Zeppelin easily wins from any safe position, as proved by arithmetic simplification:
\begin{sequentdeduction}[array]
  \linfer[diffgameind]
  {\linfer
    {\linfer[qear]
      {\lclose}
      {\lsequent{} {\lexists{y\in B}{\lforall{z\in B}{(2x_1(v_1+py_1+rz_1) + 2x_2(v_2+py_2+rz_2)\geq0)}}}}
    }
    {\lsequent{} {\lexists{y\in B}{\lforall{z\in B}{\subst[(2x\stimes\D{x}\geq0)]{\D{x}}{v+py+rz}}}}}
  }
  {\lsequent{\norm{x}^2\geq c^2} {\dbox{\pdiffgame{\D{x}=v+py+rz}{y\in B}{z\in B}}{\,\norm{x}^2\geq c^2}}}
\end{sequentdeduction}
For a mediocre propeller (with \(0\leq r<p \leq \norm{v} + r\)), the differential game is significantly more challenging, but the Zeppelin still wins when it starts at sufficient distance to the obstacle.
It may take up to duration $\frac{c}{p-r}$ to progress by a distance of $c$ in the direction orthogonal to $v$, during which the wind field displaces the Zeppelin by $\frac{c}{p-r}v$.
With focal point \(q\mdefeq\frac{-c}{p-r}v\), which has orthogonal complement \(\ortho{q}=(-q_2,q_1)\), choose condition $C$ as the regions outside the tangents through $q$ to the circle of radius $c$ (see \rref{fig:Zeppelin-safe}):%
\begin{align*}
C \mdefequiv\,&
c q\stimes(x-q)  \pm  \sqrt{\norm{q}^2-c^2} \ortho{q}\stimes x \geq0
\\
\mequiv\, & c q\stimes(x-q)  +  \sqrt{\norm{q}^2-c^2} \ortho{q}\stimes x \geq0 
~\lor~ c q\stimes(x-q)  -  \sqrt{\norm{q}^2-c^2} \ortho{q}\stimes x \geq0
\\
\mequiv\, &
c (q_1(x_1-q_1)+q_2(x_2-q_2)) + \sqrt{q_1^2+q_2^2-c^2} (q_1 x_2 - q_2 x_1)\geq0\\
\lor\,& c (q_1(x_1-q_1)+q_2(x_2-q_2)) - \sqrt{q_1^2+q_2^2-c^2} (q_1 x_2 - q_2 x_1)\geq0
\end{align*}
\begin{figure}[tb]
  \includegraphics[height=6cm]{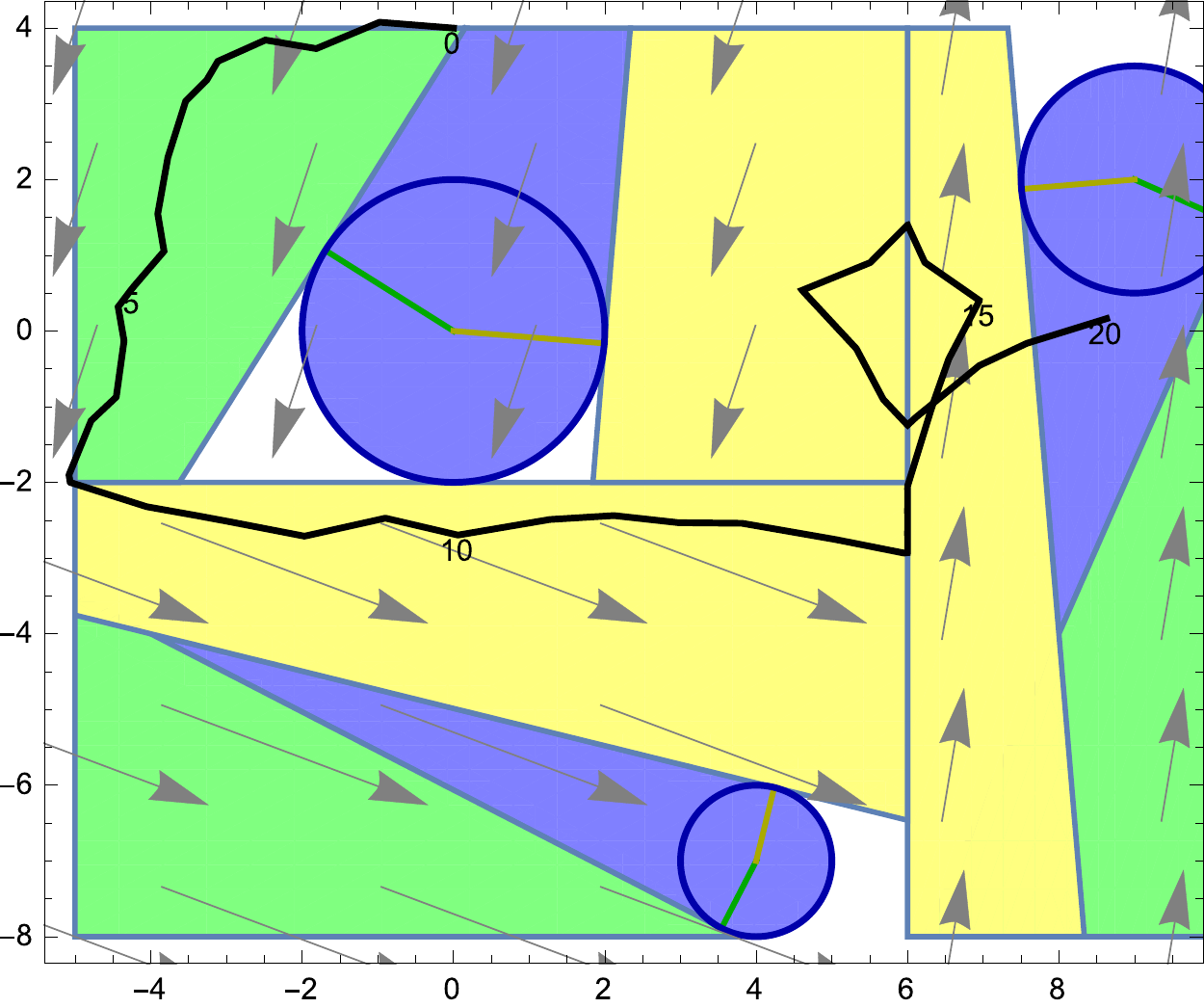}%
  \hfill%
  \includegraphics[height=6cm]{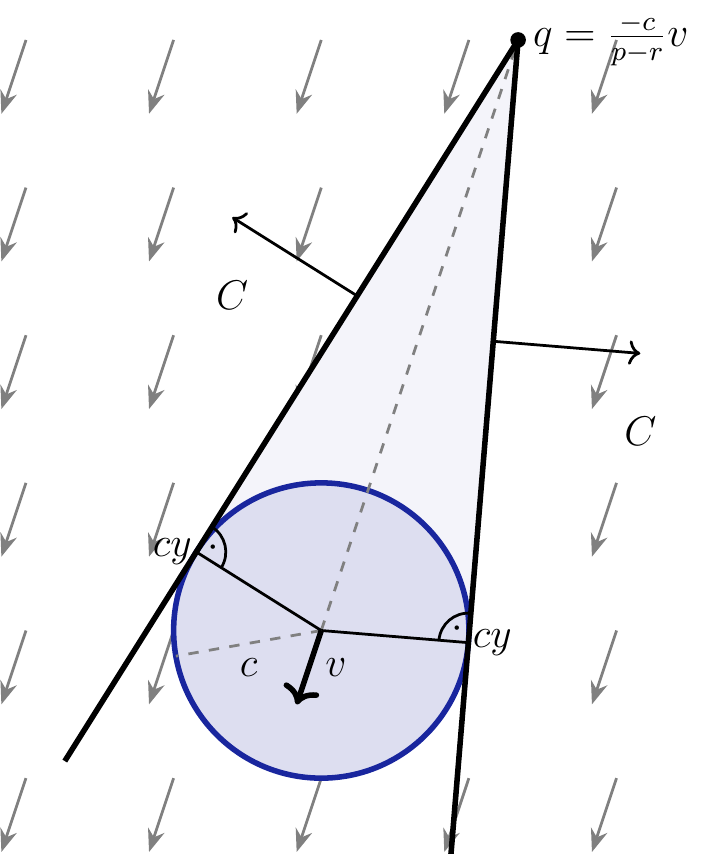}
  \caption{Local safety zones for Zeppelin obstacle parcours with same response trajectory (left) and illustration of construction of safety condition $C$ and witness for $y$ (right)}
  \label{fig:Zeppelin-safe}
\end{figure}%
Both disjuncts of $C$ prove to be differential game invariants:
\begin{sequentdeduction}[array]
\linfer[diffgameind]
  {\linfer
    {\linfer[qear]
      {\lclose}
      {\lsequent{} {\lexists{y\in B}{\lforall{z\in B}{({c q\stimes(v+py+rz)  \pm  \sqrt{\norm{q}^2-c^2} \ortho{q}\stimes (v+py+rz) \geq0
})}}}}
    }
    {\lsequent{} {\lexists{y\in B}{\lforall{z\in B}{\subst[(
    c q\stimes\D{x}  \pm  \sqrt{\norm{q}^2-c^2} \ortho{q}\stimes \D{x} \geq0
)]{\D{x}}{v+py+rz}}}}}
  }
  {\lsequent{C} {\dbox{\pdiffgame{\D{x}=v+py+rz}{y\in B}{z\in B}}{\,c q\stimes(x-q)  \pm  \sqrt{\norm{q}^2-c^2} \ortho{q}\stimes x \geq0
}}}
\end{sequentdeduction}
using the (two) tangent points of the tangent through $q$ to the circle of radius $c$ as witnesses for $y$, after scaling the tangent points by $\frac{1}{c}$ to be in $B$:
\[
y\mdefeq\frac{c}{\norm{q}^2} (c q \pm \sqrt{\norm{q}^2-c^2} \ortho{q}) / c
=\frac{1}{\norm{q}^2} (c q \pm \sqrt{\norm{q}^2-c^2} \ortho{q})
\]
The tangent through the (unscaled) tangent point $cy$ is
\(c q\stimes(x-q) \pm \sqrt{\norm{q}^2-c^2} \ortho{q}\stimes x=0\), which, indeed, goes through $cy$ and $q$. And $cy$ is on the circle of radius $c$, as $\norm{cy}=c$.

Further, $C$ itself is an invariant by monotonicity (rule \irref{M} from \rref{app:dGL-HG-axiomatization}) after splitting into both cases (by rule \irref{orl}).
By monotonicity (\irref{M}), the proof continues to prove
\[
{c q\stimes(x-q)  \pm  \sqrt{\norm{q}^2-c^2} \ortho{q}\stimes x \geq0}
\limply
\dbox{\pdiffgame{\D{x}=v+py+rz}{y\in B}{z\in B}}{\,\norm{x}^2\geq c^2}
\]
essentially using the Cauchy-Schwarz inequality for arithmetic.
The case for \(o\neq(0,0)\) results from the above proof by replacing $x$ with $x-o$ everywhere, including in $C$.
This proves the \dGL formula \rref{eq:Zeppelin} in \rref{ex:Zeppelin} with
\begin{equation}
  c q\stimes(x-o-q)  \pm  \sqrt{\norm{q}^2-c^2} \ortho{q}\stimes (x-o) \geq0
  \land c>0 \land 0\leq r<p \leq \norm{v} + r
  \label{eq:Zeppelin-invariant}
\end{equation}
as a loop invariant if \rref{eq:Zeppelin-invariant} is  assumed initially.
Otherwise, iteration (by axiom \irref{iterateb} from \rref{app:dGL-HG-axiomatization}) shows the postcondition holds after 0 iterations of the loop anyway and that \rref{eq:Zeppelin-invariant} is an invariant after the first loop iteration.
The proof would be similar without the assumption \(p \leq \norm{v} + r\) when performing a corresponding case distinction whether the wind field is outpowered or whether the propeller is mediocre.
Observe that \rref{fig:Zeppelin-safe} illustrates that the response from \rref{fig:Zeppelin} is outside the respective (green and yellow) safety zones for the obstacles.
It ends squarely within a provably unsafe zone (blue) and would, thus, continue toward a collision under a best response by the opponent.
\end{example}

\section{Soundness Proof}

The differential game invariant proof rule \irref{diffgameind} is a natural generalization of differential invariants \cite{DBLP:journals/logcom/Platzer10,DBLP:conf/cav/PlatzerC08,DBLP:journals/lmcs/Platzer12} for differential equations, also with disturbance, to differential games.
Its quantifier pattern directly corresponds to the information pattern of the differential game.

The only difficulty is its soundness proof.
The premise of \irref{diffgameind} shows that, at every point in space $x$, a local control action $y\in Y$ exists for Demon that will, for all local control actions $z\in Z$ that Angel could respond with, make the Lie-derivative \(\dder[\D{x}][f(x,y,z)]{F}\) true.
In conventional wisdom, this makes the truth-value of $F$ never change.
However, it is not particularly obvious whether those various local control actions for each $x$ at various points of the state space can be reassembled into a single coherent control signal that is measurable as a function of time and passes muster on leading the whole differential game to a successful response for each nonanticipative strategy of Angel.
Certainly, the original quantification over nonanticipative strategies and measurable control signals from the semantics is hard to capture in useful first-order proof rules.
It also took decades to justify Isaacs' equations for differential games, however innocent they may look.
Fortunately, unlike Isaacs \citeyear{Isaacs:DiffGames}, differential game invariants already have most required advances of mathematics already at their disposal.

This section proves \emph{soundness} of the differential game proof rules:
If their premise is valid, then so is their conclusion.
Proving soundness assumes the \emph{premise} (above rule bar) to be valid and considers a state $\iget[state]{\I}\in\linterpretations{\Sigma}{V}$ in which the \emph{antecedent} (left of $\limply$) of the \emph{conclusion} (below bar) is true to show that its \emph{succedent} (right of $\limply$) is true in $\iget[state]{\I}$, too.

The remainder of this section proves soundness, first of differential game refinements (\rref{sec:diffgamerefine}) then of differential game invariants (\rref{sec:differential-game-value}--\rref{sec:diffgameind}) and differential game variants (\rref{sec:diffgamefin}).
The soundness proof for \irref{diffgameind} proves the arithmetized postcondition to be a viscosity subsolution (\rref{sec:viscosity-PDE}) of the lower Isaacs partial differential equation that characterizes (\rref{sec:differential-game-Isaacs}) the lower value whose sign characterizes (\rref{sec:differential-game-value}) winning regions (\rref{sec:dGL}) independently of premature stopping (\rref{sec:frozen-games}).

\subsection{Differential Game Refinement} \label{sec:diffgamerefine}

Rule \irref{diffgamerefine} can be proved sound using the notions introduced so far.
The key is to exploit the Borel measurability and existence of semialgebraic Skolem functions to extract measurable and nonanticipative correspondence functions from its premise.
Semialgebraic functions are Borel measurable and, thus, suitable for composition (\rref{rem:preimage}).

\begin{lemma} \label{lem:semialgebraic->Borel}
  Semialgebraic functions are Borel measurable.
\end{lemma}
\begin{proofatend}
Let $f$ be a semialgebraic function.
The proof of its Borel measurability is by induction along the Borel hierarchy using \rref{rem:preimage}.

\begin{inparaenum}[\noindent\itshape 1\upshape)]
\item The preimage \(\ipreimage{f}(A)\) of any semialgebraic set $A$ under a semialgebraic function $f$ is semialgebraic \cite[Proposition 2.83]{BasuPR06}, thus, Borel measurable.

\item By \rref{rem:preimage}, the preimage \(\ipreimage{f}(\scomplement{A})=\scomplement{(\ipreimage{f}(A))}\) of the complement $\scomplement{A}$ of any set $A$ whose preimage \(\ipreimage{f}(A)\) is Borel is a complement of a Borel set and, thus, Borel.

\item By \rref{rem:preimage}, the preimage \(\ipreimage{f}(\capfold_{i\in I} A_i)=\capfold_{i\in I} \ipreimage{f}(A_i)\) of an intersection of any family of sets $A_i$ whose preimages \(\ipreimage{f}(A_i)\) are Borel, is an intersection of Borel sets and, thus, Borel.
\qedhere
\end{inparaenum}
\end{proofatend}

The proof of soundness of rule \irref{diffgamerefine} is based on composing the winning strategy from the antecedent with a semialgebraic Skolem function extracted as a witness from the local semialgebraic correspondence of the variables from the premise to obtain a winning strategy for the succedent.
This construction can be shown to preserve measurability of the resulting controls by \rref{lem:semialgebraic->Borel} and to lead to a subsumption of the differential games.
Since the semialgebraic Skolem function for $y$ from the premise is Borel, its composition can be used to show that Angel already had all control choices in the antecedent's differential game that she has in the succedent's differential game.

\begin{theorem}[Differential game refinement]
\label{thm:diffgamerefine}
Differential game refinements are sound (proof rule \irref{diffgamerefine}).
\end{theorem}
\begin{proof}
\renewcommand*{\sol}[1][]{x_{#1}}%
The formulas \(u\in U, v\in V, y\in Y, z\in Z\) only have the indicated free variables, so write \m{\imodel{\I}{Z}} for the set of values for $z$ that satisfy \(z\in Z\), etc.
The premise implies
\[
\entails \lforall{u\in U}{\lexists{y\in Y}{\lforall{z\in Z}{\lexists{v\in V}{\lforall{x}{\big(f(x,y,z)=g(x,u,v)\big)}}}}}
\]
Since this formula and its parts describe semialgebraic sets and real-closed fields have definable Skolem functions by the definable choice theorem \cite[Corollary 3.3.26]{Marker_2002}, this induces a \emph{semialgebraic}, so by \rref{lem:semialgebraic->Borel} also Borel measurable, function \(\bar{y}:\imodel{\I}{U}\to\imodel{\I}{Y}\) such that\footnote{%
Substitution of a semialgebraic function $\bar{y}(u)$ for $y$ into a formula \(F(u,y)\) of real arithmetic is definable, e.g., by
\(\lforall{y}{(y=\bar{y}(u)\limply F(u,y))}\).
The subsequent proof only needs $\bar{y}$ to be measurable, which the measurable selection theorem \cite[\S6 Theorem 6.13]{RepovsSemenov}  guarantees. 
Its inconvenience is that $\bar{y}$ cannot be syntactically inserted into the logical formulas but their mathematical equivalents would be used.
}
\begin{equation}
\entails {\lforall{z\in Z}{\lexists{v\in V}{\lforall{x}{\big(f(x,\bar{y}(u),z)=g(x,u,v)\big)}}}}
\label{eq:DGR-projection-function}
\end{equation}
The validity \rref{eq:DGR-projection-function} similarly induces a semialgebraic, thus, by \rref{lem:semialgebraic->Borel} Borel measurable function \(\bar{v}:\imodel{\I}{U}\times\imodel{\I}{Z}\to\imodel{\I}{V}\) such that
\begin{equation}
\entails {\lforall{x}{\big(f(x,\bar{y}(u),z)=g(x,u,\bar{v}(u,z))\big)}}
\label{eq:DGR-projection-function-v}
\end{equation}
To show validity of the conclusion, consider any state $\iget[state]{\I}$ in which its antecedent is true and show that its succedent is true.
That is, assume
\(\iget[state]{\I} \in \dstrategyfor[\pdiffgame{\D{x}=g(x,u,v)}{u\in U}{v\in V}]{\imodel{\I}{F}}\), i.e.\
\begin{equation}
\mforall{T{\geq}0}{\mforall{\gamma\in\stratA[U][V]{\eta}}{\mexists{u\in\controlD[U]{\eta}}{\mforall{0{\leq}\zeta{\leq}T}{}}}} \response[g]{\zeta}{\iget[state]{\I}}{u}{\gamma(u)} \in \imodel{\I}{F}
\label{eq:DGR-antecedent}
\end{equation}
It remains to be shown that
\(\iget[state]{\I} \in \dstrategyfor[\pdiffgame{\D{x}=f(x,y,z)}{y\in Y}{z\in Z}]{\imodel{\I}{F}}\),
which is
\begin{equation}
\mforall{T{\geq}0}{\mforall{\beta\in\stratA{\eta}}{\mexists{y\in\controlD{\eta}}{\mforall{0{\leq}\zeta{\leq}T}{}}}} \response[f]{\zeta}{\iget[state]{\I}}{y}{\beta(y)} \in \imodel{\I}{F}
\label{eq:DGR-succedent}
\end{equation}
Consider any \(T\geq0\) and \(\beta\in\stratA{\eta}\).
From \rref{eq:DGR-antecedent}, obtain some \(u\in\controlD[U]{\eta}\) corresponding to
\[
\gamma(u)(s) \mdefeq \bar{v}\big(u(s),\beta(\bar{y}\compose u)(s)\big)
\]
which defines a function \m{\gamma:\controlD[U]{\eta}\to\controlA[V]{\eta}},
because the composition \(\bar{y}\compose u\) of \emph{Borel} measurable function $\bar{y}$ with measurable $u$ is measurable (\rref{rem:preimage}), which makes \(\beta(\bar{y}\compose u)\) measurable and so is its composition with the Borel measurable $\bar{v}$ since $u$ was measurable to begin with.
The function $\gamma$ is also nonanticipative, so is a strategy \m{\gamma\in\stratA[U][V]{\eta}}, because
for all \(\eta\leq s\leq T\) and \(u,\hat{u}\in\controlD[U]{\eta}\):
\begin{align*}
&\text{if}~u(\tau)=\hat{u}(\tau) &&\text{for a.e.}~\eta\leq\tau\leq s\\
&\text{so}~(\bar{y}\compose u)(\tau)=(\bar{y}\compose\hat{u})(\tau) &&\text{for a.e.}~\eta\leq\tau\leq s\\
&\text{then}~\beta(\bar{y}\compose u)(\tau)=\beta(\bar{y}\compose\hat{u})(\tau) &&\text{for a.e.}~\eta\leq\tau\leq s\\
&\text{hence}~\gamma(u)(\tau)=\gamma(\hat{u})(\tau) &&\text{for a.e.}~\eta\leq\tau\leq s
\end{align*}
because \m{\beta\in\stratA{\eta}} and the compositions with Borel measurable functions $\bar{y}$ and $\bar{v}$ preserve equality a.e.\footnote{\label{foot:almost-everywhere-composition}%
If $f$ is function and \(g(\tau)=\hat{g}(\tau)\) for a.e.\ \(\tau\), then \(f(g(\tau))=f(\hat{g}(\tau))\) for a.e.\ \(\tau\), because the composition $f\compose g$ satisfies that
\(\{\tau \with f(g(\tau))\neq f(\hat{g}(\tau))\} 
\subseteq \{\tau \with g(\tau)\neq \hat{g}(\tau)\}\)
is contained in a set of measure 0.
}
Define the control $y$ for strategy $\beta$ by \(y(s) \mdefeq (\bar{y}\compose u)(s) = \bar{y}(u(s))\).
The corresponding responses \(\sol[f]\) and \(\sol[g]\) of the respective differential games satisfy
\begin{align*}
\hspace*{-0.4cm}
  \D{\sol[f]}(s) &{=} f(\sol[f](s), y(s), \beta(y)(s)) {=} f\big(\sol[f](s), (\bar{y}\compose u)(s), \beta(\bar{y}\compose u)(s)\big)\\
\hspace*{-0.2cm}
  \D{\sol[g]}(s) &{=} g(\sol[g](s), u(s), \gamma(u)(s)) {=} g\big(\sol[g](s), u(s), \bar{v}\big(u(s),\beta(\bar{y}\compose u)(s)\big)\big)
\end{align*}
which \rref{eq:DGR-projection-function-v} equates as follows:
\[
\hspace*{-0.2cm}
f\big(\sol[f](s), \bar{y}(u(s)), \beta(\bar{y}\compose u)(s)\big)
= g\big(\sol[f](s), u(s), \bar{v}\big(u(s),\beta(\bar{y}\compose u)(s)\big)\big)
\]
so that the response \(\sol[f]\) solves the same differential equation that \(\sol[g]\) does, which shows \(\sol[f]=\sol[g]\) by uniqueness (\rref{lem:response}).
Consequently, the antecedent \rref{eq:DGR-antecedent} implies \rref{eq:DGR-succedent}, which shows the conclusion of \irref{diffgamerefine} to be valid since the initial state $\iget[state]{\I}$ was arbitrary.
\end{proof}

\subsection{Values of Differential Games} \label{sec:differential-game-value}

Differential games have a unique payoff \rref{eq:payoff} for each pair of controls \(y\in\controlD{\eta}, z\in\controlA{\eta}\) and initial data $\eta,\xi$ by \rref{lem:response}.
The payoff may change when the players change their control, though.
How the players best change their controls depends on their respective opponent's control, and vice versa.
Still there is a sense in which there is an optimal payoff if both players rationally optimize their respective control.
Different choices for the informational advantage give rise to two (generally different) ways of assigning optimal payoff to a differential game: the lower and the upper value, whose signs can ultimately be related to the existence of corresponding winning strategies.

Using the response \(x(s) = \response[f]{s}{\iget[state]{\I}}{y}{\beta(y)}\) of differential game \rref{eq:diffgame} for initial condition \(x(\eta)=\iget[state]{\I}\) with time horizon $T$, the \emph{lower value} of differential game \rref{eq:diffgame} with the player for $Z$ minimizing payoff $g(x(T))$ and the player for $Y$ maximizing $g(x(T))$ captures the optimal payoff with nonanticipative strategies \m{\beta\in\stratA{\eta}} for minimizer for $Z$, i.e.\ when the minimizer for $Z$ has the informational advantage to move last \cite{DBLP:journals/tams/ElliottK74,DBLP:journals/indianam/EvansSouganidis84,BardiRP99}.
It is defined as:
\begin{align}
  \hspace*{-2pt}
  V(\eta,\iget[state]{\I}) &= 
  \inf_{\beta\in\stratA{\eta}} \sup_{y\in\controlD{\eta}}  g(\response[f]{T}{\iget[state]{\I}}{y}{\beta(y)})
  \label{eq:lower-value}
  \\
  &= \inf_{\beta\in\stratA{\eta}} \sup_{y\in\controlD{\eta}} V(\eta+\sigma, \response[f]{\eta+\sigma}{\iget[state]{\I}}{y}{\beta(y)})
  \label{eq:V-dynamic-programming}
\intertext{%
where \rref{eq:V-dynamic-programming} is the dynamic programming optimality condition \cite[Thm 3.1]{DBLP:journals/tams/ElliottK74,DBLP:journals/indianam/EvansSouganidis84} for any \(0\leq \eta<\eta+\sigma\leq T\) and $\iget[state]{\I}\in\reals^n$.
With the response \(x(s) = \response[f]{s}{\iget[state]{\I}}{\alpha(z)}{z}\),
the \emph{upper value} of differential game \rref{eq:diffgame} captures the optimal payoff when maximizer for $Y$ moves last and is defined as:}
  \hspace*{-2pt}
  U(\eta,\iget[state]{\I}) &= \sup_{\alpha\in\stratD{\eta}} \inf_{z\in\controlA{\eta}} 
  g(\response[f]{T}{\iget[state]{\I}}{\alpha(z)}{z})
  \label{eq:upper-value}
  \\
  &= \sup_{\alpha\in\stratD{\eta}} \inf_{z\in\controlA{\eta}} U(\eta+\sigma,\response[f]{\eta+\sigma}{\iget[state]{\I}}{\alpha(z)}{z}) \hspace*{-5pt}
  \label{eq:U-dynamic-programming}
\end{align}
for any \(0\leq\eta<\eta+\sigma\leq T\) and $\iget[state]{\I}\in\reals^n$, again with \rref{eq:U-dynamic-programming} being the dynamic programming optimality condition for differential games.
The lower and upper values are bounded and Lipschitz \cite[3.2]{DBLP:journals/tams/ElliottK74,DBLP:journals/indianam/EvansSouganidis84}:
\begin{theorem}[{Continuous values%
}] \label{thm:UV-bounded-Lipschitz}
  For any $T\geq0$, both $V$ and $U$ are bounded and Lipschitz (in $\eta,\xi$).
\end{theorem}

Lower and upper values are mixed infima/suprema, so it is not clear whether the optima are achievable by any concrete control or a concrete nonanticipative strategy.
The following observation relates signs of values $V$ and $U$ to the existence of strategies and controls for winning their corresponding differential game at time $T$.
Positive values, e.g., are equivalent to winning strategies winning with positive lower bounds.
\begin{lemma}[Signs of value] \label{lem:lower-value}
  Let $T\geq0$.
  \begin{enumerate}
  \item
  \label{case:V>0}%
  \(V(0,\xi)>0\) iff \(\mexists{b{>}0}{\mforall{\beta\in\stratA{\eta}}{\mexists{y\in\controlD{\eta}}{g(\response[f]{T}{\iget[state]{\I}}{y}{\beta(y)})>b}}}\).%
  \item
  \label{case:V<0}%
  \(V(0,\xi)<0\) iff \(\mexists{b{<}0}{\mexists{\beta\in\stratA{\eta}}{\mforall{y\in\controlD{\eta}}{g(\response[f]{T}{\iget[state]{\I}}{y}{\beta(y)})<b}}}\).%
  \item
  \label{case:V>=0}%
  \(V(0,\xi)\geq0\) iff \(\mforall{b{<}0}{\mforall{\beta\in\stratA{\eta}}{\mexists{y\in\controlD{\eta}}{g(\response[f]{T}{\iget[state]{\I}}{y}{\beta(y)})\geq b}}}\).%
  \item
  \label{case:V<=0}%
  \(V(0,\xi)\leq0\) iff \(\mforall{b{>}0}{\mexists{\beta\in\stratA{\eta}}{\mforall{y\in\controlD{\eta}}{g(\response[f]{T}{\iget[state]{\I}}{y}{\beta(y)})\leq b}}}\).%

  \item
  \label{case:V=0}%
  \(V(0,\xi)=0\) iff \(\mforall{b{>}0}{}\)$($
  \({\mforall{\beta\in\stratA{\eta}}{\mexists{y\in\controlD{\eta}}{g(\response[f]{T}{\iget[state]{\I}}{y}{\beta(y)})\geq -b}}}\) and\\
  \phantom{\(V(0,\xi)=0\) iff \(\mforall{b{>}0}{}\)$($ } \
  \({\mexists{\beta\in\stratA{\eta}}{\mforall{y\in\controlD{\eta}}{g(\response[f]{T}{\iget[state]{\I}}{y}{\beta(y)})\leq b}}}\)$)$.%
  {}

  Similar relations hold for the upper value, e.g.:
  \item \(U(0,\xi)>0\) iff \(\mexists{b{>}0}{\mexists{\alpha\in\stratD{\eta}}{\mforall{z\in\controlA{\eta}}{g(\response[f]{T}{\iget[state]{\I}}{\alpha(z)}{z})>b}}}\).%
  \label{case:U>0}
  
  \item \(U(0,\xi)\geq0\) iff \(\mforall{b{<}0}{\mexists{\alpha\in\stratD{\eta}}{\mforall{z\in\controlA{\eta}}{g(\response[f]{T}{\iget[state]{\I}}{\alpha(z)}{z})\geq b}}}\).%
  \label{case:U>=0}

  \end{enumerate}
\end{lemma}
\begin{proofatend}
\rref{case:V>=0} is the contrapositive of \rref{case:V<0}, which proves as follows.
If \m{V(0,\iget[state]{\I})<0}, then \(V(0,\iget[state]{\I})<2b\) for some $b<0$.
Consequently, by definition \rref{eq:lower-value}, \(\mexists{\beta\in\stratA{\eta}}{}\) such that
\[\sup_{y\in\controlD{\eta}}  g(\response[f]{T}{\iget[state]{\I}}{y}{\beta(y)}) < b\].
Hence,
\(\mexists{\beta\in\stratA{\eta}}{\mforall{y\in\controlD{\eta}}{g(\response[f]{T}{\iget[state]{\I}}{y}{\beta(y)}) < b <0}}\).
The converse direction proves accordingly, where $b,\beta$ are witnesses for the inequality
\[V(0,\iget[state]{\I})=\inf_{\beta\in\stratA{\eta}} \sup_{y\in\controlD{\eta}} g(\response[f]{T}{\iget[state]{\I}}{y}{\beta(y)}) \leq b<0\]

\noindent
\rref{case:V<=0} is the contrapositive of \rref{case:V>0}, which proves as follows.
If \(V(0,\iget[state]{\I})>0\), then \(V(0,\iget[state]{\I})>2b\) for some \(b>0\).
Thus, by \rref{eq:lower-value},
\(\mforall{\beta\in\stratA{\eta}}{\sup_{y\in\controlD{\eta}} g(\response[f]{T}{\iget[state]{\I}}{y}{\beta(y)})>2b}\).
Hence,
\(\mforall{\beta\in\stratA{\eta}}{\mexists{y\in\controlD{\eta}}{g(\response[f]{T}{\iget[state]{\I}}{y}{\beta(y)})>b}}\).

\noindent
\rref{case:V=0} combines \rref{case:V>=0} with \rref{case:V<=0}.
Cases~\ref{case:U>0} and~\ref{case:U>=0} are dual.
\end{proofatend}

Contrary to occasional misconceptions in the literature, 
\(V(0,\xi)\geq0\) does \emph{not} imply the existence of a control achieving nonnegative value for each nonanticipative strategy.
As elaborated in its consequence, \rref{case:V=0}, value \(V(0,\xi)=0\), which satisfies $\geq0$ as well, merely implies that controls can get arbitrarily close to payoff 0 without revealing a prediction about its sign.
This is problematic, because it is precisely the sign  that matters for determining whether there really is an actual winning strategy or not.

With significantly more thought, however, there is a way of rescuing the situation for the differential games of \dGL.
The following \rref{lem:V>=0} is a stronger version of \rref{case:V>=0} and shows that the simplicity of \rref{case:V>0} does, indeed, continue to hold for $\geq$ instead of $>$.
The proof is a more complex functional-analytic argument based on the results developed in the remainder of this section using Tychonoff's theorem, the Borel swap, and a continuous dependency result for Carath\'edory solutions that justifies continuous responses of differential games.
This stronger version, \rref{lem:V>=0}, makes it possible to lift differential game invariants to closed sets.
It has been stated in the literature before \cite[Lem.\,8]{DBLP:journals/tac/MitchellBT05} but only without proof or with incorrect proof.

\begin{lemma}[Closed signs of values] \label{lem:V>=0}
  Let $T\geq0$.
  Then
  \(V(0,\iget[state]{\I})\geq0\) iff
  \\
  \({\mforall{\beta\in\stratA{\eta}}{\mexists{y\in\controlD{\eta}}{g(\response[f]{T}{\iget[state]{\I}}{y}{\beta(y)})\geq 0}}}\).%
\end{lemma}
\begin{proofatend}
\begin{inparaenum}
\item[``$\mylpmi$'':]
This direction follows from \rref{case:V>=0} of \rref{lem:lower-value} as \(0\geq b\) for all \(b<0\).

\item[``$\mimply$'':]
Let \(\bar{B}\mdefeq\{\bar{b}:\controlD{\eta}\to\controlA{\eta} ~\text{Borel measurable}\}\).
Note that \(\bar{B} \subseteq \stratA{\eta}\), because the mapping \(\bar{b}(y)(s)\mdefeq \bar{b}(y(s))\) is nonanticipative, even independent of other times.
The infimum over bigger sets is smaller, thus, by \(V(0,\iget[state]{\I})\geq0\):
\begin{align*}
0 &\leq \inf_{\beta\in\stratA{\eta}} \sup_{y\in\controlD{\eta}}  g(\response[f]{T}{\iget[state]{\I}}{y}{\beta(y)})
\leq \inf_{\bar{b}\in\bar{B}} \sup_{y\in\controlD{\eta}}  g(\response[f]{T}{\iget[state]{\I}}{y}{\bar{b}(y)}))\\
&= \max_{y\in\controlD{\eta}} \min_{z\in\controlA{\eta}}  g(\response[f]{T}{\iget[state]{\I}}{y}{z}) &&\text{by \rref{lem:Borel-swap}}
\end{align*}
Hence,
\(\mexists{y\in\controlD{\eta}}\mforall{z\in\controlA{\eta}} g(\response[f]{T}{\iget[state]{\I}}{y}{z})\geq0\)
as $\min,\max$ extrema will happen for some concrete $y,z$.
Since this applies for all possible values \(\beta(y)\in\controlA{\eta}\) of any \(\beta\in\stratA{\eta}\), this implies \({\mforall{\beta\in\stratA{\eta}}{\mexists{y\in\controlD{\eta}}{g(\response[f]{T}{\iget[state]{\I}}{y}{\beta(y)})\geq 0}}}\).
The last step is the counterpart of Herbrandization for measurable functions.

It remains to see that \rref{lem:Borel-swap} is applicable.
The $[0,T]$-fold product \(\{y:[0,T]\to Y\}\) of compact space $Y$ is compact by Tychonoff's theorem \cite[\S9.5.3]{Bourbaki:top1} with respect to the product topology, i.e. the topology of pointwise convergence, i.e.\ \(y_n\to y\) for \(n\to\infty\) iff \(y_n(s)\to y(s)\) for \(n\to\infty\) for all $s$.
As pointwise limits of measurable functions are measurable \cite[9.9]{Walter:Ana2}, $\controlD{\eta}$ is a closed subset, so remains compact \cite[\S9.3.3]{Bourbaki:top1}. %
Similarly, $\controlA{\eta}$ is compact.
That \(g(\response[f]{T}{\iget[state]{\I}}{y}{z})\) is continuous (in the product topology which is the one of pointwise convergence) as a functional of $\passfunction{y}$ and $\passfunction{z}$, as required by \rref{lem:Borel-swap}, follows from \rref{lem:continuous-response} and continuity of $g$ (\rref{def:diffgame}).
\qedhere
\end{inparaenum}
\end{proofatend}

\rref{lem:V>=0} would not hold for infinite time horizon $T=\infty$ or non-compact control sets $Y,Z$.
For example, \(\pevolve{\D{x}=-x}\) converges to 0 for $T\to\infty$ without ever reaching it
and \(\pdiffgame{\D{x}=-xy}{y\in[0,\infty)}{}\) converges to 0 for $y\to\infty$ at $T=1$.
Likewise \(\pdiffgame{\D{x}=-x/y}{y\in(0,1]}{}\) converges to 0 for $y\to0$ at $T=1$.

The next lemma explains how the quantifier order seemingly swaps in the proof of \rref{lem:V>=0}.
The quantifier swap is accompanied by a change of types\footnote{%
This quantifier swap is related to the swap from \(\lforall{x}{\lexists{y}{p(x,y)}}\) in first-order logic to \(\lexists{F}{\lforall{x}{p(x,F(x))}}\) in second-order logic with a function $F$.
}, though, to move from Borel-measurable strategies (which are functions on controls) to plain controls by measurable Herbrandization.
The maximum over a compact $A$ of the minimum over a compact $B$ of a continuous function is the same as the infimum over all Borel-measurable responses \(\bar{b}:A\to B\) of the supremum over $A$ of said function.

\begin{lemma}[{Borel swap \cite{Quincampoix:SADCO}}] \label{lem:Borel-swap}
  If \(g:A\times B\to\reals\) is continuous on  compact $A,B$ and
  \(\bar{B} \mdefeq \{\bar{b}:A\to B ~\text{Borel measurable}\}\) then:
  \begin{align*}
    \max_{a\in A}\min_{b\in B} g(a,b) &= \inf_{\bar{b}\in\bar{B}} \sup_{a\in A} g(a,\bar{b}(a))\\
    \min_{a\in A}\max_{b\in B} g(a,b) &= \sup_{\bar{b}\in\bar{B}} \inf_{a\in A} g(a,\bar{b}(a))
  \end{align*}
\end{lemma}
\begin{proofatend}
Both equations imply each other by duality. 
It suffices to prove the first one.
As $A,B$ are compact and $g$ is continuous, 
\(\sup_{a\in A} \inf_{b\in B} g(a,b) = \max_{a\in A} \min_{b\in B} g(a,b)\),
because continuous functions assume their extremal values on compact sets.

\begin{inparaenum}
\item[``$\leq$'':]
Fix $\bar{b}\in\bar{B}$. For any $a\in A$:
\(g(a,\bar{b}(a)) \geq \inf_{b\in B} g(a,b)\).
Since this inequality is a weak inequality and holds for all $a\in A$, it continues to hold for the supremum:
\(\sup_{a\in A} g(a,\bar{b}(a)) \geq \sup_{a\in A} \inf_{b\in B} g(a,b)\).
Since $\bar{b}\in\bar{B}$ was arbitrary, this weak inequality continues to hold for the infimum:
\(\inf_{\bar{b}\in\bar{B}} \sup_{a\in A} g(a,\bar{b}(a)) \geq \sup_{a\in A} \inf_{b\in B} g(a,b)\).

\item[``$\geq$'':]
Fix $\varepsilon>0$. For any $a\in A$ choose a $b_a\in B$ such that
\(g(a,b_a) \leq \inf_{b\in B} g(a,b) + \frac{\varepsilon}{2}\), which is possible by the definition of infima.
Since $g$ is continuous and $B$ compact, the function \(a\mapsto \inf_{b\in B} g(a,b)\) is continuous.
As a continuous image of a compact set there, thus, is a finite open cover \(O_i\subseteq B\) and \(b_i \in A\) such that
\(g(a,b_i) \leq \inf_{b\in B} g(a,b) + \varepsilon\) for all \(a\in O_i\).
Thus,
\(g(a,b(a)) \leq \inf_{b\in B} g(a,b) + \varepsilon\) for all \(a\in A\)
for the function $b\in\bar{B}$:
\begin{equation}
b(a) \mdefeq b_i ~\text{if}~ a\in O_i \setminus \cupfold_{j<i} O_j
\label{eq:Borel-swap-piece}
\end{equation}
which is Borel measurable as a piecewise constant composition of constants on a finite number of Borel sets.
Since the above inequality holds for all $a\in A$, it continues to hold for the supremum: 
\(\sup_{a\in A} g(a,b(a)) \leq \sup_{a\in A}\inf_{b\in B} g(a,b) + \varepsilon\).
Since this upper bound holds for $b\in\bar{B}$ from \rref{eq:Borel-swap-piece}, it continues to hold for the infimum over all $\bar{b}\in\bar{B}$:
\(\inf_{\bar{b}\in\bar{B}} \sup_{a\in A} g(a,\bar{b}(a)) \leq \sup_{a\in A}\inf_{b\in B} g(a,b) + \varepsilon \to \sup_{a\in A}\inf_{b\in B} g(a,b)\) (for \(\varepsilon\to0\)).
\qedhere
\end{inparaenum}
\end{proofatend}

Continuous dependency results for Carath\'eodory solutions on their initial data are standard \cite{Walter:DGL}.
Continuity of the payoff functional in the product topology for the proof of \rref{lem:V>=0} needs continuous dependence on the right-hand side of the differential equation \rref{eq:diffgame}, though.
The following lemma shows that Carath\'eodory solutions, fortunately, also depend continuously on the right-hand side if uniformly bounded and uniformly Lipschitz.
Even Carath\'eodory solutions of differential equations are smoother than the equations in the sense that pointwise converge of the equations implies not just pointwise but even uniform convergence of the solution.
\begin{lemma}[Continuous dependence] \label{lem:uniform-Caratheodory-convergence}
  Let \(h_n{:}[\eta,T]{\times}\reals^k{\to}\reals^k\) be a sequence of functions that are measurable in $t$, uniformly $L$-Lipschitz in $x$, and with common supremum bound.
  If \(h_n\to h\) for \(n\to\infty\) pointwise and $x,x_n$ are Carath\'eodory solutions of
  \[
  \D{x}(s) = h(s,x(s)) \quad \D{x_n}(s) = h_n(s,x_n(s)) \quad x(\eta)=x_n(\eta)
  \]
  then \(x_n\to x\) uniformly on $[\eta,T]$ for \(n\to\infty\).
\end{lemma}
\begin{proofatend}
The assumptions state that
\(\norm{h_n(t,x)}\leq B\), 
\(\norm{h_n(t,x)-h_n(t,y)}\leq L\norm{x-y}\) for all $n,t,x,y$ 
and \(h_n(t,x)\to h(t,x)\) for \(n\to\infty\) and all $t,x$.
By \cite[\S10.XIX]{Walter:DGL}, $x$ and $x_n$ are Carath\'eodory solutions of their respective differential equation iff they satisfy corresponding Lebesgue integral equations:
\begin{align*}
x(t) &= x(\eta) + \int_\eta^t h(s,x(s)) ds\\
x_n(t) &= x_n(\eta) + \int_\eta^t h_n(s,x_n(s)) ds
\intertext{Consequently, they differ by}
\norm{x(t)-x_n(t)}
&= \norm{\int_\eta^t h(s,x(s))-h_n(s,x_n(s)) ds}\\
& = \norm{\int_\eta^t h(s,x(s))-h_n(s,x(s)) + h_n(s,x(s))-h_n(s,x_n(s)) ds}\\
& {\leq} \int_\eta^t \norm{h(s,x(s)){-}h_n(s,x(s))}ds {+} \int_\eta^t \norm{h_n(s,x(s)){-}h_n(s,x_n(s))} ds\\
& \leq \int_\eta^t \norm{h(s,x(s))-h_n(s,x(s))}ds + \int_\eta^t L \norm{x(s)-x_n(s)} ds\\
\intertext{Due to its norm, the first term is nondecreasing, hence Gr\"onwall's inequality implies:} %
& \norm{x(t)-x_n(t)}
\leq \underbrace{e^{\int_\eta^t Lds}}_{e^{Lt-L\eta}} \int_\eta^t \norm{h(s,x(s))-h_n(s,x(s))} ds \to 0
\end{align*}
for \(n\to\infty\) by dominated convergence \cite[9.14]{Walter:Ana2}, as \(\norm{h(s,x(s))-h_n(s,x(s))}\to0\) for all $s$ and \(\norm{h(s,x(s))-h_n(s,x(s))}\) is bounded by the Lebesgue-integrable function $2B$ since all $h_n$ are bounded by the same $B$ and so is $h$ as their pointwise limit.
\end{proofatend}

Since the responses of differential games are Carath\'eodory solutions of differential equation \rref{eq:diffgame}, \rref{lem:uniform-Caratheodory-convergence} generalizes to a continuous dependency result for differential game responses (in the product topology corresponding to pointwise convergence, which, as in \rref{lem:uniform-Caratheodory-convergence}, even leads to a uniformly convergent response).

\begin{lemma}[Continuous response]
  \label{lem:continuous-response}
  Responses of a differential game depend continuously on the controls.
  That is, if \(y_n\to y\) and \(z_n\to z\) for \(n\to\infty\) pointwise, then their responses converge 
  \(\response[f]{\argholder}{\xi}{y_n}{z_n} \to \response[f]{\argholder}{\xi}{y}{z}\) for \(n\to\infty\) uniformly on $[\eta,T]$.
\end{lemma}
\begin{proofatend}
Let \(y_n\to y\) and \(z_n\to z\) for \(n\to\infty\) pointwise.
Then the respective right-hand sides of \rref{eq:diffgame} converge: \(f(s,x,y_n(s),z_n(s))\to f(s,x,y(s),z(s))\) for \(n\to\infty\) pointwise by continuity of $f$.
The responses \(\response[f]{\argholder}{\xi}{y}{z}\) 
and \(x_n(s)\mdefeq \response[f]{s}{\xi}{y_n}{z_n}\) 
solve \rref{eq:diffgame}, which, with the abbreviations \(h(s,x) \mdefeq f(s,x,y(s),z(s))\)
and \(h_n(s,x) \mdefeq f(s,x,y_n(s),z_n(s))\), is
\begin{align*}
\D{x}(s) &= h(s,x(s)) & x(\eta)&=\xi
\\
\D{x_n}(s) &= h_n(s,x_n(s)) & x_n(\eta)&=\xi
\end{align*}
Since \(h_n(s,x) \to h(s,x)\) pointwise for \(n\to\infty\), $h_n$ and $h$ satisfy the assumptions of \rref{lem:uniform-Caratheodory-convergence} using the Lipschitz constant $L$ of $f$ in $x$ and a bound on $f$ from \rref{def:diffgame}.
Consequently, \(x_n\to x\) uniformly for \(n\to\infty\) by \rref{lem:uniform-Caratheodory-convergence}.
\end{proofatend}
Controls are usually not continuous over time, nor continuous functions of the state \cite[\S2.2]{Hajek}.
Yet, \rref{lem:continuous-response} entails that the responses still depend continuously on the controls in the product topology.
\rref{lem:continuous-response} may not hold when replacing $z_n$ by $\beta(y_n)$, because the nonanticipative strategy $\beta$ does not generally depend continuously on $y_n$, so \(\beta(y_n)\) may not converge to \(\beta(y)\) as \(y_n\to y\).
This is despite the observation:
\begin{remark}
$\stratA{\eta}$ is compact in the product topology of pointwise convergence.
\end{remark}
\begin{proofatend}
By Tychonoff's theorem \cite[\S9.5.3]{Bourbaki:top1}, also the product space \(\{\beta:\controlD{\eta}\to \controlA{\eta}\}\) is compact since \(\controlA{\eta}\) is compact (proof of \rref{lem:V>=0}).
Since pointwise limits of nonanticipative functions are nonanticipative, $\stratA{\eta}$ is a closed subset, thus, still compact \cite[\S9.3.3]{Bourbaki:top1}. %
To see that pointwise limits of nonanticipative functions are nonanticipative,
let \(\beta_n\to\beta\), i.e.\ \(\beta_n(y)\to\beta(y)\) for all $y$, which, because of the nested product topology, is \(\beta_n(y)(s)\to\beta(y)(s)\) for all $s$ and all $y$.
Let \(y(\tau)=\tilde{y}(\tau)\) for a.e. \(\eta\leq\tau\leq s\).
Then,
\(\beta_n(y)(\tau)=\beta_n(\tilde{y})(\tau)\) for a.e.\ \(\eta\leq\tau\leq s\), as \(\beta_n\in\stratA{\eta}\) for all $n$.
This equality a.e.\ is preserved for both limits \(\beta_n(y)(\tau)\to\beta(y)(\tau)\) and \(\beta_n(\tilde{y})(\tau)\to\beta(\tilde{y})(\tau)\) such that
\(\beta(y)(\tau)=\beta(\tilde{y})(\tau)\) for a.e.\ \(\eta\leq\tau\leq s\).
\end{proofatend}

Equations \rref{eq:lower-value}--\rref{eq:U-dynamic-programming} define the lower and upper values of a differential game, which, by Lemmas~\ref{lem:lower-value} and \ref{lem:V>=0} characterize the existence of winning strategies, but neither the original definitions \rref{eq:lower-value},\rref{eq:upper-value} nor the dynamic programming equations \rref{eq:V-dynamic-programming},\rref{eq:U-dynamic-programming} are computable principles \cite{DBLP:journals/jar/Platzer08} except possibly by discrete approximation, which can lead to erroneous decisions.
This is what makes the proof rules in \rref{fig:diffgameind} interesting.

\subsection{Viscosity Solutions} \label{sec:viscosity-PDE}

The lower \rref{eq:lower-value} and upper values \rref{eq:upper-value} of a differential game, whose sign characterize winning regions (Lemmas~\ref{lem:lower-value} and \ref{lem:V>=0}), can be characterized as satisfying a partial differential equation when using a suitably generalized notion of solutions that tolerates the fact that value functions are often non-differentiable, so are no classical solutions.

This section recalls viscosity solutions, which have been identified as the appropriate notion of weak solutions for Hamilton-Jacobi type partial differential equations \cite{DBLP:journals/tams/CrandallL83,Barles13}.
The presentation uses an elegant characterization of viscosity solutions with Fr\'echet sub- and superdifferentials, which capture all derivatives from below and from above a function \cite{Bressan:HJ,Barles13}.
The conceptual simplifications made possible by Fr\'echet sub/superdifferentials for differential games are also exploited in the proofs about the expressive power of \dGL (\rref{sec:differential-game-embedding}).
They are based on single-sided understandings of the \emph{gradient operator} \(D = (\Dp[x_1]{}, \dots, \Dp[x_n]{})\).
To emphasize the affected variables $x$, the gradient operator $D$ is also written as $D_x$.
Another common notation for a single variable $t$ is to write \(x_t\) instead of \(D_t x\).
\begin{definition}[Subdifferentials, superdifferentials] \label{def:subsuperdifferentials}
  Let $\Omega\subseteq\reals^n$ be open.
  The \emph{superdifferential} $\superdiff{u}(x)$ of a function \m{u:\Omega\to\reals} at $x\in\Omega$ and the \emph{subdifferential} $\subdiff{u}(x)$ of $u$ at $x$ are defined as:
  \begin{align*}
  \superdiff{u}(x) &\mdefeq \{p\in\reals^n : \limsup_{y\to x} \frac{u(y)-u(x)-p\stimes(y-x)}{\norm{y-x}} \leq0\}\\
  \subdiff{u}(x) &\mdefeq \{p\in\reals^n : \liminf_{y\to x} \frac{u(y)-u(x)-p\stimes(y-x)}{\norm{y-x}} \geq0\}
  \end{align*}
\end{definition}%
\begin{figure}[tb]
  \includegraphics[width=6cm]{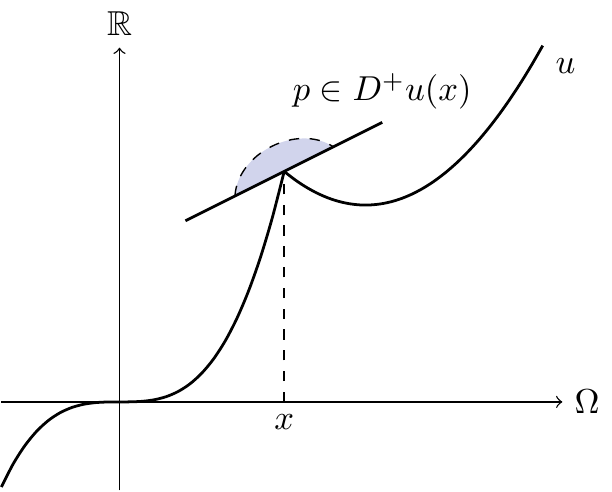}%
  \hfill%
  \includegraphics[width=6cm]{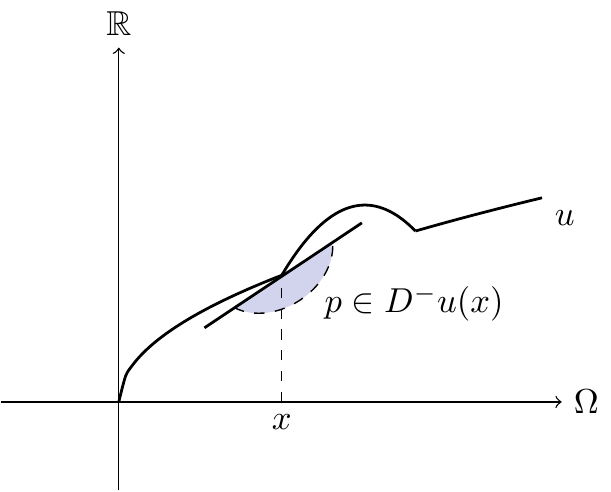}
  \caption{One of infinitely many superdifferentials $p\in\superdiff{u}(x)$ at $x$ (left) and one of many subdifferentials $p\in\subdiff{u}(x)$ at $x$ (right) for two functions}
  \label{fig:subsuperdifferentials}
\end{figure}
Both \(\superdiff{u}(x)\), \(\subdiff{u}(x)\) are closed and convex \cite[\S6.4.3]{AubinFrankowska}.
They align with geometric notions (illustrated in \rref{fig:subsuperdifferentials}) and with the classical conditions for viscosity solutions in terms of test functions as follows \cite[Thm.\,3.3]{Barles13}.

\begin{lemma}[Characterizations {\cite[Lem.\,2.2,2.5]{Bressan:HJ}}]%
\label{lem:subsuperdifferentials}%
  Let $u$ be a continuous function on an open set $\Omega\subseteq\reals^n$, i.e.\ $u\in\continuouss{\Omega}{}$. Then
  \begin{enumerate}
  \item $p\in \superdiff{u}(x)$ iff the hyperplane \(y\mapsto u(x)+p\stimes(y-x)\) is tangent from above to the graph of $u$ at $x$.
  That is:\\
  \m{u(x) + p\stimes(y-x) \geq u(y) ~ \text{for all $y$ sufficiently close to $x$}}
  \\
  Similarly, $p\in \subdiff{u}(x)$ iff the hyperplane is tangent to the graph of $u$ at $x$ from below:\\
  \m{u(x) + p\stimes(y-x) \leq u(y) ~\text{for all $y$ sufficiently close to $x$}}
  
  \item $p\in \superdiff{u}(x)$ iff there is a $v\in\continuouss[1]{\Omega}{}$, i.e.\ a continuously differentiable function \m{v:\Omega\to\reals}, such that \(Dv (x) = p\) and \(u-v\) has a local maximum\footnote{%
  The \emph{test function} $v$ can be assumed to satisfy $v(x)=u(x)$ without loss of generality in both cases. %
  Furthermore, \(u-v\) can be assumed to have a strict local maximum/minimum at $x$. %
  The property is also equivalent when using smooth $v\in\continuouss[\infty]{\Omega}{}$ instead \cite{DBLP:journals/tams/CrandallEL84}.} at $x$.
  \item $p\in \subdiff{u}(x)$ iff there is a $v\in\continuouss[1]{\Omega}{}$ such that \(Dv (x) = p\) and \(u-v\) has a local minimum at $x$.
  \item If \(\superdiff{u}(x)\neq\emptyset\) and \(\subdiff{u}(x)\neq\emptyset\), then $u$ is differentiable at $x$.
  \item If $u$ is differentiable at $x$, then \(\superdiff{u}(x)=\subdiff{u}(x)=\{Du (x)\}\) is the gradient $Du (x)$.
  \item \(\{x\in\Omega \with \superdiff{u}(x)\neq\emptyset\}\) and \(\{x\in\Omega \with \subdiff{u}(x)\neq\emptyset\}\) are dense in $\Omega$.
  \end{enumerate}
\end{lemma}

Superdifferentials of minima (e.g., \rref{fig:subsuperdifferentials}left) as well as subdifferentials of maxima (\rref{fig:subsuperdifferentials}right) are well-behaved even if differentials and gradients are ill-defined at the points of non-differentiability.
\begin{lemma} \label{lem:superdiffmin}
  The superdifferential \(\superdiff{u}(x)\) of the pointwise minimum \(u(x)\mdefeq\min_i u_i(x)\) of the functions \(u_1,\dots,u_k:\Omega\to\reals\) at \(x\in\Omega\) is the convex hull of their support
  (the case \(\subdiff{\max_i u_i}(x)\) is analogous):
  \[
  \superdiff{u}(x) = \convexhull \cupfold_{u_i(x)=u(x)} \superdiff{u_i}(x)
  \]
\end{lemma}
\begin{proofatend}
\begin{inparaenum}
\item[``$\supseteq$'':]
Let \(p\in\superdiff{u_i}(x)\) for some $i$ with \(u_i(x)=u(x)\), then \(p\in\superdiff{u}(x)\) because: 
\[
 \limsup_{y\to x} \frac{u(y)-u(x)-p\stimes(y-x)}{\norm{y-x}} 
\leq  \limsup_{y\to x} \frac{u_i(y)-u_i(x)-p\stimes(y-x)}{\norm{y-x}} \leq 0
\]
because \(u_i(x)=u(x)\) and \(u(y)\leq u_i(y)\) for all $y$ and $i$.
Since $\superdiff{u}(x)$ is convex \cite[\S6.4.3]{AubinFrankowska}, $\superdiff{u}(x)$ thus contains the convex hull of all such vectors \(p\in\superdiff{u_i}(x)\) for some $i$ with \(u_i(x)=u(x)\), which results in the right-hand side.

\item[``$\subseteq$'':]
Consider any $x$ and let \(u(x)=u_i(x)\).
Let \(p\in\superdiff{u}(x)\), i.e.\ for all $y$ close to $x$:
\[p\stimes(y-x) \leq u(y)-u(x)
= \min_j u_j(y) - u_i(x)
\leq u_i(y)-u_i(x)
\qedhere
\]
\end{inparaenum}
\end{proofatend}

Subdifferentials and superdifferentials enable a conceptually easy definition of viscosity solutions of partial differential equations: subsolutions via lower bounds for all superdifferentials and supersolution via upper bounds for all subdifferentials.
\begin{definition}[Viscosity solution] \label{def:viscosity}
  Let \(F:\Omega\times\reals\times\reals^n\to\reals\) be continuous with an open \(\Omega\subseteq\reals^n\).
  A continuous function \(u\in\continuouss{\Omega}{}\) is a \emph{viscosity solution} of the first-order \emph{partial differential equation} (PDE)
  \begin{equation}
  F(x,u(x),Du (x)) = 0
  \label{eq:PDE}
  \end{equation}
  for \emph{terminal} boundary problems
  iff it satisfies both:
  \begin{align*}
   \text{subsolution:}~&F(x,u(x),p) \geq 0 ~\text{for all}~ p\in\superdiff{u}(x) ~\text{and all}~ x\in\Omega\\
   \text{supersolution:}~&F(x,u(x),p) \leq 0 ~\text{for all}~ p\in\subdiff{u}(x) ~\text{and all}~ x\in\Omega
   \end{align*}
\end{definition}
By \rref{lem:subsuperdifferentials}, viscosity solutions are classical solutions, i.e.\ equation \rref{eq:PDE} holds for the gradient $Du(x)$, at points $x$ where they are actually differentiable.
Otherwise only the viscosity inequalities hold for the super- and subdifferentials, respectively.
PDEs are not extensional, though: \rref{eq:PDE} and \(-F(x,u,D u)=0\) can have different viscosity solutions \cite[Remark 4.4]{Bressan:HJ}, yet have the same classical solutions (if any).

The partial differential equation of relevance for differential games is the terminal\footnote{\label{foot:initial-terminal}%
Signs in \emph{terminal} value problems reverse compared to \emph{initial} value problems \cite[Chapter 10.3]{DBLP:journals/indianam/EvansSouganidis84,Evans}.
A terminal subsolution $u$ of \rref{eq:HamiltonJacobi-TVP} induces a corresponding initial subsolution \(w(t,x)= u(T-t,x)\) of \(w_t-H(T-t,x,Dw)=0, w(0,x)=g(x)\) and likewise for supersolutions.
}
evolutionary \emph{Hamilton-Jacobi equation}
  \begin{subnumcases}{\label{eq:HamiltonJacobi-TVP}}
  u_t + H(t,x,Du)=0 &\text{in}~\((0,T)\times\reals^n\) \label{eq:HamiltonJacobi-PDE}\\
  u(T,x) = g(x) &\text{in}~\(\reals^n\)
  \end{subnumcases}
with a continuous \emph{Hamiltonian} \(H:[0,T]\times\reals^n\times\reals^n\to\reals\) and
  a bounded and uniformly continuous \(g:\reals^n\to\reals\) as terminal value at $T$.
  Bounded, uniformly continuous solutions suffice here by \rref{thm:UV-bounded-Lipschitz}.
  By \rref{eq:HamiltonJacobi-PDE}, the Hamiltonian $H$ describes the time-derivative $u_t$ of $u$ but its value depends on the space-derivatives $Du=D_x u$ of $u$.

Comparison theorems \cite[Thm.\,5.3]{Bressan:HJ}\cite[Thm.\,5.2]{Barles13}\cite[\S2, Thm.\,3.3]{BardiRP99} that propagate inequalities$^{\ref{foot:initial-terminal}}$ of sub- and supersolutions on the boundary to inequalities on the whole domain are the major workhorses for PDEs.

\begin{theorem}[{Comparison}]
  \label{thm:comparison}%
  Let $u,v$ be bounded, uniformly continuous sub- and supersolutions of \rref{eq:HamiltonJacobi-PDE} and \(u \leq v\) on \(\{T\}\times\reals^n\), then \(u\leq v\) on \([0,T]\times\reals^n\)
  provided at least one of the following conditions is true:
  \begin{enumerate}

  \item \(H\) is Lipschitz, i.e.\ there is a $C$ such that
  \begin{align*}
  \norm{H(t,x,p)-H(t,x,q)} &\leq C\norm{p-q}\\
  \norm{H(t,x,p)-H(s,y,p)} &\leq C(\norm{t-s}+\norm{x-y})(1+\norm{p})
  \end{align*}
  \item $u$ is Lipschitz in $x$ uniformly in $t$, i.e.
  \(\norm{u(t,x)-u(t,y)}\leq L\norm{x-y}\) for all $x,y,t$
  \item $v$ is Lipschitz in $x$ uniformly in $t$
  \end{enumerate}
\end{theorem}

\noindent
Generalizations to bounded open $\Omega$ or to non-Lipschitz Hamiltonians satisfying modules of continuity exist \cite[Thm.\,5.2]{Barles13}.
Comparison theorems are powerful but limited to comparing subsolution $u$ and supersolutions $v$ of a single PDE.
Fortunately, they can be generalized to a monotone comparison principle for two different PDEs with related Hamiltonians.
If $u$ is growing faster than $v$ but ends below, so \(u\leq v\) at $T$, then $u$ must have been smaller all along, which remains true for viscosity solutions:
\begin{corollary}[Monotone comparison]
  \label{cor:monotone-comparison}%
  Assume one of the conditions of \rref{thm:comparison} holds or that Hamiltonian $J$ is Lipschitz.
  Let $u$ be a viscosity subsolution of \rref{eq:HamiltonJacobi-PDE}
  and let $v$ be a viscosity supersolution of
  \[
  v_t + J(t,x,Dv)=0 \quad\text{in}~(0,T)\times\reals^n
  \]
  If \(u \leq v\) on \(\{T\}\times\reals^n\) and \(H\leq J\), then \(u \leq v\) on \([0,T]\times\reals^n\).
\end{corollary}
\begin{proofatend}
$v$ is a supersolution of \(v_t + J(t,x,Dv)=0\)  if:
\begin{equation}
\tau + J(t,x,p)\leq0 \quad \mforallr{(\tau,p)\in\subdiff{v}(x)}
\label{eq:HamiltonJacobi-supersolution-J}
\end{equation}
Thus, $v$ is also a supersolution of
\(v_t + H(t,x,Dv)=0\), i.e.\
\[
\tau + H(t,x,p)\leq0 \quad \mforallr{(\tau,p)\in\subdiff{v}(x)}
\]
which follows from \rref{eq:HamiltonJacobi-supersolution-J} using \(H\leq J\).
In the case where the conditions of \rref{thm:comparison} are satisfied, this implies \(u\leq v\) by \rref{thm:comparison}.
Otherwise $J$ is Lipschitz, and the proof proceeds as follows.
First, $u$ is a subsolution of \(u_t + H(t,x,Du)=0\) if:
\begin{equation}
\tau + H(t,x,p)\geq0 \quad \mforallr{(\tau,p)\in\superdiff{u}(x)}
\label{eq:HamiltonJacobi-subsolution}
\end{equation}
Thus, $u$ is also a subsolution of
\(u_t + J(t,x,Du)=0\), i.e.\
\[
\tau + J(t,x,p)\geq0 \quad \mforallr{(\tau,p)\in\superdiff{u}(x)}
\]
which follows from \rref{eq:HamiltonJacobi-subsolution} using \(J\geq H\).
As $J$ is Lipschitz, \rref{thm:comparison} implies \(u\leq v\).
\end{proofatend}

\subsection{Isaacs Equations} \label{sec:differential-game-Isaacs}

Seminal results \cite{DBLP:journals/diffeq/Souganidis85,DBLP:journals/diffeq/BarronEJ84,DBLP:journals/indianam/EvansSouganidis84} characterize the upper and lower values of differential games as weak solutions of the Isaacs partial differential equation \cite{Isaacs:DiffGames}, which is a Hamilton-Jacobi PDE.
Isaacs intuitively identified these PDEs for differential games, which were only justified to be in correct alignment with differential games after an appropriate notion of weak solutions had been developed decades later \cite{DBLP:journals/diffeq/Souganidis85}.
For reference, \rref{app:Isaacs} provides a proof of \rref{thm:differential-game-Isaacs} for the differential games in this article.

\begin{theorem}[Isaacs PDE {\cite[Thm 4.1]{DBLP:journals/indianam/EvansSouganidis84}}] \label{thm:differential-game-Isaacs}
  The lower value $V$ from \rref{eq:lower-value} of differential game \rref{eq:diffgame} is the unique bounded, uniformly continuous viscosity solution of the lower Isaacs partial differential equation:
  \begin{align}
  &
  \begin{cases}
  V_t + H^-(t,x,D V) = 0 & (0\leq t\leq T, x\in\reals^n)\\
  V(T,x) = g(x) & (x\in\reals^n)
  \end{cases}
  \label{eq:lower-Isaacs}
  \\
  &H^-(t,x,p) = \max_{y\in Y} \min_{z\in Z} f(t,x,y,z)\stimes p \runcost{+ h(t,x,y,z)}
  \notag
  \end{align}
  The upper value $U$ from \rref{eq:upper-value} is the unique such solution of the upper Isaacs equation:
  \begin{align}
  &
  \begin{cases}
  U_t + H^+(t,x,D U) = 0 & (0\leq t\leq T, x\in\reals^n)\\
  U(T,x) = g(x) & (x\in\reals^n)
  \end{cases}
  \label{eq:upper-Isaacs}
  \\
  &H^+(t,x,p) = \min_{z\in Z} \max_{y\in Y} f(t,x,y,z)\stimes p \runcost{+ h(t,x,y,z)}
  \notag
  \end{align}
\end{theorem}
The first equation of \rref{lem:Borel-swap} illustrates the swapped quantification order of lower value \rref{eq:lower-value} compared to its Hamiltonian \rref{eq:lower-Isaacs} due to different types.
The second equation of \rref{lem:Borel-swap} similarly explains the quantifier swap from upper value  \rref{eq:upper-value} compared to its Hamiltonian \rref{eq:upper-Isaacs}.
The following result has been reported without a detailed proof, but is straightforward with the help of monotone comparisons (\rref{cor:monotone-comparison}).
\begin{corollary}[{Minimax \cite[Corollary 4.2]{DBLP:journals/indianam/EvansSouganidis84}}] \label{cor:minimax}
  \(V\leq U\) holds always.
  If
  \(H^+(t,x,p) = H^-(t,x,p)\)
  for all $0\leq t\leq T, x,p\in\reals^n$, then
  \(V=U\), i.e.\ the game has value.
\end{corollary}
\begin{proofatend}
\(H^- \leq H^+\) holds by definition, so monotone comparison \rref{cor:monotone-comparison} implies \(V\leq U\).
If \(H^- = H^+\) holds, too, then \rref{cor:monotone-comparison} also implies \(U\leq V\).
\end{proofatend}
The fact \(V\leq U\) follows from the observation that the player who chooses last is at an advantage for optimizing the resulting value.
The assumption \(H^+(t,x,p) = H^-(t,x,p)\) corresponds to the Hamiltonians being independent of the order of choice, which implies \(V=U\) so that the order of choice in the whole differential game is irrelevant.

If one fixed finite time horizon $T$ were sufficient, \rref{thm:differential-game-Isaacs} could be used with \rref{lem:lower-value} and~\ref{lem:V>=0} to answer the question of the existence of winning strategies for this time horizon $T$ \emph{if} its PDE \rref{eq:lower-Isaacs} can be solved.
Numerical approximation schemes for \rref{eq:lower-Isaacs} are, indeed, one way of answering game questions, but they are inherently subject to discrete approximation errors that may lead to erroneous decisions that have not yet been overcome \cite{DBLP:journals/tac/MitchellBT05}.
By contrast, \irref{diffgameind} provides a sound way of proving the existence of winning strategies even for all time horizons.
Yet, proving proof rule \irref{diffgameind} itself to be sound requires more effort, which the subsequent sections pursue.

\subsection{Frozen Games} \label{sec:frozen-games}

For a fixed time horizon $T$, the results from \rref {sec:differential-game-value} and \ref{sec:differential-game-Isaacs} characterize winning regions of differential games by signs of the solutions of their corresponding PDEs, but that only helps if Angel commits to a fixed time horizon $T$ and maximal stopping time \(\zeta=T\) by advance notice.
Lifting these characterizations to the case where Angel decides to stop early by choosing \(\zeta<T\) is possible by repeating the above analysis for minimum payoff games \cite{Serea02}.
This leads to less convenient PDEs, though.

A more modular way is to add an extra freeze input \cite{DBLP:journals/tac/MitchellBT05} for Angel player, which she can control to slow down or lock the system in place.
A freeze factor \(c\in[0,1]\) multiplies the differential game and is under Angel's control, which will keep the system unmodified (\(c=1\)), in stasis (\(c=0\)), or in slow motion (\(0<c<1\)).
Angel controls both time $\zeta$ and freeze factor $c$.
So the frozen system does not actually need early stopping, because she can freeze it with control choice \(c=0\) instead in order to lock its state in place. The quantifier for stopping time $\zeta$ in \rref{def:DHG-semantics} is, thus, irrelevant:

\begin{lemma}[Frozen values] \label{lem:nostopping-whenfrozen}
  For any atomically open formula $F$ it is the case that:\\
  \m{\iget[state]{\I} \in \dstrategyfor[\pdiffgame{\D{x}=c\genDG{x}{y}{z}}{y\in Y}{z\in Z\land c\in[0,1]}]{\imodel{\I}{F}}}
  iff its lower value satisfies \(V(0,\iget[state]{\I})>0\) for all \(T\geq0\)
  with the arithmetization \(g\mdefeq\arithmetize{F}\) as payoff.
  Accordingly for atomically closed $F$.
\end{lemma}
\begin{proofatend}
\begin{inparaenum}
\item[``$\mimply$'':] by \rref{case:V>0} of \rref{lem:lower-value} (or \rref{lem:V>=0}) for $b\mdefeq g(\xi)/2$ by \rref{lem:arithmetize}.

\item[``$\mylpmi$'':]
By \rref{lem:arithmetize} and \rref{case:V>0} of \rref{lem:lower-value}, it only remains to be shown that stopping time $\zeta$ can always be instantiated to time horizon $T$ in \rref{def:DHG-semantics} for this game.
Instead of stopping prematurely at \m{\zeta<T}, Angel can set her extra freeze input $c$ to 0 at time $\zeta$, because \(c=0\) will already keep the value of $x$ constant.
The step function
\[
c(t) \mdefeq
\begin{cases}
  1 &\text{if}~t\leq\zeta\\
  0 &\text{if}~t>\zeta
\end{cases}
\]
required as the appropriate control input for the freeze factor to freeze at time $\zeta$ is Borel measurable.
\end{inparaenum}
The proof for closed $F$ uses \rref{lem:V>=0} instead of \rref{lem:lower-value}.
\end{proofatend}
This result exploits that durations of differential games are unobservable except when adding a clock \(\D{t}=1\) to the differential game to measure the progress of time, which would be frozen along with $x$, though, so that freezing is unobservable by the players.

When replacing all differential games with their frozen version, \rref{lem:nostopping-whenfrozen} implies that the results from \rref {sec:differential-game-value}--\ref{sec:differential-game-Isaacs} characterize their winning regions by signs of values.
That approach works flawlessly.
It is more efficient to exploit the structure of the frozen game to remove the freeze factor $c$ with a minimal change in the Hamiltonian.

\begin{lemma}[Frozen Isaacs] \label{lem:frozen-Isaacs}
  According to \rref{thm:differential-game-Isaacs}, let $H^-$ and $H^+$ be the Hamiltonians for the lower and upper values of 
  \begin{equation}
  \pdiffgame{\D{x}=f(x,y,z)}{y\in Y}{z\in Z}
  \label{eq:pdiffgame}
  \end{equation}
  Then the lower and upper values of the \emph{frozen differential game}
  \begin{equation}
  \pdiffgame{\D{x}=c f(x,y,z)}{y\in Y}{z\in Z\land c\in[0,1]}
  \label{eq:pdiffgame-frozen}
  \end{equation}
  respect the lower \rref{eq:lower-Isaacs} and upper \rref{eq:upper-Isaacs} Isaacs equations with Hamiltonians $J^-$ and $J^+$:
  \begin{align}
  J^-(t,x,p) &= \min(0, H^-(t,x,p)) \label{eq:lower-frozen-Isaacs-Hamiltonian}\\
  J^+(t,x,p) &= \min(0, H^+(t,x,p)) \notag
  \end{align}
\end{lemma}
\begin{proofatend}
By \rref{thm:differential-game-Isaacs}, the lower value and upper value of \rref{eq:pdiffgame-frozen} satisfy the lower and upper Isaacs equations with the following Hamiltonians, which simplify as shown:
\begin{align*}
J^-(t,x,p) &= \max_{y\in Y} \min_{z\in Z} \min_{c\in[0,1]} c f(t,x,y,z)\stimes p
= \max_{y\in Y} \min_{z\in Z} \min(0, f(t,x,y,z)\stimes p)\\
 &= \min(0, \max_{y\in Y} \min_{z\in Z} f(t,x,y,z)\stimes p)
 \\
J^+(t,x,p) &= \min_{c\in[0,1]} \min_{z\in Z} \max_{y\in Y} c f(t,x,y,z)\stimes p
 = \min(0, \min_{z\in Z} \max_{y\in Y} f(t,x,y,z)\stimes p)
\end{align*}
since $\min$ and $\max$ are mutually distributive.
By monotone comparison \rref{cor:monotone-comparison}, those transformations do not change the solution.
\end{proofatend}
When starting both differential games in the same initial state with the same payoff, the lower and upper value of \rref{eq:pdiffgame}, thus, dominate the lower and upper value, respectively, of \rref{eq:pdiffgame-frozen}, by \rref{cor:monotone-comparison}, because \(J^-(t,x,p) \leq H^-(t,x,p)\) and \(J^+(t,x,p) \leq H^+(t,x,p)\).
The freeze input $c$ can be removed from the Hamiltonian by \rref{lem:frozen-Isaacs}.
Indeed, $c$ does not ever need to be introduced into differential games explicitly either, because both winning regions are identical, based on \cite{DBLP:journals/tac/MitchellBT05}:

\begin{lemma}[Superfluous freezing] \label{lem:freezing}
  Let \(X\subseteq\linterpretations{\Sigma}{V}\).
  Then
  \begin{align*}
   \dstrategyfor[\pdiffgame{\D{x}=\genDG{x}{y}{z}}{y\in Y}{z\in Z}]{X}
  &=
  \dstrategyfor[\pdiffgame{\D{x}=c\genDG{x}{y}{z}}{y\in Y}{z\in Z\land c\in[0,1]}]{X}
  \\
  \strategyfor[\pdiffgame{\D{x}=\genDG{x}{y}{z}}{y\in Y}{z\in Z}]{X}
  &=
  \strategyfor[\pdiffgame{\D{x}=c\genDG{x}{y}{z}}{y\in Y}{z\in Z\land c\in[0,1]}]{X}
  \end{align*}
\end{lemma}
\begin{proofatend}
By \rref{thm:dGL-determined}, the equations imply each other, so the proof only considers the first equation.

\begin{inparaenum}
\item[\noindent``$\supseteq$'':]
This inclusion follows from the soundness of this \irref{diffgamerefine} proof step (\rref{thm:diffgamerefine}):
\[
\hspace*{-0.3cm}
  {\linfer[diffgamerefine]
    {\lforall{u{\in} Y}{\lexists{y{\in} Y}{\lforall{z{\in} Z}{\lexists{v{\in} Z,c{\in}[0{,}1]}{\lforall{x}{(f(x,y,z)=c f(x,u,v))}}}}}}
    {\dbox{\pdiffgame{\D{x}{=}c f(x,u,v)}{u{\in} Y}{v{\in} Z{\land} c{\in}[0{,}1]}}{F} \limply \dbox{\pdiffgame{\D{x}{=}f(x,y,z)}{y{\in} Y}{z{\in} Z}}{F}}
  }
\]
whose premise proves using the witnesses \(y\mdefeq u, v\mdefeq z, c\mdefeq1\).

\item[\noindent``$\subseteq$'':]
This direction has been shown elsewhere \cite[Corollary 5]{DBLP:journals/tac/MitchellBT05}.
The idea of the proof is as follows.
The addition of $c$ does not affect the game behavior or capabilities, because its only effect is a time-dilation, and time-invariant differential equations \m{\D{x}=f(x,y,z)} are invariant under time rescaling if time itself is unobservable.
Which it is, unless the differential game includes a clock \m{\D{t}=1}, in which case that clock will be frozen when $c<1$ as well, because its frozen counterpart is \m{\D{t}=c}.
\qedhere
\end{inparaenum}
\end{proofatend}

In a similar way, \emph{differential games restricted to evolution domains} are expressible by the \emph{dual freezing game} that gives another freeze factor $b$ to Demon with which he can suspend the evolution should Angel ever try to leave the domain.
A differential game with evolution domain $\ivr$ has to always remain in $\ivr$ and stop before leaving it.
But only Angel is in control of time. She might try to leave $\ivr$ temporarily and sneak back before Demon notices, which is forbidden.
Adding the dual freeze factor $b$ to the game gives Demon the option of slowing the game down and challenging Angel to demonstrate it still is in $\ivr$.
Ensuring that Demon does not slow the game down just to prevent Angel's progress to victory is possible by exploiting hybrid games around it:
\[
\pupdate{\pumod{t}{x_0}};~ (\pdiffgame{\D{x}=b\genDG{x}{y}{z}\syssep\D{t}=1}{y\in Y\land b\in[0,1]}{z\in Z});
~
\ptest{\ivr};~
\dtest{(x_0=t)}
\]
This reduction assumes that the (vectorial) differential game \(\D{x}=\genDG{x}{y}{z}\) already contains a deterministic clock \(\D{x_0}=1\) and adds a separate unfrozen absolute clock \m{\D{t}=1} starting from the same value after the assignment \m{\pupdate{\pumod{t}{x_0}}}.
To slow the system down, Demon needs to choose \m{b<1} on a set of non-zero measure (otherwise \m{b=1} a.e., which has no effect).
That will slow down the frozen \m{\D{x_0}=b} compared to the unfrozen \m{\D{t}=1}, so that Demon fails his time-synchronicity dual test \(\dtest{(x_0=t)}\) and loses.
Unless he correctly points out that the system left the domain $\ivr$, in which case Angel will lose because she fails her test \(\ptest{\ivr}\) first.
Even though Demon has no influence on Angel's choice of time $\zeta$, he can choose \(b=0\) to force the game into stasis any time.
He just needs to use that power wisely or else he will lose the game for false allegations.
This is the differential game analogue of the ``there and back again game'' for differential equations with evolution domains \cite{DBLP:journals/tocl/Platzer15}.
Differential hybrid games, thus, enable simpler differential games compared to incorporating state constraints directly into a differential game by special-purpose techniques \cite{DBLP:conf/optta/Rapaport98}.
\subsection{Soundness of Differential Game Invariants} \label{sec:diffgameind}

This completes the background results required for proving soundness of differential game invariant rule \irref{diffgameind}.
The soundness proof proves the arithmetized postcondition (\rref{lem:arithmetize}), from an initial state that satisfies it, to be a time-independent viscosity subsolution (\rref{sec:viscosity-PDE}) for all time horizons of the lower Isaacs PDE \rref{eq:lower-Isaacs} that characterizes (\rref{sec:differential-game-Isaacs}) the lower value \rref{eq:lower-value} whose sign, in turn, characterizes (\rref{sec:differential-game-value}) differential game winning regions (\rref{sec:dGL}) even for premature stopping (\rref{sec:frozen-games}).

\begin{theorem}[Differential game invariants] \label{thm:diffgameind}
  Differential game invariants are sound (proof rule \irref{diffgameind}).
\end{theorem}
\begin{proof}
To prove soundness, assume the premise to be valid and assume the antecedent of the conclusion true in a state $\iget[state]{\I}$, written \(\imodels{\I}{F}\) as notation for \(\iget[state]{\I} \in \imodel{\I}{F}\):
\begin{align}
&\entails  \lexists{y\in Y}{\lforall{z \in Z}{ \dder[\D{x}][f(x,y,z)]{F} }}
\label{eq:diffgameind-premise}
\\
&\imodels{\I}{F}
\label{eq:diffgameind-antecedent}
\end{align}
To make the proof easier to follow, the proof first considers the case where $F$ is an atomic formula even if that follows from subsequent cases.

\begin{inparaenum}[\noindent\itshape 1\upshape)]
\item \label{case:post>0}
Consider the case where $F$ is of the form \m{F\mequiv(g>0)} for a (smooth) term $g$.
Then the (valid) premise \rref{eq:diffgameind-premise} of rule \irref{diffgameind} specializes to \m{\lexists{y \in Y}{\lforall{z \in Z}{\subst[\D{(g>0)}]{\D{x}}{f(x,y,z)}}}}, which is
\begin{equation}
\entails
\lexists{y \in Y}{\lforall{z \in Z}{(\subst[\der{g}]{\D{x}}{f(x,y,z)}\geq0)}}
\label{eq:diffgameind-premise-gt}
\end{equation}
When $\xi\in\linterpretations{\Sigma}{V}$ is a state, adopt the usual mathematical liberties of writing \m{g(\xi)} for the value \m{\ivaluation{\I[\xi]}{g}} of term $g$ in state $\xi\in\linterpretations{\Sigma}{V}$ to simplify notation substantially and keep it closer to standard mathematical practice.
Similarly for \(f(x,y,z)\), since it will be clear from the context whether the term \(f(x,y,z)\) or its value is being referred to.
If all the $x,y,z$ are variables, \(f(x,y,z)\) is a term. If, instead, $\xi,\eta,\zeta$ are all (vectors of) reals, \(f(\xi,\eta,\zeta)\) refers to the corresponding value \(\ivaluation{\imodif[state]{\imodif[state]{\imodif[state]{\I[\sigma]}{x}{\xi}}{y}{\eta}}{z}{\zeta}}{f(x,y,z)}\) (for any state $\sigma$ since $x,y,z$ are all free variables).
For variable $x$ and values $\eta,\zeta$, the mixed case $f(x,\eta,\zeta)$ evaluates in state $\iget[state]{\I}$ to the value \(f(\iget[state]{\I},\eta,\zeta)=\ivaluation{\imodif[state]{\imodif[state]{\I}{y}{\eta}}{z}{\zeta}}{f(x,y,z)}\), which will be used sparingly to avoid confusion.

Consider any time horizon $T>0$ of Angel's choosing.
The case $T=0$ follows from \rref{eq:diffgameind-antecedent}.
The proof first shows that the time-invariant extension function \(\bar{g}(t,x) \mdefeq g(x)\) is a subsolution of the lower Isaacs equation \rref{eq:lower-Isaacs} with its unique solution $V$ (\rref{thm:differential-game-Isaacs}), which, by \rref{thm:comparison}, implies \m{\bar{g}\leq V}, because both functions coincide at time $T$.
\begin{sublemma} \label{slem:ming-subsolution1}
\(\bar{g}(t,x) \mdefeq g(x)\) for smooth term $g$ is a subsolution of lower Isaacs equation \rref{eq:lower-Isaacs}.
\end{sublemma}
\begin{subproof}
Since \m{\bar{g}} is smooth, it, by \rref{lem:subsuperdifferentials}, is a subsolution iff it satisfies the subsolution inequality classically at every $(\eta,\xi)$:
\begin{equation}
\underbrace{\bar{g}_t(\eta,\xi)}_0 + \underbrace{\max_{y\in Y} \min_{z\in Z} f(\xi,y,z)\stimes D_x \bar{g} (\eta,\xi)}_{\geq0} \geq0
\label{eq:gbar-subsolution}
\end{equation}
which holds since $\bar{g}$ is time-invariant so its time-derivative $\bar{g}_t$ vanishes and by premise \rref{eq:diffgameind-premise-gt}, recalling that
\(f(\xi,y,z) \stimes D_x \bar{g} (\eta,\xi) = \ivaluation{\I[\xi]}{\subst[\der{g}]{\D{x}}{f(x,y,z)}}\) for all $\zeta,y,z$ by \rref{lem:Lie-relation}, so that \rref{eq:diffgameind-premise-gt} implies:%
\[
\mexists{y\in Y}{\mforall{z\in Z} { f(\xi,y,z) \stimes D_x \bar{g} (\eta,\xi) = \ivaluation{\I[\xi]}{\subst[\der{g}]{\D{x}}{f(x,y,z)}} \geq 0 }}
\]
By \rref{eq:gbar-subsolution}, $\bar{g}$ is a subsolution of \rref{eq:lower-Isaacs}, so
\m{g(\xi)=\bar{g}(\eta,\xi)\leq V(\eta,\xi)} for all $\eta,\xi$ by \rref{thm:comparison}, which is applicable because $V$ is bounded and uniformly continuous by \rref{thm:differential-game-Isaacs}, and Lipschitz in $x,t$ by \rref{thm:UV-bounded-Lipschitz}, thus, Lipschitz in $x$ \emph{uniformly in $t$} since $t$ is bounded by $T$ so the maximum Lipschitz bound among \(t\in[0,T]\) is finite.
For applicability of \rref{thm:comparison}, note that $g$ and $\bar{g}$ are bounded and Lipschitz by \rref{def:diffgame} (using the relevant domain from \rref{lem:well-defined}) and, thus, uniformly continuous by \rref{sec:preliminaries}.
\end{subproof}

So \m{V(\eta,\iget[state]{\I})\geq g(\iget[state]{\I})>0} for all $\eta$ and any initial state $\iget[state]{\I}$ satisfying antecedent \(F \mequiv (g>0)\) of the conclusion of \irref{diffgameind}, i.e.\ \rref{eq:diffgameind-antecedent} which is \(g(\iget[state]{\I})>0\).
Hence, \rref{case:V>0} of \rref{lem:lower-value} implies
\[
\mforall{\beta\in\stratA{\eta}} {\mexists{y\in\controlD{\eta}}{g(\response[f]{T}{\iget[state]{\I}}{y}{\beta(y)}) > 0}}
\]
This shows that Demon can achieve \m{g>0} from any initial state $\iget[state]{\I[\xi]}$ where \m{g>0} holds provided that Angel decides to evolve the full duration $T$, which she does not have to.

Since \m{g(\iget[state]{\I[\xi]})\leq V(t,\iget[state]{\I[\xi]})} is a \emph{time-independent} lower bound for all times $t$ and all time horizons $T$, Angel cannot achieve a lower value of $g$ by stopping earlier either:
\begin{sublemma} \label{slem:payoff-lower-bound-transferral}
  If the payoff $g$ is a time-independent subsolution of \rref{eq:lower-Isaacs} with \(g(\iget[state]{\I})>0\), then
\begin{equation}
  \iget[state]{\I} \in \dstrategyfor[\pdiffgame{\D{x}=\genDG{x}{y}{z}}{y\in Y}{z\in Z}]{\imodel{\I}{g>0}}
  \label{eq:payoff-lower-bound-true}
\end{equation}
  The case \(g(\iget[state]{\I})\geq0\) is accordingly with \m{\imodel{\I}{g\geq0}} instead.
\end{sublemma}
\begin{subproof}
Since $g$ satisfies its own boundary condition, $g$ is a subsolution of \rref{eq:lower-Isaacs} iff:
\[\underbrace{\tau}_0 + H^-(t,x,p)\geq0 ~\text{for all}~ (\tau,p)\in\superdiff{g}(t,x) ~\text{and all}~t,x\]
In particular, $g$ is also a subsolution of the frozen lower Isaacs equation with Hamiltonian \rref{eq:lower-frozen-Isaacs-Hamiltonian} from \rref{lem:frozen-Isaacs}, since $0+0\geq0$ already holds:
\[
\underbrace{\tau}_0 + \min(0,H^-(t,x,p))\geq0 ~\text{for}~ (\tau,p)\in\superdiff{g}(t,x) ~\text{and}~t,x
\]
Thus, the lower value of the frozen game \rref{eq:pdiffgame-frozen} has lower bound $g$.
By \rref{lem:nostopping-whenfrozen}, the frozen game does not need any premature stopping, so that \rref{lem:lower-value} proves
\[\iget[state]{\I} \in \dstrategyfor[\pdiffgame{\D{x}=cf(x,y,z)}{y\in Y}{z\in Z\land c\in[0,1]}]{\imodel{\I}{g>0}}\]
since $T\geq0$ was arbitrary.
The ``$\supseteq$'' inclusion of \rref{lem:freezing}, which was proved directly by differential game refinement \rref{thm:diffgamerefine}, implies \rref{eq:payoff-lower-bound-true}, concluding the subproof.
\end{subproof}

An alternative to \rref{slem:payoff-lower-bound-transferral} proceeds directly without freezing: $g$ also is a subsolution of the Isaacs equation for infimum cost \cite{Serea02}
\begin{align*}
&\min(v_t(t,x) + h^-(x,v(t,x),D_x v(t,x)),\, g(x)-v(t,x)) = 0\\
&h^-(x,r,p) =
\begin{cases}
 \max_{y\in Y}\min_{z\in Z} f(x,y,z)\stimes p &\text{if}~g(x)\leq r\\
\infty \hfil &\text{if}~g(x)>r
\end{cases}
\end{align*}
which the infimum cost value over time solves
\[
  v(\eta,\xi) = 
  \inf_{\beta\in\stratA{\eta}} \sup_{y\in\controlD{\eta}}  \min_{t\leq T} g(\response[f]{t}{\iget[state]{\I}}{y}{\beta(y)})
\]
because the choice of $g(x)$ for $v(t,x)$ satisfies
\[\min(\tau + h^-(x,\bar{g}(t,x),p),\, g(x)-\bar{g}(t,x)) \geq 0 ~\mforall{(\tau,p)\in\superdiff{\bar{g}}(x)}\]
\rref{lem:lower-value} carries over to this infimum cost value $v$ with an extra \(\mexists{t\leq T}\)for time, so that \(0<g(\iget[state]{\I}) \leq v(0,\iget[state]{\I})\) directly shows
\[
\iget[state]{\I} \in \dstrategyfor[\pdiffgame{\D{x}=\genDG{x}{y}{z}}{y\in Y}{z\in Z}]{\imodel{\I}{g>0}}
\]
Even if this alternative proof also works for time-dependent $g$, its downside is that its PDE assumes a convex image of $f$ under $Y$ and $Z$ to facilitate discontinuous games \cite{Serea02}, which are not needed, because hybrid games cover discontinuous change.

\item \label{case:post>=0}
Consider the case where $F$ is of the form \m{F\mequiv(g\geq0)} for a (smooth) term $g$.
Then the proof proceeds as in \rref{case:post>0}, since the premise of \irref{diffgameind} is still \rref{eq:diffgameind-premise-gt}, because \m{\der{g\geq0}} is equivalent to \m{\der{g>0}} by \rref{def:derivation}.
In that case, the antecedent \rref{eq:diffgameind-antecedent} only implies \m{\imodels{\I}{g\geq0}} in the initial state $\iget[state]{\I}$, thus, \m{V(\eta,\iget[state]{\I})\geq g(\xi)\geq0} for all $\eta$.
Yet, then \rref{lem:V>=0} instead of \rref{lem:lower-value} still implies
\[
\mforall{\beta\in\stratA{\eta}} {\mexists{y\in\controlD{\eta}}{g(\response[f]{T}{\iget[state]{\I}}{y}{\beta(y)}) \geq 0}}
\]
which shows the conclusion of rule \irref{diffgameind} by \rref{slem:payoff-lower-bound-transferral}.

\item \label{case:open-post}
Consider the case where $F$ is atomically open.
By congruence, it is enough to consider the case where $F$ is normalized by \((a<b) \mequiv (b-a>0)\) so that it is built with $\land,\lor$ from formulas of the form \m{g_i>0} for polynomials $g_i$.
Let \m{I\mdefeq\{i \with g_i(\iget[state]{\I})>0\}\neq\emptyset} be the set of all indices $i$ whose atomic formula $g_i>0$ is true in the initial state $\iget[state]{\I}$.
As a replacement for the previous \rref{slem:ming-subsolution1},
the subsequent \rref{slem:ming-subsolution} shows that the time-invariant minimum \(\bar{g}(t,x) \mdefeq \min_{i\in I} g_i(x)\) of the involved continuously differentiable $g_i$ is still a subsolution of the lower Isaacs equation even if $\bar{g}$ itself is not smooth.
\begin{sublemma} \label{slem:ming-subsolution}
$\bar{g}$ is a subsolution of the lower Isaacs equation \rref{eq:lower-Isaacs}.
\end{sublemma}
Since $\bar{g}$ is time-invariant, validity of the conclusion of \irref{diffgameind} follows with \rref{slem:payoff-lower-bound-transferral} from \rref{slem:ming-subsolution} like \rref{case:post>0} followed from \rref{slem:ming-subsolution1} using the observation that the combination of subformulas of $F$ that were true initially will remain true using \rref{lem:arithmetize}, because
\(0<\bar{g}(\eta,\iget[state]{\I})\leq V(\eta,\iget[state]{\I})\) for all $\eta$ and any initial state $\iget[state]{\I}$ that satisfies the antecedent \rref{eq:diffgameind-antecedent}.

\begin{subproof}[Subproof of \rref{slem:ming-subsolution}]
The proof idea from \rref{slem:ming-subsolution1} no longer works, because $\bar{g}$ has no differentials at points where the minimum switches from one term $g_i$ to another $g_j$ unless their differentials happen to align.
This proof uses superdifferentials instead.

The premise \rref{eq:diffgameind-premise} in this case yields
\begin{equation}
\entails \lforall{x}{\lexists{y\in Y}{\lforall{z\in Z}{\landfold_i \dder[\D{x}][f(x,y,z)]{(g_i\geq0)}}}}
\label{eq:diffgameind-premise-closed}
\end{equation}
which, in mathematical metalanguage corresponds to
\begin{equation}
\mforall{x}{\mexists{y\in Y}{\mforall{z\in Z}{f(x,y,z)\stimes\sder{g_i}(x)\geq0~\text{for all}~i}}}
\label{eq:mdiffgameind-premise-closed}
\end{equation}
because \(\dder[\D{x}][f(x,y,z)]{(g_i\geq0)}\)
is \(\subst[\der{g_i}]{\D{x}}{f(x,y,z)}\geq0\),
which is \(\sdder[\D{x}][f(x,y,z)]{g_i}(x) \geq0\) by \rref{lem:Lie-relation}.
Proving that $\bar{g}$ is a subsolution of lower Isaacs PDE \rref{eq:lower-Isaacs} requires proving
\begin{equation}
\underbrace{\tau}_0 + \max_{y\in Y} \min_{z\in Z} f(x,y,z)\stimes p \geq0
~ \text{for all \((\tau,p)\in\superdiff{\bar{g}}(t,x)\) and all $x\in\linterpretations{\Sigma}{V}$}
\label{eq:ming-subsolution}
\end{equation}
Since $\bar{g}$ is time-invariant, it is differentiable by $t$ with derivative $0$ everywhere, hence the time component of its superdifferential coincides with the classical gradient $0$ by \rref{lem:subsuperdifferentials}.
Dropping time from the notation simplifies condition \rref{eq:ming-subsolution} to:
\begin{equation}
\max_{y\in Y} \min_{z\in Z} f(x,y,z)\stimes p \geq0
~\text{for all}~p\in\superdiff{\bar{g}}(x)
~\text{and all}~x\in\linterpretations{\Sigma}{V}
\label{eq:ming-subsolution-droptime}
\end{equation}
Rephrasing \rref{eq:ming-subsolution-droptime}, it remains to show:
\[
\mforall{x\in\linterpretations{\Sigma}{V}} \mforall{p\in\superdiff{\bar{g}(x)}} \mexists{y\in Y} \mforall{z\in Z} f(x,y,z)\stimes p \geq0
\]
For any $x$, using the corresponding \(y\in Y\) from \rref{eq:mdiffgameind-premise-closed}, this is true for all $z\in Z$ and all $i$:
\[
f(x,y,z)\stimes\superdiff{g_i}(x)\geq0
~\text{that is}~ f(x,y,z)\stimes\sder{g_i}(x)\geq0
\]
because \(\superdiff{g_i}(x)=\{\sder{g_i}(x)\}\) by \rref{lem:subsuperdifferentials}.
According to \rref{lem:superdiffmin}, all convex generators of \(\superdiff{\bar{g}}\), thus, satisfy that same property, which continues to hold for convex combinations, since
for any \(p,q\in\superdiff{\bar{g}}(x)\) and \(\lambda\in[0,1]\):
\[
f(x,y,z)\stimes(\lambda p+(1-\lambda)q)
=
\lambda f(x,y,z)\stimes  p + (1-\lambda) f(x,y,z)\stimes q \geq0
\]
This proves \rref{eq:ming-subsolution-droptime}, so that $\bar{g}$ is a subsolution of \rref{eq:lower-Isaacs}.
\end{subproof}

\item \label{case:closed-post}
The case where $F$ is atomically closed  proceeds as in \rref{case:open-post}.
The premise of \irref{diffgameind} is equivalent to the premise in \rref{case:open-post}, because \m{\der{a\geq b}} and \m{\der{a>b}} are equivalent by \rref{def:derivation}.
The additional thought for closed sets is as for \rref{case:post>=0}.
Since \m{\bar{g}} is a subsolution, the same combination of subformulas of $F$ that were true initially will remain true.

\item \label{case:any-post}
The case where $F$ is any first-order formula (quantifier-free by quantifier elimination \cite{tarski_decisionalgebra51}) reduces to \rref{case:closed-post}.
By congruence, it is enough to consider the case where $F$ is normalized by \((a<b) \mequiv (b-a>0)\) and \((a=b) \mequiv (a-b\geq0 \land b-a\geq0)\) etc.\ so that it is built with $\land,\lor$ from formulas of the form \m{g_i\geq0} or \m{h_j>0}.
Replace every strict inequality \(h_j>0\) in $F$ that is true in the initial state $\iget[state]{\I}$ by a new weak inequality \(g_j\geq0\) with the term \(g_j \mdefeq h_j-a_j\), which is still true in the initial state when choosing the constant \(a_j \mdefeq h_j(\iget[state]{\I})>0\).
Replace every strict inequality \(h_j>0\) that is not true in the initial state $\iget[state]{\I}$ by \(-1\geq0\).
The resulting formula $G$ is closed, true in the initial state, and, if Demon has a strategy to achieve $G$, then, by monotonicity of winning regions (rule \irref{M} in \rref{app:dGL-HG-axiomatization}), he also has a strategy to achieve the original $F$, because \(\entails G\limply F\).
\rref{case:closed-post} implies that Demon can achieve $G$, because the premise of \irref{diffgameind} that \rref{case:closed-post} assumes for $G$ is implied by the premise for $F$ since \(\der{h_j>0}\) is equivalent to \(\der{h_j\geq0}\) which is equivalent to \(\der{h_j-a_j\geq0}\) by \rref{def:derivation} as \(\der{a_j}=0\) for constant $a_j$.
Likewise \(\der{-1\geq0} \mequiv (0\geq0)\) is trivially implied.
\end{inparaenum}

This concludes the proof of \rref{thm:diffgameind}, demonstrating soundness for rule \irref{diffgameind}.
\end{proof}

\subsection{Soundness of Differential Game Variants} \label{sec:diffgamefin}

Since rule \irref{diffgamefin} settles for a conservative quantifier pattern, the soundness proof for rule \irref{diffgameind} can be adapted more easily to prove soundness of rule \irref{diffgamefin} as well.

\begin{theorem}[Differential game variants] \label{thm:diffgamefin}
  Differential game variants are sound (proof rule \irref{diffgamefin}).
\end{theorem}
\begin{proof}
Let \(\imodels{\I}{g<0}\), i.e.\ \(g(\iget[state]{\I})<0\), otherwise Angel wins by choosing \(T=0\).
The proof follows the same principle as the proof of \rref{thm:diffgameind} by using the duality \rref{thm:dGL-determined}, since the same game is played in \m{\dbox{\pdiffgame{\D{x}=\genDG{x}{y}{z}}{y\in Y}{z\in Z}}{}} and \m{\ddiamond{\pdiffgame{\D{x}=\genDG{x}{y}{z}}{y\in Y}{z\in Z}}{}} with the same partition of control advantage and information just from another player's perspective.
To facilitate proof reuse, rule \irref{diffgamefin} uses a conservative information pattern, so that the duality allows to swap player controls and consider
\(\dbox{\pdiffgame{\D{x}=\genDG{x}{y}{z}}{z\in Z}{y\in Y}}{(g\geq0)}\).
This formula cannot be expected to be true, since the initial state does not need to satisfy \m{g\geq0}, for Angel would stop right away then.
Yet, the study of its value will still prove to be informative and, in particular, reuse the proof of \rref{thm:diffgameind}.
The only, but critical, change is that \irref{diffgamefin} does not assume the postcondition to hold in the beginning and, instead, requires a proof that it will finally be reached.
This leads to the following variation on the choice of the subsolution for the comparison theorem.
Let $\varepsilon\in\reals$ be the value whose existence the premise shows.
For postcondition formula \m{g\geq0}, consider 
\(\bar{g}(t,x) \mdefeq g(x)+\frac{\varepsilon}{2}(T-t)\).
This $\bar{g}$ is smooth, so, by \rref{lem:subsuperdifferentials}, a subsolution of the lower Isaacs equation \rref{eq:lower-Isaacs} iff:
\begin{equation}
\underbrace{\bar{g}_t(t,x)}_{-\frac{\varepsilon}{2}} + \underbrace{\max_{y\in Y} \min_{z\in Z} f(x,y,z)\stimes D_x \bar{g} (t,x)}_{\geq\varepsilon} \geq0
\label{eq:barg-eps-subsolution}
\end{equation}
which again holds by premise using \rref{lem:Lie-relation} if its assumption \(g(x)\leq0\) holds.
The left-hand side of \rref{eq:barg-eps-subsolution} is $\geq\varepsilon-\frac{\varepsilon}{2}>0$ on the closed set \m{\imodel{\I}{g\leq0}}, and is a continuous function, so it continues to be ${>}0$ on sufficiently small neighborhoods of \m{\imodel{\I}{g\leq0}}.
Thus, the proof in \rref{thm:diffgameind} continues to work when restricting the domain to a sufficiently small open neighborhood $\mathcal{U}$ of \m{\imodel{\I}{g\leq0}}.
Since \m{\bar{g}(\eta,\iget[state]{\I})\leq V(\eta,\iget[state]{\I})} follows from \rref{thm:comparison} as in \rref{thm:diffgameind}, \rref{lem:V>=0} implies the conclusion of \irref{diffgamefin} if \(0\leq V(0,\iget[state]{\I})\), which will happen for large enough time horizons $T$ according to the definition of $\bar{g}$.
In particular, \(0<\bar{g}(\eta,\iget[state]{\I})\leq V(\eta,\iget[state]{\I})\) when $T$ is sufficiently large, e.g. \(T>-\frac{2}{\varepsilon}g(\iget[state]{\I})>0\), which is under Angel's control.
The existence of a (unique) solution of such a duration $T$ follows from Perron's existence theorem for Hamilton-Jacobi PDEs \cite[Thm.\,7.1]{Barles13}.

For this time horizon $T$, by \rref{lem:lower-value}, player Demon of the flipped game, who plays for Angel's controls of the original differential game, will ultimately be in a state where \m{g\geq0}, if he just happens to be lucky that such a long time is played and the game does not stop prematurely, so \(\zeta=T\) is chosen, in which case \rref{eq:lower-Isaacs} characterizes the lower value (otherwise the frozen Isaacs Hamiltonian \rref{eq:lower-frozen-Isaacs-Hamiltonian} would apply so that \rref{eq:barg-eps-subsolution} stops holding).
For the original differential game, in which Angel is in charge of controlling the time, this means that she can win \m{g\geq0} by just playing long enough, which is under her control, and by limiting herself to \(\zeta=T\), which is her choice, too.
Since \(0<\bar{g}(\eta,\iget[state]{\I})\leq V(\eta,\iget[state]{\I})\) for all $\eta$ for this $T$, and \(g(\response{s}{\iget[state]{\I}}{y}{\beta(y)})\) is continuous in $s$ (\rref{lem:response}), Angel will win into \(\imodel{\I}{g\geq0}\) before leaving the open neighborhood $\mathcal{U}$ of \m{\imodel{\I}{g\leq0}}.
\end{proof}

It is of apparent significance for the soundness of rule \irref{diffgamefin} that the lower bound $\varepsilon$ holds for all $x$, not just that there is an $\varepsilon$ for every $x$.
Otherwise, the progress might converge (long) before \m{g\geq0} is reached.
It is also possible to prove soundness of \irref{diffgamefin} based on the soundness proof of rule \irref{diffgamerefine}.
That works by replacing the Hamiltonian in \rref{eq:barg-eps-subsolution} by a uniformly continuous continuation $J$ (which exists by Tietze \cite[2.19]{Walter:Ana2}) to the full space, which agrees with the Hamiltonian from \rref{eq:barg-eps-subsolution} on the open neighborhood $\mathcal{U}$ of \m{\imodel{\I}{g\leq0}} and shares the same lower bound $\varepsilon$, but globally.
The proof then uses soundness of the $\ddiamond{\cdot}{}$ dual of rule \irref{diffgamerefine} to show that the original game has a winning strategy since the game corresponding to $J$ has a winning strategy for \m{g\geq0}.
The only additional thought is that it is enough to restrict the premise of \irref{diffgamerefine} to the set of $x$ that can occur during the game starting from $\iget[state]{\I}$, which is where the values of the original game and the one for the Hamiltonian $J$ coincide by Tietze \cite[2.19]{Walter:Ana2}.
\section{Differential Game Embeddings} \label{sec:differential-game-embedding}

The previous sections have immersed differential games within hybrid games to form differential hybrid games and studied how their properties can be proved.
This is a useful approach in practice.
The alternative is to understand how differential games relate to (non-differential) hybrid games from a theoretical perspective.
The logic \(\dGL_{\text{HG}}\) from \cite{DBLP:journals/tocl/Platzer15} is differential game logic of \emph{hybrid games}, which is the fragment of \dGL that has no differential games, except differential equations \(\pevolve{\D{x}=f(x)}\).
The logic \(\dGL_{\text{DG}}\) is differential game logic of \emph{differential games} in which all games are of the form \(\pdiffgame{\D{x}=\genDG{x}{y}{z}}{y\in Y}{z\in Z}\).
For emphasis, \(\dGL_{\text{DHG}}\) is differential game logic \dGL for full \emph{differential hybrid games} in which all operators of \rref{def:dGL-DHG} are allowed.
Tracing in \dGL the characterizations developed here \emph{only for open or closed postconditions} gives:

\begin{theorem}[Differential game characterization] \label{thm:DG<=HG}
  Differential games are hybrid games, i.e.\
  \(\dGL_{\textup{DHG}}\) and \(\dGL_{\textup{HG}}\) are equally expressive:\footnote{Logic $\mathcal{B}$ is at least as expressive as $\mathcal{A}$, written \(\mathcal{A}\leq \mathcal{B}\) if every formula of $\mathcal{A}$ can be expressed by an equivalent formula of $\mathcal{B}$.
  Further, \(\mathcal{A}\mequiv \mathcal{B}\) if \(\mathcal{A}\leq \mathcal{B}\) and \(\mathcal{B}\leq \mathcal{A}\).
  And \(\mathcal{A}<\mathcal{B}\) if \(\mathcal{A}\leq \mathcal{B}\) but not \(\mathcal{B}<\mathcal{A}\).
} \(\dGL_{\textup{HG}} \mequiv \dGL_{\textup{DHG}}\).
\end{theorem}
\begin{proofatend}
This proof uses the encoding results in \rref{app:differential-game-embedding}.
The nontrivial direction \(\dGL_{\text{DHG}} \leq \dGL_{\text{HG}}\) can be shown by a careful analysis of the constructions involved in characterizing differential games. 
The original definition of differential games and their behavior in terms of nonanticipative strategies and measurable functions of control input does not lead to a characterization without facing substantial challenges of having to characterize higher-order quantification in large-cardinality function classes.
The indirect characterization of a differential game in terms of its Isaacs PDEs proves to be more useful.
Using expressiveness results for the base logic \cite{DBLP:conf/lics/Platzer12b,DBLP:journals/tocl/Platzer15}, it is enough to consider the new differential game cases
\begin{equation}
  \dbox{\pdiffgame{\D{x}=\genDG{x}{y}{z}}{y\in Y}{z\in Z}}{F}
  \label{eq:pdiffgamebox}
\end{equation}
and
\(\ddiamond{\pdiffgame{\D{x}=\genDG{x}{y}{z}}{y\in Y}{z\in Z}}{F}\).
By \rref{thm:dGL-determined} it is enough to consider just \rref{eq:pdiffgamebox}.

\newcommand{\Ix}{\I[\xi]}%

\begin{inparaenum}[\noindent\itshape 1\upshape)]
\item \label{case:open-embedding}
Consider the case where $F$ is atomically open.
By \rref{lem:freezing}, \rref{eq:pdiffgamebox} is equivalent to its frozen analogon\footnote{%
As in \rref{slem:payoff-lower-bound-transferral} of \rref{thm:diffgameind}, \rref{thm:DG<=HG} can alternatively be proved using the (more involved) Isaacs PDEs for infimum cost \cite{Serea02} instead of the frozen differential game from \rref{lem:freezing}.
}
\(\dbox{\pdiffgame{\D{x}=c\genDG{x}{y}{z}}{y\in Y}{z\in Z\land c\in[0,1]}}{F}\).
By \rref{lem:nostopping-whenfrozen}, the latter needs no premature stopping and is true in a state $\iget[state]{\Ix}$ iff \(V(0,\iget[state]{\Ix})>0\) for all $T\geq0$, using the realization \(g\mdefeq\arithmetize{F}\) as payoff (using \rref{lem:well-defined}).
By \rref{lem:frozen-Isaacs}, $V$ satisfies the lower Isaacs equation \rref{eq:lower-Isaacs} with the Hamiltonian \rref{eq:lower-frozen-Isaacs-Hamiltonian}.
Thus, \rref{eq:pdiffgamebox}
is true in $\iget[state]{\Ix}$ iff \(V(0,\iget[state]{\Ix})>0\) for all $T\geq0$.
The quantification over time horizon $T$ is definable in \dGL.
So is the condition whether the state characterized by a variable vector $x$ satisfies \(V(0,x)>0\) provided that $V$ and its evaluation can be characterized, which is what \rref{cor:continuoussGoedel} in \rref{app:differential-game-embedding} shows since $V$ is continuous by \rref{thm:UV-bounded-Lipschitz}.
By \rref{thm:differential-game-Isaacs}, $V$ is the unique bounded, uniformly continuous viscosity solution of the lower Isaacs equation \rref{eq:lower-Isaacs} with the Hamiltonian \rref{eq:lower-frozen-Isaacs-Hamiltonian} from \rref{lem:frozen-Isaacs}.
Boundedness and uniform continuity are characterizable with first-order logic over the reals (see \rref{sec:preliminaries}), since evaluation of $V$ is by \rref{cor:continuoussGoedel}.
The terminal condition,
\(V(T,x) = g(x)\) for all $x$, is characterizable by quantification and evaluation by \rref{cor:continuoussGoedel}.
The fact that $V$ solves the (by \rref{lem:frozen-Isaacs} frozen) Isaacs equation
\[
  V_t + \max_{y\in Y} \min_{z\in Z} \min_{c\in[0,1]} c\itimes f(x,y,z)\stimes DV = 0
\]
can be characterized by the definable condition
\[
  \tau + \max_{y\in Y} \min_{z\in Z} \min_{c\in[0,1]} c\itimes f(x,y,z)\stimes p \geq 0
  ~~
  \mforallr{(\tau,p)\in\superdiff{V}(t,x)}
\]
provided quantification over all superdifferentials \((\tau,p)\) is definable.
Once that succeeds, the argument is the same to characterize that $V$ is a viscosity supersolution.

Dropping the time coordinates $t,\tau$ from the notation for simplicity,
\rref{def:subsuperdifferentials} implies that
\(p\in\superdiff{V}(x)\) iff
\[
\limsup_{y\to x} \frac{V(y)-V(x)-p\stimes(y-x)}{\norm{y-x}} \leq0
\]
which is characterizable as follows.
Abbreviating \(({V(y)-V(x)-p\stimes(y-x)})/{\norm{y-x}}\) by $h(y)$, which is definable, the above can be rephrased equivalently using:
\[
\limsup_{y\to x} h(y) = \inf_{\varepsilon>0} \sup \{h(y) \with 0<\norm{y-x}<\varepsilon\}
\]
Whether, for an $\varepsilon>0$, the inner $\sup$ has value $s$ is definable as a least upper bound:
\[
\lforall{y}{(0<\norm{y-x}<\varepsilon \limply s\geq h(y))}
\land\\
\lforall{b}{(\lforall{y}{(0<\norm{y-x}<\varepsilon \limply b\geq h(y))} \limply s\leq b)}
\]
A similar first-order formula characterizes the value of the outer $\inf$ in terms of this $s$.

As viscosity supersolution conditions are definable correspondingly, the set of states where \dGL formula \rref{eq:pdiffgamebox} is true is characterizable in \dGL without differential games.

\item \label{case:closed-embedding}
The case of closed formulas $F$ is accordingly, using the criterion \rref{case:V>=0} from \rref{lem:lower-value} or \rref{lem:V>=0} instead.
\end{inparaenum}
Note that the elegant layered approach for hybrid systems logic \dL, which is based on lifting complete approaches for open formulas to closed and then to arbitrary formulas \cite{DBLP:conf/lics/Platzer12b}, does not work for \dGL, because the Barcan axiom of \dL that it rests on is not sound for \dGL \cite{DBLP:journals/tocl/Platzer15}.
\end{proofatend}

\begin{theorem}[Expressive power] \label{thm:DG<HG}
  Differential games are strictly less expressive than hybrid games, i.e.\ \(\dGL_{\textup{DG}}\) is less expressive than \(\dGL_{\textup{HG}}\):
  \(\dGL_{\textup{DG}} < \dGL_{\textup{HG}}\).
\end{theorem}
\begin{proofatend}
The proof of \rref{thm:DG<=HG} does not rely on special features of hybrid games but continues to work when characterizing differential games in \dL, the corresponding logic of hybrid systems \cite{DBLP:conf/lics/Platzer12b}.
The result, thus, follows since \cite[Thm.\,5.3]{DBLP:journals/tocl/Platzer15} shows that hybrid systems are strictly less expressive than hybrid games.
\end{proofatend}

This is surprising, because the contrary holds for hybrid systems.
Hybrid systems are equivalently reducible to differential equations \cite{DBLP:conf/lics/Platzer12b}.
\rref{thm:DG<HG} shows that this situation reverses for differential games versus hybrid games.

\section{Related Work}

A general overview of the long history of differential games since their conception \cite{Isaacs:DiffGames,Friedman} and breakthroughs of their viscosity understanding \cite{DBLP:journals/diffeq/Souganidis85,DBLP:journals/diffeq/BarronEJ84,DBLP:journals/indianam/EvansSouganidis84} is discussed in the literature \cite{BardiRP99}.
The related work discussion here focuses on differential games as they relate to hybrid games.
Hybrid games themselves \cite{DBLP:conf/hybrid/NerodeRY96,DBLP:conf/concur/HenzingerHM99,DBLP:journals/IEEE/TomlinLS00,DBLP:journals/jopttapp/DharmattiDR06,DBLP:journals/corr/abs-0911-4833,DBLP:journals/tcs/VladimerouPVD11} are discussed elsewhere \cite{DBLP:journals/tocl/Platzer15}.
See \cite{DBLP:journals/tac/MitchellBT05} for a helpful overview of hybrid systems verification and how Lagrangian verification relates to Eulerian verification.
The relationship of differential games to (robust) control theory \cite{BlanchiniMilano}, which is interesting yet limited to piecewise continuous controls or linearity assumptions and, thus, does not give a sound approach for differential games with measurable inputs, is elaborated in the literature \cite{BardiRP99,DBLP:journals/tac/MitchellBT05,Cardaliaguet2007}.

Previous techniques for differential games revolve around numerically solving the PDEs that they induce \cite{BardiRP99,Isaacs:DiffGames,DBLP:journals/tac/MitchellBT05}, corresponding viability theory formulations \cite{Cardaliaguet2007}, or classically by passing to the limit when considering lower and upper time-discrete approximations with strategies changing at finitely many points \cite{Friedman}.
The latter cannot be implemented and its theoretical understanding has been revolutionized by the invention of viscosity solutions \cite{DBLP:journals/tams/CrandallL83,DBLP:journals/indianam/EvansSouganidis84,Barles13}.
The former are interesting but do not yield proofs, because PDEs are highly nontrivial to solve.
A number of subtle soundness issues have been reported \cite{DBLP:journals/tac/MitchellBT05} for different shapes of the respective sets.
These numerical approximation schemes cannot provide correctness guarantees, because their error is unbounded.
Unlike results in \dGL, numerical PDE solutions are also only for a fixed time horizon $T$.

Viability theory provides geometric notions for differential games with a robustness margin \cite{DBLP:journals/siamco/Aubin91,DBLP:journals/tcs/KohnNRY95,Cardaliaguet2007,DBLP:conf/hybrid/BayenCS07}.
Its algorithms converge to the correct answer only in the limit \cite{CardaliaguetQSP99}.
They give internally consistent answers on the discretization grid, but errors off the grid and outside the reachable set are still unbounded, and inherent discontinuities of value functions from viability theory complicate the numerics \cite{DBLP:journals/tac/MitchellBT05}.
Viability has been considered for hybrid systems \cite{DBLP:journals/tac/GaoLQ07} with affine dynamics and convexity assumptions and only if no input influences the discrete state, which goes against the spirit of hybridness.
To simplify the problem, continuous controls or strategies \cite{DBLP:journals/gamerev/SaintPierre04} or convex control images with affine dynamics are assumed \cite{DBLP:journals/siamco/Cardaliaguet96}; see \cite{Cardaliaguet2007,BardiRP99}.

Special-purpose cases for differential games where players play some limited form of hybrid input have been considered \cite{DBLP:journals/jopttapp/DharmattiDR06}.
There is an argument to be made in favor of more modular designs such as \dGL, where discrete and continuous games are integrated side-by-side as first-class citizens in a modular programming language, as opposed to all intermingled in one differential game.
The fact that systems become easier when understood as combinations of simpler elements has already been equally paramount for the success of hybrid systems \cite{DBLP:conf/lics/Platzer12a}.

Differential game logic for hybrid games \emph{without differential games} has been introduced along with an axiomatization and theoretical analysis in prior work \cite{DBLP:journals/tocl/Platzer15}.
Here, differential games are integrated modularly into hybrid games.
The focus is on the characterization, study, and proof principles of differential games, leveraging compositionality principles of logic to cover differential hybrid games.
This leads to the first sound proof approach for differential games and combinations with hybrid games.
The resulting differential hybrid games are the only games that support both discrete and continuous state change with adversarial dynamical interaction during both.

\section{Conclusions and Future Work}

Differential game invariants, variants, and refinements are simple and sound inductive proof techniques for differential games, which embed compositionally into differential hybrid games.
The primary challenge was their soundness proof, which uses superdifferentials to show that their arithmetizations are viscosity subsolutions of the Isaacs PDE characterizing the lower value whose sign characterizes winning regions.

Induction can be defined in different ways for differential equations such as checking near boundaries with sufficient care to prevent soundness issues.
Similar flexibility is expected for differential games, for which differential game invariants are the first induction principle.
In passing, \rref{thm:diffgameind} showed soundness of superdifferentials for differential invariants, which will be investigated in future work.
Recent advances in generating differential invariants should generalize to differential game invariants.

\appendix

\printproofs

\section{Encoding Proofs for Embedding} \label{app:differential-game-embedding}

\newcommand{\ati}[3]{#1^{(#2)}_{#3}}%
The hybrid systems logic \dL \cite{DBLP:conf/lics/Platzer12b} is the sublogic of \dGL that has differential equations but neither duality $\pdual{}$ nor differential games.
By $A^B$ denote the set of functions \(B\to A\).
The proof of \rref{thm:DG<=HG} is based on the following encoding results.

\begin{lemma}[{$\reals$-G{\"o}del encoding \cite[Lem.\,4]{DBLP:journals/jar/Platzer08}}] \label{lem:realGodel}
  \index{R-Goedel@$\reals$-G{\"o}del}
  The logical relation $\text{at}(Z,n,j,z)$, which holds iff~$Z$ is a real number that represents a G{\"o}del encoding of a sequence of~$n$ real numbers with real value~$z$ at position~$j$ (for \m{1\leq j\leq m}), is definable in \dL.
  For a formula \m{\mapply{\phi}{z}} abbreviate \m{\lexists{z}{(\text{\normalfont{at}}(Z,n,j,z)\land\mapply{\phi}{z})}} by~\m{\mapply{\phi}{\ati{Z}{n}{j}}}.
  \index{_Znj_@$Z^{(n)}_j$}%
\end{lemma}%
\begin{corollary}[Infinite $\reals$-G\"odel encoding]
  The bijection 
  \m{\reals\isoto\reals^{\naturals}} is characterizable in \dL by a formula \m{\text{at}(Z,\infty,j,z)}, which holds iff $Z$ is a real number that represents a G\"odel encoding of an $\omega$-infinite sequence of real numbers with real value $z$ at position $j$.
For a formula $\mapply{\phi}{z}$, abbreviate \m{\lexists{z}{(\text{\normalfont{at}}(Z,\infty,j,z)\land\mapply{\phi}{z})}} by~\m{\mapply{\phi}{\ati{Z}{\infty}{j}}}.
\end{corollary}
\begin{proofatend}
\(\text{at}(Z,\infty,j,z)\) is definable by repeated unpairing using \rref{lem:realGodel}
\[
\ddiamond{\prepeat{(j:=j-1; Z:=\ati{Z}{2}{2})}}{(j=0\land z=\ati{Z}{2}{1})}
\]
The use of an abbreviation formula like $\ati{Z}{2}{2}$ inside a modality is definable (most easily in rich-test \dL).
\end{proofatend}
\begin{corollary} \label{cor:rational}
  The bijections \m{\naturals\isoto\rationals} and \m{\reals\isoto\reals^{\rationals}} are characterizable in \dL.
\end{corollary}
\begin{proofatend}
\dL can define the formula \m{\text{rat}(n,p,q)}, which holds iff $\frac{p}{q}$ is the $n$-th rational number (in some arbitrary but fixed definable order):
\[
\text{rat}(n,p,q) \lbisubjunct p=\ati{n}{2}{1}\land q=\ati{n}{2}{2} \land q>0
\qedhere
\]
\end{proofatend}

\begin{corollary} \label{cor:continuoussGoedel}
  The bijection \m{\reals\isoto\continuouss{\reals}{\reals}} from reals to the continuous functions on the reals is characterizable in \dL.
\end{corollary}
\begin{proofatend}
Since continuous functions are uniquely defined by their values on the rationals $\rationals$,
\rref{cor:rational} shows that \dL can characterize the bijection by
\[
\lforall{\varepsilon{>}0}{\lexists{\delta{>}0}{\lforall[\rationals]{\frac{p}{q}}{\lforall[\naturals]{n}{}}}}
 ~\Big(\text{rat}(n,p,q) \land \abs{x-\frac{p}{q}}<\delta\limply
\abs{z-\ati{F}{\infty}{n}}<\varepsilon)
\Big)
\]
Observe that the enumeration of \m{\frac{p}{q}} from \rref{cor:rational} enumerates identical fractions with different denominators repeatedly, which would allow for the definition of inconsistent $F$ that give different values at \m{\frac{p}{q}} and \m{\frac{2p}{2q}}.
This is easily overcome, e.g., by skipping fractions that cancel, which can be checked by divisibility or Euclid's $\gcd$ algorithm, which are both definable with programs in \dL. 
\end{proofatend}
\section{Non-differential Hybrid Game Axiomatization} \label{app:dGL-HG-axiomatization}

For reference, \rref{fig:dGL} shows a sound and complete axiomatization from prior work \cite{DBLP:journals/tocl/Platzer15} for the case of differential game logic for hybrid games with differential equations but without differential games.
The axiomatization is designed on top of the first-order Hilbert calculus (modus ponens, uniform substitution, and Bernays' $\forall$-generalization) with all instances of valid formulas of first-order logic as axioms, including first-order real arithmetic.
The only change of \rref{fig:dGL} compared to prior work \cite{DBLP:journals/tocl/Platzer15} is the use of dualization to convert $\ddiamond{\argholder}{}$ axioms into $\dbox{\argholder}{}$ axioms.
This is a cosmetic change to make it easier for the reader to appreciate how differential game invariants (proof rule \irref{diffgameind}) integrate seamlessly into the proof calculus for the other operators of differential hybrid games.

\newcommand{\solf}{y}%

\begin{figure}[bth]
  \renewcommand*{\irrulename}[1]{\text{#1}}%
  \renewcommand{\linferenceRuleNameSeparation}{~~}
  \newdimen\linferenceRulehskipamount%
  \linferenceRulehskipamount=1mm%
  \newdimen\lcalculuscollectionvskipamount%
  \lcalculuscollectionvskipamount=0.1em%
  \begin{calculuscollections}{\columnwidth}
    \begin{calculus}
      \cinferenceRule[diamond|$\didia{\cdot}$]{diamond axiom}
      {\linferenceRule[equiv]
        {\lnot\dbox{\alpha}{\lnot\phi}}
        {\ddiamond{\alpha}{\phi}}
      }
      {}
      \cinferenceRule[assignb|$\dibox{:=}$]{assignment / substitution axiom}
      {\linferenceRule[equiv]
        {\mapply[x]{\phi}{\theta}}
        {\dbox{\pupdate{\umod{x}{\theta}}}{\mapply[x]{\phi}{x}}}
      }
      {}%
      \cinferenceRule[evolveb|$\dibox{'}$]{evolve}
      {\linferenceRule[equiv]
        {\lforall{t{\geq}0}{\dbox{\pupdate{\pumod{x}{\solf(t)}}}{\phi}}\hspace{1cm}}
        {\dbox{\pevolve{\D{x}=f(x)}}{\phi}}
      }{\m{\D{\solf}(t)=f(\solf)}}%
      \cinferenceRule[testb|$\dibox{?}$]{test}
      {\linferenceRule[equiv]
        {(\ivr \limply \phi)}
        {\dbox{\ptest{\ivr}}{\phi}}
      }{}
      \cinferenceRule[choiceb|$\dibox{\cup}$]{axiom of nondeterministic choice}
      {\linferenceRule[equiv]
        {\dbox{\alpha}{\phi} \land \dbox{\beta}{\phi}}
        {\dbox{\pchoice{\alpha}{\beta}}{\phi}}
      }{}
      \cinferenceRule[composeb|$\dibox{{;}}$]{composition}
      {\linferenceRule[equiv]
        {\dbox{\alpha}{\dbox{\beta}{\phi}}}
        {\dbox{\alpha;\beta}{\phi}}
      }{}
      \cinferenceRule[iterateb|$\dibox{{}^*}$]{iteration/repeat unwind pre-fixpoint, even fixpoint}
      {\linferenceRule[lpmi]
        {\phi \land \dbox{\alpha}{\dbox{\prepeat{\alpha}}{\phi}}}
        {\dbox{\prepeat{\alpha}}{\phi}}
      }{}%
      \cinferenceRule[dualb|$\dibox{{^d}}$]{dual}
      {\linferenceRule[equiv]
        {\lnot\dbox{\alpha}{\lnot\phi}}
        {\dbox{\pdual{\alpha}}{\phi}}
      }{}
      \cinferenceRule[M|M]{$\dbox{}{}$ monotonic / $\dbox{}{}$-generalization} %
      {\linferenceRule[formula]
        {\phi\limply\psi}
        {\dbox{\alpha}{\phi}\limply\dbox{\alpha}{\psi}}
      }{}
      \cinferenceRule[invind|ind]{inductive invariant}
      {\linferenceRule[formula]
        {\psi\limply\dbox{\alpha}{\psi}}
        {\psi\limply\dbox{\prepeat{\alpha}}{\psi}}
      }{} %
    \end{calculus}%
  \end{calculuscollections}
  \caption{Differential game logic axiomatization for hybrid games without differential games}
  \label{fig:dGL}
\end{figure}%

\section{Proof of Isaacs Equations} \label{app:Isaacs}

For the sake of completeness, this section shows a proof of \rref{thm:differential-game-Isaacs} that is simplified compared to its original version \cite{DBLP:journals/indianam/EvansSouganidis84}.
The proof of \rref{thm:differential-game-Isaacs} uses two lemmas.

\newcommand{\gder}[2][f]{\nabla_{#1}{#2}}%
\begin{lemma} \label{lem:U-v-difference}
Let \(v\in\continuouss[1]{(0,T)\times\reals^n}{}\).
The upper value $U$ of \rref{eq:diffgame} satisfies for any \(0\leq\eta\leq\eta+\sigma\leq T\):
\begin{multline*}
U(\eta,\xi)-v(\eta,\xi)
=
\sup_{\alpha\in\stratD{\eta}} \inf_{z\in\controlA{\eta}} \Big(
\int_\eta^{\eta+\sigma} 
\runcost{h(s,x(s),\alpha(z)(s),z(s))
+} \gder[f]{v}(s) ds
+ U(\eta+\sigma,x(\eta+\sigma)) - v(\eta+\sigma,x(\eta+\sigma))\Big)
\label{eq:U-v-difference}
\end{multline*}
where \(x(\zeta) = \response[f]{\zeta}{\iget[state]{\I}}{\alpha(z)}{z}\) is the response of \rref{eq:diffgame} for $\alpha(z)$ and $z$ and
\[\gder[f]{v}(s) \mdefeq v_t(s,x(s)) + f(s,x(s),\alpha(z)(s),z(s))\stimes D_x v(s,x(s))\]
\end{lemma}
\begin{proof}
The result follows from the dynamic programming optimality condition \rref{eq:U-dynamic-programming} with step size $\sigma$.
Recall
\begin{equation}
U(\eta,\xi) = 
\sup_{\alpha\in\stratD{\eta}} \inf_{z\in\controlA{\eta}} \runcost{\int_\eta^{\eta+\sigma} h(s,x(s),\alpha(z)(s),z(s))ds +} U(\eta+\sigma,x(\eta+\sigma))
\tag{\ref{eq:U-dynamic-programming}$^*$}
\end{equation}
using the fundamental theorem of calculus \cite[Thm.\,9.23]{Walter:Ana2} (since $v$ is differentiable on the open interval $(\eta,\eta+\sigma)$ and continuous on the closed interval $[\eta,\eta+\sigma]$):
\[
v(\eta+\sigma,x(\eta+\sigma)) - v(\eta,\xi)
= \int_\eta^{\eta+\sigma} \D[t]{v(t,x(t))}(s) ds
 = 
\int_\eta^{\eta+\sigma} \gder[f]{v}(s) ds
\label{eq:fundamental-calculus-v}
\qedhere
\]
\end{proof}

\begin{lemma}[{\cite[Lem.\,4.3]{DBLP:journals/indianam/EvansSouganidis84}}] \label{lem:4.3}
Let \(v\in\continuouss[1]{(0,T)\times\reals^n}{}\).
  \begin{equation}
  v_t(\eta,\xi) + H^+(\eta,\xi,Dv(\eta,\xi)) \leq-\theta<0
  \label{eq:non-upper-Isaac-max}
  \end{equation}
\[
\text{implies}~
\text{for all sufficiently small}~ {\sigma}~{\mexists{z\in\controlA{\eta}}{\mforall{\alpha\in\stratD{\eta}{}}}}
\int_{\eta}^{\eta+\sigma} \runcost{h(s,x(s),\alpha(z)(s),z(s)) +} \gder[f]{v}(s) ds \leq -\frac{\sigma\theta}{2}
\]
  \begin{equation}
  v_t(\eta,\xi) + H^+(\eta,\xi,Dv(\eta,\xi)) \geq\theta>0
  \label{eq:non-upper-Isaac-min}
  \end{equation}
\[
\text{implies}~
\text{for all sufficiently small}~{\sigma}~{\mexists{\alpha\in\stratD{\eta}{\mforall{z\in\controlA{\eta}{}}}}}
\int_{\eta}^{\eta+\sigma} \runcost{h(s,x(s),\alpha(z)(s),z(s)) + }\gder[f]{v}(s) ds \geq \frac{\sigma\theta}{2}
\]
\end{lemma}
\begin{proof}
\newcommand{\upperIsaacVisc}[4]{\Lambda(#1,#2,#3,#4)}%
To simplify the assumptions, abbreviate
\[
\upperIsaacVisc{t}{x}{y}{z}
\mdefeq
v_t(t,x) + f(t,x,y,z)\stimes D_x v(t,x) \runcost{+ h(t,x,y,z)}
\]
First prove the first inequality.
By the definition of $H^+$, \rref{eq:non-upper-Isaac-max} is
\[\min_{z\in Z} \max_{y\in Y} \upperIsaacVisc{\eta}{\xi}{y}{z} \leq -\theta<0\]
which implies for some $z^*\in Z$ that
\[\max_{y\in Y} \upperIsaacVisc{\eta}{\xi}{y}{z^*} \leq -\theta<0\]
Since $\upperIsaacVisc{t}{x}{y}{z}$ is (uniformly) continuous
\[\max_{y\in Y} \upperIsaacVisc{s}{x(s)}{y}{z^*} \leq -\frac{\theta}{2}\]
for \(s\in[\eta,\eta+\sigma]\) with a sufficiently small $\sigma$ when $x(\argholder)$ is the response of \rref{eq:diffgame} for any $y,z$ with initial condition \(x(\eta)=\xi\).
Consequently, for the constant control \(z(\argholder)\mdefeq z^*\), any \(\alpha\in\stratD{\eta}\) gives
\[
\upperIsaacVisc{s}{x(s)}{\alpha(z)(s)}{z(s)} \leq -\frac{\theta}{2}
\]

Now, prove the second inequality \rref{eq:non-upper-Isaac-min}, which is
\[
\min_{z\in Z}\max_{y\in Y} \upperIsaacVisc{\eta}{\xi}{y}{z} \geq \theta>0
\]
which implies that, for each $z\in Z$, there is a \(y\in Y\) such that
\[
\upperIsaacVisc{\eta}{\xi}{y}{z} \geq \theta
\]
Since $\upperIsaacVisc{t}{x}{y}{z}$ is (uniformly) continuous
\begin{equation}
\upperIsaacVisc{\eta}{\xi}{y}{\zeta} \geq \frac{3\theta}{4}
\label{eq:ball-bound}
\end{equation}
for all $\zeta\in Z$ in an open ball around $z$.
Since this holds for all $z\in Z$ and $Z$ is compact, there is a finite open covering of $Z$ with open balls $B_i$ within which \rref{eq:ball-bound} holds for all $\zeta\in B\cap Z$.
Pick a function \(c:Z\to Y\) such that \(c(z)\) is the center of the closest ball $B_i$ to $z$ (breaking ties arbitrarily).
Then, for all $z\in Z$:
\[
\upperIsaacVisc{\eta}{\xi}{c(z)}{z} \geq \frac{3\theta}{4}
\]
Since $\upperIsaacVisc{t}{x}{y}{z}$ is (uniformly) continuous,
\begin{equation}
\upperIsaacVisc{\eta}{\xi}{c(z)}{z} \geq \frac{\theta}{2}
\label{eq:here}
\end{equation}
for \(s\in[\eta,\eta+\sigma]\) with a sufficiently small $\sigma$ when $x(\argholder)$ is the response of \rref{eq:diffgame} for any $y,z$ with initial condition \(x(\eta)=\xi\).
Construct \(\alpha\in\stratD{\eta}\) for \(z\in\controlA{\eta}\) as \(\alpha(z)(s) \mdefeq c(z(s))\) for all $s$.
Then \rref{eq:here} implies
\[
\upperIsaacVisc{s}{x(s)}{\alpha(z)(s)}{z(s)} \geq \frac{\theta}{2}
\]
for all \(s\in[\eta,\eta+\sigma]\), which implies the desired inequality by integration from $\eta$ to $\eta+\sigma$.
\end{proof}

\begin{proof}[of \rref{thm:differential-game-Isaacs}]
$U$ can be shown to be the viscosity solution of the upper Isaacs equation.
The proof for $V$ is dual.
First, $U$ satisfies the terminal condition \(U(T,\xi)=g(x(T))=g(\xi)\) for all $\xi\in\reals^n$.

Second, $U$ is shown to be a subsolution of \rref{eq:upper-Isaacs}, that is
\[
\tau + H^+(\eta,\xi,p) \geq 0
\quad
\mforallr{(\tau,p)\in\superdiff{U(\eta,\xi)}}
\]
By \rref{lem:subsuperdifferentials}, this is equivalent to showing
  \begin{equation}
  v_t(\eta,\xi) + H^+(\eta,\xi,Dv(\eta,\xi)) \geq0
  \label{eq:upper-Isaac-max}
  \end{equation}
for all \(v\in\continuouss[1]{(0,T)\times\reals^n}{}\)
that make $U-v$ attain a local maximum at \((\eta,\xi)\in(0,T)\times\reals^n\), i.e.\
\begin{equation}
U(\eta,\xi)-v(\eta,\xi) \geq U(\eta+\sigma,x(\eta+\sigma)) - v(\eta+\sigma,x(\eta+\sigma))
\label{eq:U-v-local-max}
\end{equation}
for sufficiently small $\sigma$ and $x(\argholder)$ solving \rref{eq:diffgame} with initial condition \(x(\eta)=\xi\).
Otherwise, if \rref{eq:upper-Isaac-max} were not the case, then there would be a $\theta$ such that
  \begin{equation}
  v_t(\eta,\xi) + H^+(\eta,\xi,Dv(\eta,\xi)) \leq-\theta<0
  \tag{\ref{eq:non-upper-Isaac-max}$^*$}
  \end{equation}
By \rref{lem:U-v-difference}, \rref{eq:U-v-local-max} implies for any \(0\leq\eta\leq\eta+\sigma\leq T\)
\begin{equation}
\sup_{\alpha\in\stratD{\eta}} \inf_{z\in\controlA{\eta}} \int_\eta^{\eta+\sigma} 
\runcost{h(s,x(s),\alpha(z)(s),z(s))
+} \gder[f]{v}(s) ds
 \geq 0
 \label{eq:U-v-local-max-dynamic-programming3}
\end{equation}
By \rref{lem:4.3}, \rref{eq:non-upper-Isaac-max} implies
for all sufficiently small $\sigma$ \({\mexists{z\in\controlA{\eta}}{\mforall{\alpha\in\stratD{\eta}{}}}}\)
\[
\int_{\eta}^{\eta+\sigma} \runcost{h(s,x(s),\alpha(z)(s),z(s)) +} \gder[f]{v}(s) ds \leq - \frac{\sigma\theta}{2}
\]
This choice of $z$ (that is even common for all $\alpha$) implies in particular
\begin{equation}
\sup_{\alpha\in\stratD{\eta}} \inf_{z\in\controlA{\eta}} 
\int_{\eta}^{\eta+\sigma} \runcost{h(s,x(s),\alpha(z)(s),z(s)) +} \gder[f]{v}(s) ds  \leq
- \frac{\sigma\theta}{2}
\label{eq:swaps}
\end{equation}
Equation \rref{eq:U-v-local-max-dynamic-programming3} contradicts \rref{eq:swaps} and, thus, refutes \rref{eq:non-upper-Isaac-max} and proves \rref{eq:upper-Isaac-max}.

Third, $U$ is shown to be a supersolution of \rref{eq:upper-Isaacs}, that is
\[
\tau + H^+(\eta,\xi,p) \leq 0
\quad
\mforallr{(\tau,p)\in\subdiff{U(\eta,\xi)}}
\]
By \rref{lem:subsuperdifferentials}, this is equivalent to showing
  \begin{equation}
  v_t(\eta,\xi) + H^+(\eta,\xi,Dv(\eta,\xi)) \leq0
  \label{eq:upper-Isaac-min}
  \end{equation}
for all \(v\in\continuouss[1]{(0,T)\times\reals^n}{}\)
that make $U-v$ attain a local minimum at \((\eta,\xi)\in(0,T)\times\reals^n\), i.e.\
\begin{equation}
U(\eta,\xi)-v(\eta,\xi) \leq U(\eta+\sigma,x(\eta+\sigma)) - v(\eta+\sigma,x(\eta+\sigma))
\label{eq:U-v-local-min}
\end{equation}
for sufficiently small $\sigma$ and $x(\argholder)$ solving \rref{eq:diffgame} with initial condition \(x(\eta)=\xi\).
Otherwise, if \rref{eq:upper-Isaac-min} were not the case, then there would be a $\theta$ such that
  \begin{equation}
  v_t(\eta,\xi) + H^+(\eta,\xi,Dv(\eta,\xi)) \geq\theta>0
  \tag{\ref{eq:non-upper-Isaac-min}$^*$}
  \end{equation}
By \rref{lem:U-v-difference}, \rref{eq:U-v-local-min} implies for any \(0\leq\eta\leq\eta+\sigma\leq T\)
\begin{equation}
\sup_{\alpha\in\stratD{\eta}} \inf_{z\in\controlA{\eta}} \int_\eta^{\eta+\sigma} 
\runcost{h(s,x(s),\alpha(z)(s),z(s))
+} \gder[f]{v}(s) ds
 \leq 0
 \label{eq:U-v-local-min-dynamic-programming3}
\end{equation}
By \rref{lem:4.3}, \rref{eq:non-upper-Isaac-min} implies for all sufficiently small $\sigma$
\({\mexists{\alpha\in\stratD{\eta}{\mforall{z\in\controlA{\eta}{}}}}}\)
\[
\int_{\eta}^{\eta+\sigma} \runcost{h(s,x(s),\alpha(z)(s),z(s)) +} \gder[f]{v}(s) ds \geq \frac{\sigma\theta}{2}
\]
This choice of $\alpha$ demonstrates the lower bound
\begin{equation}
\sup_{\alpha\in\stratD{\eta}} \inf_{z\in\controlA{\eta}} 
\int_{\eta}^{\eta+\sigma} \runcost{h(s,x(s),\alpha(z)(s),z(s)) +} \gder[f]{v}(s) ds \geq \frac{\sigma\theta}{2}
\label{eq:thing}
\end{equation}
Equation \rref{eq:U-v-local-min-dynamic-programming3} contradicts \rref{eq:thing} and, thus, refutes \rref{eq:non-upper-Isaac-min} and proves \rref{eq:upper-Isaac-min}.
\end{proof}

\section*{Acknowledgment}
The author appreciates helpful discussions with Max Niedermeier, Bruce Krogh, Sarah Loos, and especially Noel Walkington's advice.

This material is based upon work supported by the National Science Foundation under NSF CAREER Award CNS-1054246.

Any opinions, findings, and conclusions or recommendations expressed in this publication are those of the author(s) and do not necessarily reflect the views of the National Science Foundation.

\bibliographystyle{plainurl}\bibliography{platzer,bibliography}

\end{document}